\documentclass[11pt]{article}%
\usepackage{amsfonts}
\usepackage{amsmath}
\usepackage{fullpage}
\usepackage{amssymb}
\usepackage{graphicx}
\usepackage{hyperref}%
\providecommand{\U}[1]{\protect\rule{.1in}{.1in}}
\newtheorem{theorem}{Theorem}

\newtheorem{conjecture}[theorem]{Conjecture}
\newtheorem{corollary}[theorem]{Corollary}

\newtheorem{definition}[theorem]{Definition}

\newtheorem{lemma}[theorem]{Lemma}

\newtheorem{proposition}[theorem]{Proposition}
\newtheorem{remark}[theorem]{Remark}

\numberwithin{equation}{section}
\newenvironment{proof}[1][Proof]{\noindent\textbf{#1.} }{\ \rule{0.5em}{0.5em}}
\begin{document}

\title{\textbf{R\'{e}nyi generalizations of the conditional quantum mutual
information}}
\author{Mario Berta\thanks{Institute for Quantum Information and Matter, California
Institute of Technology, Pasadena, California 91125, USA}
\and Kaushik P. Seshadreesan\thanks{Hearne Institute for Theoretical Physics,
Department of Physics and Astronomy, Louisiana State University, Baton Rouge,
Louisiana 70803, USA}
\and Mark M. Wilde\footnotemark[2] \thanks{Center for Computation and Technology,
Louisiana State University, Baton Rouge, Louisiana 70803, USA}}
\maketitle

\begin{abstract}
The conditional quantum mutual information $I(A;B|C)$ of a tripartite state
$\rho_{ABC}$\ is an information quantity which lies at the center of many
problems in quantum information theory. Three of its main properties are that
it is non-negative for any tripartite state, that it decreases under local
operations applied to systems $A$ and $B$, and that it obeys the duality
relation $I(A;B|C)=I(A;B|D)$ for a four-party pure state on systems $ABCD$.
The conditional mutual information also underlies the squashed entanglement,
an entanglement measure that satisfies all of the axioms desired for an
entanglement measure. As such, it has been an open question to find R\'{e}nyi
generalizations of the conditional mutual information, that would allow for a
deeper understanding of the original quantity and find applications beyond the
traditional memoryless setting of quantum information theory. The present
paper addresses this question, by defining different $\alpha$-R\'{e}nyi
generalizations $I_{\alpha}(A;B|C)$ of the conditional mutual information,
some of which we can prove converge to the conditional mutual information in
the limit $\alpha\rightarrow1$. Furthermore, we prove that many of these
generalizations satisfy non-negativity, duality, and monotonicity with respect
to local operations on one of the systems $A$ or $B$ (with it being left as an
open question to prove that monotoniticity holds with respect to local
operations on both systems). The quantities defined here should find
applications in quantum information theory and perhaps even in other areas of
physics, but we leave this for future work. We also state a conjecture
regarding the monotonicity of the R\'{e}nyi conditional mutual informations
defined here with respect to the R\'{e}nyi parameter~$\alpha$. We prove that
this conjecture is true in some special cases and when $\alpha$ is in a
neighborhood of one.

\end{abstract}


\section{Introduction}

How much correlation do two parties have from the perspective of a third? This
kind of correlation is what the conditional quantum mutual information
quantifies. Indeed, let $\rho_{ABC}$ be a density operator corresponding to a
quantum state shared between three parties, say, Alice, Bob, and Charlie. Then
the conditional quantum mutual information is defined as%
\begin{equation}
I(A;B|C)_{\rho}\equiv H(AC)_{\rho}+H(BC)_{\rho}-H(C)_{\rho}-H(ABC)_{\rho},
\label{eq:CMI-definition}%
\end{equation}
where $H(F)_{\sigma}\equiv-$Tr$\{\sigma_{F}\log\sigma_{F}\}$ is the von
Neumann entropy of a state $\sigma_{F}$ on system $F$ and we unambiguously let
$\rho_{C}\equiv\ $Tr$_{AB}\{\rho_{ABC}\}$ denote the reduced density operator
on system $C$, for example. Refs.~\cite{DY08,YD09} provided a compelling
operational interpretation of the conditional quantum mutual information in
terms of the quantum state redistribution protocol:\ given many independent
copies of a four-party pure state $\psi_{ADBC}$, with a sender possessing the
$D$ and $B$ systems, a receiver possessing the $C$ systems, and the sender and
receiver sharing noiseless entanglement before communication begins, the
optimal rate of quantum communication necessary to transfer the $B$~systems to
the receiver is given by $\tfrac{1}{2}I(A;B|C)_{\psi}$.

It is a nontrivial fact, known as strong subadditivity of quantum entropy
\cite{PhysRevLett.30.434,LR73}, that the conditional quantum mutual
information of any tripartite quantum state is non-negative. This can be
viewed as a general constraint imposed on the marginal entropy values of
arbitrary tripartite quantum states. Strong subadditivity also implies that
the conditional mutual information can never increase under local quantum
operations performed on the systems $A$ and $B$ \cite{CW04}, so that
$I(A;B|C)_{\rho}$ is a sensible measure of the correlations present between
systems $A$ and $B$, from the perspective of $C$. That is, the following
inequality holds%
\begin{equation}
I(A;B|C)_{\rho}\geq I(A^{\prime};B^{\prime}|C)_{\omega},
\end{equation}
where $\omega_{A^{\prime}B^{\prime}C}\equiv\left(  \mathcal{N}_{A\rightarrow
A^{\prime}}\otimes\mathcal{M}_{B\rightarrow B^{\prime}}\right)  \left(
\rho_{ABC}\right)  $ with $\mathcal{N}_{A\rightarrow A^{\prime}}$ and
$\mathcal{M}_{B\rightarrow B^{\prime}}$ arbitrary local quantum operations
performed on the input systems $A$ and $B$, leading to output systems
$A^{\prime}$ and $B^{\prime}$, respectively. Inequalities like these are
extremely useful in applications, with nearly all coding theorems in quantum
information theory invoking the strong subadditivity inequality in their proofs.

One of the most fruitful avenues of research in quantum information theory has
been the program of generalizing entropies beyond those that are linear
combinations of the von Neumann entropy
\cite{P86,RennerThesis,D09,T12,WR12,MDSFT13,WWY13,DKFRR13}. Not only is this
interesting from a theoretical perspective, but more importantly, these
generalizations have found application in operational settings in which there
is no assumption of many independent and identically distributed
(i.i.d.)~systems, so that the law of large numbers does not come into play. In
particular, the family of R\'{e}nyi entropies has proved to possess a wide
variety of applications in these non-i.i.d.~settings. More recently,
researchers have shown that nearly all of the known information quantities
being employed in the non-i.i.d.~setting are special cases of a R\'{e}nyi
family of quantum entropies \cite{MDSFT13,AD13}.

However, in spite of this aforementioned progress, it has been a vexing open
question to determine a R\'{e}nyi generalization of the conditional quantum
mutual information that can be useful in applications. On the one hand, a
potential R\'{e}nyi generalization of the conditional mutual information of a
tripartite state $\rho_{ABC}$ consists of simply taking a linear combination
of R\'{e}nyi entropies. For example, in analogy with the definition in
(\ref{eq:CMI-definition}), one could define a R\'{e}nyi generalization of the
conditional mutual information as follows:%
\begin{equation}
I_{\alpha}^{\prime}(A;B|C)_{\rho}\equiv H_{\alpha}(AC)_{\rho}+H_{\alpha
}(BC)_{\rho}-H_{\alpha}(C)_{\rho}-H_{\alpha}(ABC)_{\rho},
\label{eq:simple-renyi-generalization}%
\end{equation}
where $H_{\alpha}(F)_{\sigma}\equiv\left[  1-\alpha\right]  ^{-1}\log
$Tr$\left\{  \sigma_{F}^{\alpha}\right\}  $ is the R\'{e}nyi entropy of a
state $\sigma_{F}$ on system $F$, with parameter $\alpha\in(0,1)\cup
(1,\infty)$ (with the R\'{e}nyi entropy being defined for $\alpha\in\left\{
0,1,\infty\right\}  $ in the limit as $\alpha$ approaches $0$, $1$, and
$\infty$, respectively). Although this quantity is non-negative in some very
special cases \cite{AGS12}, in general, $I_{\alpha}^{\prime}(A;B|C)_{\rho}$
can be negative, and in fact there are some simple examples of states for
which this occurs. Furthermore, the results of \cite{LMW13}\ imply that there
are generally no linear inequality constraints on the marginal R\'{e}nyi
entropies of a multiparty quantum state other than non-negativity when
$\alpha\in\left(  0,1\right)  \cup\left(  1,\infty\right)  $. This implies
that monotonicity under local quantum operations generally does not hold for
$I_{\alpha}^{\prime}(A;B|C)_{\rho}$, and \cite{LMW13} provides many examples
of four-party states $\rho_{ABCD}$ such that $I_{\alpha}^{\prime
}(A;BD|C)_{\rho}<I_{\alpha}^{\prime}(A;B|C)_{\rho}$. For these reasons, we
feel that formulas like that in (\ref{eq:simple-renyi-generalization}) should
not be considered as R\'{e}nyi generalizations of the conditional quantum
mutual information, given that non-negativity and monotonicity under local
operations are two of the basic properties of the conditional quantum mutual
information which are consistently employed in applications. However, one could
certainly argue that the case $\alpha=2$ is useful for the class of Gaussian
quantum states, as done in \cite{AGS12}.

On the other hand, the standard approach for generalizing information
quantities such as entropy, conditional entropy, and mutual information beyond
the von Neumann setting begins with the realization that these quantities can
be written in terms of the Umegaki relative entropy $D(\rho\Vert\sigma)$
\cite{U62}:%
\begin{align}
H(A)_{\rho} &  =-D(\rho_{A}\Vert I_{A}),\\
H(A|B)_{\rho} &  \equiv H(AB)_{\rho}-H(B)_{\rho}=-\min_{\sigma_{B}}D(\rho
_{AB}\Vert I_{A}\otimes\sigma_{B}),\label{eq:optimized-CE}\\
I(A;B)_{\rho} &  \equiv H(A)_{\rho}+H(B)_{\rho}-H(AB)_{\rho}=\min_{\sigma_{B}%
}D(\rho_{AB}\Vert\rho_{A}\otimes\sigma_{B}),\label{eq:optimized-MI}%
\end{align}
where%
\begin{equation}
D(\rho\Vert\sigma)\equiv\left\{
\begin{array}
[c]{cc}%
\left[  \text{Tr}\left\{  \rho\right\}  \right]  ^{-1}\left[  \text{Tr}%
\left\{  \rho\log\rho\right\}  -\text{Tr}\left\{  \rho\log\sigma\right\}
\right]   & \text{if supp}\left(  \rho\right)  \subseteq\text{supp}\left(
\sigma\right)  \\
+\infty & \text{otherwise}%
\end{array}
\right.  .\label{eq:vn-rel-ent}%
\end{equation}
Note that the unique optimum $\sigma_{B}$ in (\ref{eq:optimized-CE}) and
(\ref{eq:optimized-MI}) turns out to be the reduced density operator$~\rho
_{B}$. The R\'{e}nyi relative entropy of order $\alpha\in\lbrack
0,1)\cup(1,\infty)$ is defined as \cite{P86}%
\begin{equation}
D_{\alpha}(\rho\Vert\sigma)\equiv\left\{
\begin{array}
[c]{cc}%
\frac{1}{\alpha-1}\log\text{Tr}\left\{  \left[  \text{Tr}\left\{
\rho\right\}  \right]  ^{-1}\rho^{\alpha}\sigma^{1-\alpha}\right\}   &
\text{if supp}\left(  \rho\right)  \subseteq\text{supp}\left(  \sigma\right)
\text{ or (}\alpha\in\lbrack0,1)\text{ and }\rho\not \perp \sigma\text{)}\\
+\infty & \text{otherwise}%
\end{array}
\right.  ,\label{eq:Renyi-rel-ent}%
\end{equation}
with the support conditions established in \cite{TCR09}. Using this quantity,
one can easily define R\'{e}nyi generalizations of entropy, conditional
entropy, and mutual information in analogy with the above formulations:%
\begin{align}
H_{\alpha}(A)_{\rho} &  =-D_{\alpha}(\rho_{A}\Vert I_{A}),\\
H_{\alpha}(A|B)_{\rho} &  \equiv-\min_{\sigma_{B}}D_{\alpha}(\rho_{AB}\Vert
I_{A}\otimes\sigma_{B}),\label{eq:cond_alpha}\\
I_{\alpha}(A;B)_{\rho} &  \equiv\min_{\sigma_{B}}D_{\alpha}(\rho_{AB}\Vert
\rho_{A}\otimes\sigma_{B}).\label{eq:alpha_mutual}%
\end{align}
Since the R\'{e}nyi relative entropy obeys monotonicity under quantum
operations for $\alpha\in\lbrack0,1)\cup(1,2]$ \cite{P86}, in the sense that
$D_{\alpha}(\rho\Vert\sigma)\geq D_{\alpha}(\mathcal{N}\left(  \rho\right)
\Vert\mathcal{N}\left(  \sigma\right)  )$ for a quantum operation
$\mathcal{N}$, the above generalizations have proven useful in several
applications (see \cite{KW09,MH11,T12} and references therein).

\section{Overview of results}

The main purpose of the present paper is to develop R\'{e}nyi generalizations
of the conditional quantum mutual information that satisfy the aforementioned
properties of non-negativity, monotonicity under local quantum operations, and
duality. We come close to achieving this goal by showing that non-negativity,
duality, and monotonicity under local operations on one of the systems $A$ or
$B$ hold for many of our R\'{e}nyi generalizations. Numerical evidence has not
falsified monotonicity under local operations holding for both systems $A$ and
$B$, but it remains an open question to determine if this holds for both
systems $A$ and $B$. Nevertheless, we think the quantities defined here should
be useful in applications in quantum information theory, and they might even
find use in other areas of physics
\cite{HT07,CC09,H10,K12,K13,HHM13,K13thesis}.

After establishing some notation and recalling definitions in the next
section, our starting point is in Section~\ref{sec:vn-cmi}, where we recall
that the conditional quantum mutual information of a tripartite state
$\rho_{ABC}$ can be written in terms of the relative entropy as follows (see
\textquotedblleft\textit{Proof of (1.5)}\textquotedblright\ in \cite{LR73}):%
\begin{equation}
I\left(  A;B|C\right)  _{\rho}=D\left(  \rho_{ABC}\Vert\exp\left\{  \log
\rho_{AC}+\log\rho_{BC}-\log\rho_{C}\right\}  \right)  .
\label{eq:CMI-rel-ent-formula}%
\end{equation}
We then recall the following generalized Lie-Trotter product formula from
\cite{S85}, with the particular form below being inspired from developments in
\cite{LP08}%
\begin{equation}
\exp\left\{  \log\rho_{AC}+\log\rho_{BC}-\log\rho_{C}\right\}  =\lim
_{\alpha\rightarrow1}\left[  \rho_{AC}^{\left(  1-\alpha\right)  /2}\rho
_{C}^{\left(  \alpha-1\right)  /2}\rho_{BC}^{1-\alpha}\rho_{C}^{\left(
\alpha-1\right)  /2}\rho_{AC}^{\left(  1-\alpha\right)  /2}\right]
^{1/\left(  1-\alpha\right)  }, \label{eq:CMI-LTrott}%
\end{equation}
where we assume that the operators $\rho_{AC}$, $\rho_{BC}$, and $\rho_{C}$
are invertible. The relation above suggests a number of R\'{e}nyi
generalizations of the relative entropy formulation in
\eqref{eq:CMI-rel-ent-formula}, one of which is%
\begin{multline}
D_{\alpha}\left(  \rho_{ABC}\middle\Vert\left[  \rho_{AC}^{\left(
1-\alpha\right)  /2}\rho_{C}^{\left(  \alpha-1\right)  /2}\rho_{BC}^{1-\alpha
}\rho_{C}^{\left(  \alpha-1\right)  /2}\rho_{AC}^{\left(  1-\alpha\right)
/2}\right]  ^{1/\left(  1-\alpha\right)  }\right)
\label{eq:renyi-gen-marginals}\\
=\frac{1}{\alpha-1}\log\text{Tr}\left\{  \rho_{ABC}^{\alpha}\rho_{AC}^{\left(
1-\alpha\right)  /2}\rho_{C}^{\left(  \alpha-1\right)  /2}\rho_{BC}^{1-\alpha
}\rho_{C}^{\left(  \alpha-1\right)  /2}\rho_{AC}^{\left(  1-\alpha\right)
/2}\right\}  .
\end{multline}
We prove that several of these R\'{e}nyi conditional mutual informations are
non-negative for $\alpha\in\lbrack0,1)\cup(1,2]$ and obey monotonicity under
local quantum operations on one of the systems $A$ or $B$ in the same range of
$\alpha$ (with the proof following from the Lieb concavity theorem~\cite{L73}
and the Ando convexity theorem~\cite{A79}).Our proof for monotonicity under
local operations depends on operator orderings in the particular R\'{e}nyi
generalization of the conditional mutual information. For example, we can show
that monotonicity under operations on the $B$ system holds for the quantity
defined in (\ref{eq:renyi-gen-marginals}), due to the fact that $\rho_{BC}$ is
\textquotedblleft placed in the middle.\textquotedblright We also consider
several limiting cases, the most important of which is the limit as
$\alpha\rightarrow1$. We prove that some of the $\alpha$-R\'{e}nyi conditional
mutual informations converge to $I\left(  A;B|C\right)  _{\rho}$ in this
limit. Note that classical and quantum quantities related to these have been
explored in prior work \cite{B12,E13}.

The sandwiched R\'{e}nyi relative entropy \cite{MDSFT13,WWY13}\ is another
variant of the R\'{e}nyi relative entropy which has found a number of
applications recently in the context of strong converse theorems
\cite{WWY13,MO13,GW13,CMW14,TWW14}. It is defined for $\alpha\in
(0,1)\cup(1,\infty)$ as follows:%
\begin{equation}
\widetilde{D}_{\alpha}\left(  \rho\Vert\sigma\right)  \equiv\left\{
\begin{array}
[c]{cc}%
\frac{1}{\alpha-1}\log\left[  \left[  \text{Tr}\left\{  \rho\right\}  \right]
^{-1}\text{Tr}\left\{  \left(  \sigma^{\left(  1-\alpha\right)  /2\alpha}%
\rho\sigma^{\left(  1-\alpha\right)  /2\alpha}\right)  ^{\alpha}\right\}
\right]  &
\begin{array}
[c]{c}%
\text{if supp}\left(  \rho\right)  \subseteq\text{supp}\left(  \sigma\right)
\text{ or}\\
\text{(}\alpha\in(0,1)\text{ and }\rho\not \perp \sigma\text{)}%
\end{array}
\\
+\infty & \text{otherwise}%
\end{array}
\right.  . \label{eq:def-sandwiched}%
\end{equation}
In Section~\ref{sec:sandwiched-Renyi-CMI}, we use this sandwiched R\'{e}nyi
relative entropy to establish a number of sandwiched R\'{e}nyi generalizations
of the conditional mutual information, one of which is%
\begin{multline}
\widetilde{D}_{\alpha}\left(  \rho_{ABC}\middle\Vert\left[  \rho_{AC}^{\left(
1-\alpha\right)  /2\alpha}\rho_{C}^{\left(  \alpha-1\right)  /2\alpha}%
\rho_{BC}^{\left(  1-\alpha\right)  /\alpha}\rho_{C}^{\left(  \alpha-1\right)
/2\alpha}\rho_{AC}^{\left(  1-\alpha\right)  /2\alpha}\right]  ^{\alpha
/\left(  1-\alpha\right)  }\right) \label{eq:renyi-sandwiched-marginals}\\
=\frac{1}{\alpha-1}\log\text{Tr}\left\{  \left(  \rho_{ABC}^{1/2}\rho
_{AC}^{\left(  1-\alpha\right)  /2\alpha}\rho_{C}^{\left(  \alpha-1\right)
/2\alpha}\rho_{BC}^{\left(  1-\alpha\right)  /\alpha}\rho_{C}^{\left(
\alpha-1\right)  /2\alpha}\rho_{AC}^{\left(  1-\alpha\right)  /2\alpha}%
\rho_{ABC}^{1/2}\right)  ^{\alpha}\right\}  ,
\end{multline}
where the equality follows from the fact that%
\begin{equation}
\text{Tr}\left\{  \left(  \sigma^{\left(  1-\alpha\right)  /2\alpha}\rho
\sigma^{\left(  1-\alpha\right)  /2\alpha}\right)  ^{\alpha}\right\}
=\text{Tr}\left\{  \left(  \rho^{1/2}\sigma^{\left(  1-\alpha\right)  /\alpha
}\rho^{1/2}\right)  ^{\alpha}\right\}  .
\end{equation}
Although both R\'{e}nyi generalizations of the conditional mutual information
feature \textquotedblleft operator sandwiches,\textquotedblright\ we give this
particular generalization the epithet \textquotedblleft
sandwiched\textquotedblright\ because it is derived from the sandwiched
R\'{e}nyi relative entropy. We prove that several of these sandwiched
R\'{e}nyi conditional mutual informations are non-negative for all $\alpha
\in\lbrack1/2,1)\cup(1,\infty)$ and that they are monotone under local quantum
operations on one of the systems $A$ or $B$ for the same range of $\alpha$
(with the proof following from recent work in \cite{Hiai20131568} and
\cite{FL13}). We can prove that some of them converge to $I\left(
A;B|C\right)  _{\rho}$ in the limit as $\alpha\rightarrow1$, and there are
other interesting quantities to consider for $\alpha=1/2$ or $\alpha=\infty$,
leading to a min- and max- version of conditional mutual information,
respectively. There are certainly other possible definitions for R\'{e}nyi
conditional mutual information that one could consider and we discuss these in
the conclusion.

One of the most curious non-classical properties of the conditional quantum
mutual information is that it obeys a duality relation \cite{DY08,YD09}. That
is, for a four-party pure state $\psi_{ABCD}$, the following equality holds%
\begin{equation}
I\left(  A;B|C\right)  _{\psi}=I\left(  A;B|D\right)  _{\psi}.
\end{equation}
In Section~\ref{sec:duality}, we prove that some variants of the R\'{e}nyi
conditional mutual information obey duality relations analogous to the above one.

A well known property of both the traditional and the sandwiched R\'{e}nyi
relative entropy is that they are monotone non-decreasing in $\alpha$. That
is, for $0\leq\alpha\leq\beta$, we have the following inequalities
\cite{TCR09,MDSFT13}:%
\begin{equation}
D_{\alpha}\left(  \rho\Vert\sigma\right)  \leq D_{\beta}\left(  \rho
\Vert\sigma\right)  ,\ \ \ \ \ \ \ \widetilde{D}_{\alpha}\left(  \rho
\Vert\sigma\right)  \leq\widetilde{D}_{\beta}\left(  \rho\Vert\sigma\right)  .
\label{eq:monotonicity-in-alpha-intro}%
\end{equation}
Section~\ref{sec:monotonicity-in-alpha}\ states an open conjecture, that the
R\'{e}nyi generalizations of the conditional mutual information obey a similar
monotonicity. We prove that this conjecture is true in some special cases, we
prove that it is true when $\alpha$ is in a neighborhood of one, and numerical
evidence indicates that it is true in general.
%
We finally conclude in Section~\ref{sec:conclusion} with a summary of our
results and a discussion of directions for future research.

\section{Notation and definitions}

\label{sec:notation}\textbf{Norms, states, channels, and measurements.} Let
$\mathcal{B}\left(  \mathcal{H}\right)  $ denote the algebra of bounded linear
operators acting on a Hilbert space $\mathcal{H}$. We restrict ourselves to
finite-dimensional Hilbert spaces throughout this paper. For $\alpha\geq1$, we
define the $\alpha$-norm of an operator $X$ as%
\begin{equation}
\left\Vert X\right\Vert _{\alpha}\equiv\text{Tr}\{(\sqrt{X^{\dag}X})^{\alpha
}\}^{1/\alpha}, \label{eq:a-norm}%
\end{equation}
and we use the same notation even for the case $\alpha\in(0,1)$, when it is
not a norm. Let $\mathcal{B}\left(  \mathcal{H}\right)  _{+}$ denote the
subset of positive semi-definite operators, and let $\mathcal{B}\left(
\mathcal{H}\right)  _{++}$ denote the subset of positive definite operators.
We also write $X\geq0$ if $X\in\mathcal{B}\left(  \mathcal{H}\right)  _{+}$
and $X>0$ if $X\in\mathcal{B}\left(  \mathcal{H}\right)  _{++}$. An\ operator
$\rho$ is in the set $\mathcal{S}\left(  \mathcal{H}\right)  $\ of density
operators (or states) if $\rho\in\mathcal{B}\left(  \mathcal{H}\right)  _{+}$
and Tr$\left\{  \rho\right\}  =1$, and an\ operator $\rho$ is in the set
$\mathcal{S}(\mathcal{H})_{++}$\ of strictly positive definite density
operators if $\rho\in\mathcal{B}\left(  \mathcal{H}\right)  _{++}$ and
Tr$\left\{  \rho\right\}  =1$. The tensor product of two Hilbert spaces
$\mathcal{H}_{A}$ and $\mathcal{H}_{B}$ is denoted by $\mathcal{H}_{A}%
\otimes\mathcal{H}_{B}$ or $\mathcal{H}_{AB}$.\ Given a multipartite density
operator $\rho_{AB}\in\mathcal{S}(\mathcal{H}_{A}\otimes\mathcal{H}_{B})$, we
unambiguously write $\rho_{A}=\ $Tr$_{B}\left\{  \rho_{AB}\right\}  $ for the
reduced density operator on system $A$. We use $\rho_{AB}$, $\sigma_{AB}$,
$\tau_{AB}$, $\omega_{AB}$, etc.~to denote general density operators in
$\mathcal{S}(\mathcal{H}_{A}\otimes\mathcal{H}_{B})$, while $\psi_{AB}$,
$\varphi_{AB}$, $\phi_{AB}$, etc.~denote rank-one density operators (pure
states) in $\mathcal{S}(\mathcal{H}_{A}\otimes\mathcal{H}_{B})$ (with it
implicit, clear from the context, and the above convention implying that
$\psi_{A}$, $\varphi_{A}$, $\phi_{A}$ may be mixed if $\psi_{AB}$,
$\varphi_{AB}$, $\phi_{AB}$ are pure). In expressions like that in
(\ref{eq:renyi-gen-marginals}) and (\ref{eq:renyi-sandwiched-marginals}), an
identity operator is implicit when not written (and should be clear from the
context), so that, for example, the expression $\rho_{BC}^{1-\alpha}$ in
(\ref{eq:renyi-gen-marginals}) should be interpreted as $I_{A}\otimes\rho
_{BC}^{1-\alpha}$.

The trace distance between two quantum states $\rho,\sigma\in\mathcal{S}%
\left(  \mathcal{H}\right)  $\ is equal to $\left\Vert \rho-\sigma\right\Vert
_{1}$. It has a direct operational interpretation in terms of the
distinguishability of these states. That is, if $\rho$ or $\sigma$ are
prepared with equal probability and the task is to distinguish them via some
quantum measurement, then the optimal success probability in doing so is equal
to $\left(  1+\left\Vert \rho-\sigma\right\Vert _{1}/2\right)  /2$. Throughout
the paper, for technical convenience and simplicity, some of our statements
apply only to states in $\mathcal{S}\left(  \mathcal{H}\right)  _{++}$. This
might seem restrictive, but in the following sense, it is physically
reasonable. Given any state $\omega\in\mathcal{S}\left(  \mathcal{H}\right)
\backslash\mathcal{S}\left(  \mathcal{H}\right)  _{++}$, there is a state
$\omega\left(  \xi\right)  =\left(  1-\xi\right)  \omega+\xi I/\dim\left(
\mathcal{H}\right)  $ for a constant $\xi>0$, so that $\omega\left(
\xi\right)  \in\mathcal{S}\left(  \mathcal{H}\right)  _{++}$ and $\left\Vert
\omega-\omega\left(  \xi\right)  \right\Vert _{1}\leq2\xi$. Thus, the bias in
distinguishing $\omega$ from $\omega\left(  \xi\right)  $ is no more than
$\xi/2$, so that $\omega\left(  \xi\right)  $ can \textquotedblleft
mask\textquotedblright\ as$~\omega$.

Throughout this paper, we take the usual convention that $f\left(  A\right)
=\sum_{i}f\left(  a_{i}\right)  \left\vert i\right\rangle \left\langle
i\right\vert $ when given a function $f$ and a Hermitian operator $A$ with
spectral decomposition $A=\sum_{i}a_{i}\left\vert i\right\rangle \left\langle
i\right\vert $. So this means that $A^{-1}$ is interpreted as a generalized
inverse, so that $A^{-1}=\sum_{i:a_{i} \neq0}a_{i}^{-1}\left\vert
i\right\rangle \left\langle i\right\vert $, $\log\left(  A\right)
=\sum_{i:a_{i} >0 }\log\left(  a_{i}\right)  \left\vert i\right\rangle
\left\langle i\right\vert $, $\exp\left(  A\right)  =\sum_{i}\exp\left(
a_{i}\right)  \left\vert i\right\rangle \left\langle i\right\vert $, etc.
Throughout the paper, we interpret $\log$ as the natural logarithm. The above
convention for $f\left(  A\right)  $ leads to the convention that $A^{0}$
denotes the projection onto the support of $A$, i.e., $A^{0}=\sum_{i:a_{i}
\neq0}\left\vert i\right\rangle \left\langle i\right\vert $. We employ the
shorthand supp$(A)$ and ker$(A)$ to refer to the support and kernel of an
operator $A$, respectively.

A linear map $\mathcal{N}_{A\rightarrow B}:\mathcal{B}\left(  \mathcal{H}%
_{A}\right)  \rightarrow\mathcal{B}\left(  \mathcal{H}_{B}\right)  $\ is
positive if $\mathcal{N}_{A\rightarrow B}\left(  \sigma_{A}\right)
\in\mathcal{B}\left(  \mathcal{H}_{B}\right)  _{+}$ whenever $\sigma_{A}%
\in\mathcal{B}\left(  \mathcal{H}_{A}\right)  _{+}$. A linear map
$\mathcal{N}_{A\rightarrow B}:\mathcal{B}\left(  \mathcal{H}_{A}\right)
\rightarrow\mathcal{B}\left(  \mathcal{H}_{B}\right)  $\ is strictly positive
if $\mathcal{N}_{A\rightarrow B}\left(  \sigma_{A}\right)  \in\mathcal{B}%
\left(  \mathcal{H}_{B}\right)  _{++}$ whenever $\sigma_{A}\in\mathcal{B}%
\left(  \mathcal{H}_{A}\right)  _{++}$. Let id$_{A}$ denote the identity map
acting on a system $A$. A linear map $\mathcal{N}_{A\rightarrow B}$ is
completely positive if the map id$_{R}\otimes\mathcal{N}_{A\rightarrow B}$ is
positive for a reference system $R$ of arbitrary size. A linear map
$\mathcal{N}_{A\rightarrow B}$ is trace-preserving if Tr$\left\{
\mathcal{N}_{A\rightarrow B}\left(  \tau_{A}\right)  \right\}  =\ $Tr$\left\{
\tau_{A}\right\}  $ for all input operators $\tau_{A}\in\mathcal{B}\left(
\mathcal{H}_{A}\right)  $. If a linear map is completely positive and
trace-preserving (CPTP), we say that it is a quantum channel or quantum
operation. A positive operator-valued measure (POVM) is a set $\left\{
\Lambda^{m}\right\}  $ of positive semi-definite operators such that $\sum
_{m}\Lambda^{m}=I$.

\bigskip

\textbf{Relative entropies.} We defined the relative entropy $D(P\Vert Q)$
between $P,Q\in\mathcal{B}\left(  \mathcal{H}\right)  _{+}$ in
(\ref{eq:vn-rel-ent}), with $P\neq0$. The definition is consistent with the
following limit, so that%
\begin{equation}
\lim_{\xi\searrow0}\left[  \text{Tr}\left\{  P\right\}  \right]
^{-1}\text{Tr}\left\{  P\left[  \log P-\log\left(  Q+\xi I\right)  \right]
\right\}  =D(P\Vert Q), \label{eq:rel-ent-consistency}%
\end{equation}
where $I$ is the identity operator acting on $\mathcal{H}$. The statement in
(\ref{eq:rel-ent-consistency}) follows because the quantity%
\begin{equation}
\lim_{\xi\searrow0}\text{Tr}\left\{  P\log\left(  Q+\xi I\right)  \right\}
\label{eq:limiting-quantity}%
\end{equation}
is finite and equal to Tr$\left\{  P\log Q\right\}  $ if supp$\left(
P\right)  \subseteq\ $supp$\left(  Q\right)  $. Otherwise,
(\ref{eq:limiting-quantity}) is infinite. The relative entropy $D(P\Vert
Q)$\ is non-negative if Tr$\left\{  P\right\}  \geq$ Tr$\left\{  Q\right\}  $,
a result known as Klein's inequality \cite{LR68}. Thus, for density operators
$\rho$ and $\sigma$, the relative entropy is non-negative, and furthermore, it
is equal to zero if and only if $\rho=\sigma$.

We defined the R\'{e}nyi relative entropy in (\ref{eq:Renyi-rel-ent}). This
definition is consistent with the following limit, so that for $\alpha
\in\lbrack0,1)\cup(1,\infty)$%
\begin{equation}
\lim_{\xi\searrow0}\frac{1}{\alpha-1}\log\text{Tr}\left\{  \left[
\text{Tr}\left\{  P\right\}  \right]  ^{-1}P^{\alpha}\left(  Q+\xi I\right)
^{1-\alpha}\right\}  =D_{\alpha}(P\Vert Q),\label{eq:supp-renyi-1}%
\end{equation}
as can be checked by a proof similar to \cite[Lemma 13]{MDSFT13}. The quantity
obeys the following monotonicity inequality for all $\alpha\in\lbrack
0,1)\cup(1,2]$:%
\begin{equation}
D_{\alpha}\left(  P\Vert Q\right)  \geq D_{\alpha}\left(  \mathcal{N}\left(
P\right)  \Vert\mathcal{N}\left(  Q\right)  \right)  ,
\end{equation}
where $P,Q\in\mathcal{B}\left(  \mathcal{H}\right)  _{+}$ and $\mathcal{N}$ is
a CPTP\ map \cite{P86}. Thus, by applying this, we find that $D_{\alpha
}\left(  P\Vert Q\right)  $ is non-negative for all $\alpha\in\lbrack
0,1)\cup(1,2]$ whenever Tr$\left\{  P\right\}  \geq$ Tr$\left\{  Q\right\}  $,
so that it is always non-negative for density operators $\rho$ and $\sigma$.
Furthermore, it is equal to zero if and only if $\rho=\sigma$.

We also defined the sandwiched R\'{e}nyi relative entropy in
(\ref{eq:def-sandwiched}). Similar to the above quantities, the definition is
consistent with the following limit, so that%
\begin{equation}
\lim_{\xi\searrow0}\frac{1}{\alpha-1}\log\left[  \left[  \text{Tr}\left\{
P\right\}  \right]  ^{-1}\text{Tr}\left\{  \left[  \left(  Q+\xi I\right)
^{\left(  1-\alpha\right)  /2\alpha}P\left(  Q+\xi I\right)  ^{\left(
1-\alpha\right)  /2\alpha}\right]  ^{\alpha}\right\}  \right]  =\widetilde
{D}_{\alpha}(P\Vert Q),\label{eq:supp-renyi-2}%
\end{equation}
as proved in \cite[Lemma 13]{MDSFT13}. Whenever supp$\left(  P\right)
\subseteq$ supp$\left(  Q\right)  $ or ($\alpha\in\left(  0,1\right)  $ and
$P\not \perp Q$), it admits the following alternate forms:%
\begin{align}
\widetilde{D}_{\alpha}\left(  P\Vert Q\right)   &  \equiv\frac{1}{\alpha
-1}\log\left[  \left[  \text{Tr}\left\{  P\right\}  \right]  ^{-1}%
\text{Tr}\left\{  \left(  Q^{\left(  1-\alpha\right)  /2\alpha}PQ^{\left(
1-\alpha\right)  /2\alpha}\right)  ^{\alpha}\right\}  \right]  \\
&  =\frac{\alpha}{\alpha-1}\log\left\Vert Q^{\left(  1-\alpha\right)
/2\alpha}PQ^{\left(  1-\alpha\right)  /2\alpha}\right\Vert _{\alpha}-\frac
{1}{\alpha-1}\log\text{Tr}\left\{  P\right\}  \\
&  =\frac{\alpha}{\alpha-1}\log\left\Vert P^{1/2}Q^{\left(  1-\alpha\right)
/\alpha}P^{1/2}\right\Vert _{\alpha}-\frac{1}{\alpha-1}\log\text{Tr}\left\{
P\right\}  .
\end{align}
It obeys the following monotonicity inequality for all $\alpha\in
\lbrack1/2,1)\cup(1,\infty)$:%
\begin{equation}
\widetilde{D}_{\alpha}\left(  P\Vert Q\right)  \geq\widetilde{D}_{\alpha
}\left(  \mathcal{N}\left(  P\right)  \Vert\mathcal{N}\left(  Q\right)
\right)  ,
\end{equation}
where $P,Q\in\mathcal{B}\left(  \mathcal{H}\right)  _{+}$ and $\mathcal{N}$ is
a CPTP\ map \cite{FL13} (see also \cite{B13,MO13,WWY13,MDSFT13}\ for other
proofs of this for more limited ranges of $\alpha$). Thus, by applying this,
we find that $\widetilde{D}_{\alpha}\left(  P\Vert Q\right)  $ is non-negative
for all $\alpha\in\lbrack1/2,1)\cup(1,\infty)$ whenever Tr$\left\{  P\right\}
\geq$ Tr$\left\{  Q\right\}  $, so that it is always non-negative for density
operators $\rho$ and $\sigma$. Furthermore, it is equal to zero if and only if
$\rho=\sigma$.

\section{Conditional quantum mutual information based on von Neumann entropy}

\label{sec:vn-cmi}In this section, we prove that the conditional quantum
mutual information has many seemingly different representations in terms of a
relative-entropy-like quantity (however all of them being equal). This paves
the way for designing different R\'{e}nyi generalizations of the conditional
quantum mutual information. Furthermore, we give a conceptually different
proof of the fact that the conditional quantum mutual information $I\left(
A;B|C\right)  $\ is monotone under local quantum operations on systems $A$
and~$B$. This alternate proof will be the basis for similar proofs when we
consider R\'{e}nyi generalizations in Sections~\ref{sec:Renyi-CMI} and
\ref{sec:sandwiched-Renyi-CMI}. Finally, we discuss how representing $I\left(
A;B|C\right)  $ as we do in Proposition~\ref{thm:CMI-formulations} allows for
a straightforward comparison of it with the minimum relative entropy
\textquotedblleft distance\textquotedblright\ to quantum Markov states, a
quantity originally considered in \cite{ILW08}.

\subsection{Various formulations of the conditional quantum mutual
information}

One of the core quantities that we consider in this paper is the following
function of four density operators $\rho_{ABC}\in\mathcal{S}\left(
\mathcal{H}_{ABC}\right)  $, $\tau_{AC}\in\mathcal{S}\left(  \mathcal{H}%
_{AC}\right)  $, $\theta_{BC}\in\mathcal{S}\left(  \mathcal{H}_{BC}\right)  $,
and $\omega_{C}\in\mathcal{S}\left(  \mathcal{H}_{C}\right)  $:%
\begin{equation}
\Delta\left(  \rho_{ABC},\tau_{AC},\theta_{BC},\omega_{C}\right)
\equiv\text{Tr}\left\{  \rho_{ABC}\left[  \log\rho_{ABC}-\log\tau_{AC}%
-\log\theta_{BC}+\log\omega_{C}\right]  \right\}  , \label{eq:Delta-quantity}%
\end{equation}
where logarithms of density operators are understood in the usual sense
described in Section~\ref{sec:notation}. Let$~I_{ABC}$ denote the identity
operator acting on $\mathcal{H}_{ABC}$. A sufficient condition for%
\begin{equation}
\lim_{\xi\searrow0}\Delta\left(  \rho_{ABC},\tau_{AC}+\xi I_{ABC},\theta
_{BC}+\xi I_{ABC},\omega_{C}+\xi I_{ABC}\right)  \label{eq:limit-xi-Delta}%
\end{equation}
to be finite and equal to (\ref{eq:Delta-quantity})\ is that%
\begin{equation}
\text{supp}\left(  \rho_{ABC}\right)  \subseteq\text{supp}\left(  \tau
_{AC}\right)  \text{, supp}\left(  \theta_{BC}\right)  \text{, supp}\left(
\omega_{C}\right)  , \label{eq:support-condition}%
\end{equation}
for the same reason given after (\ref{eq:rel-ent-consistency}). When comparing
with supp$\left(  \rho_{ABC}\right)  $, it is implicit throughout this paper
that supp$\left(  \tau_{AC}\right)  \equiv\ $supp$\left(  I_{B}\otimes
\tau_{AC}\right)  $, supp$\left(  \theta_{BC}\right)  \equiv\ $supp$\left(
I_{A}\otimes\theta_{BC}\right)  $, and supp$\left(  \omega_{C}\right)
\equiv\ $supp$\left(  I_{AB}\otimes\omega_{C}\right)  $. The condition in
\eqref{eq:support-condition} is equivalent to supp$\left(  \rho_{ABC}\right)
$ being in the intersection of the supports of $\tau_{AC}$, $\theta_{BC}$, and
$\omega_{C}$. Note that there are more general support conditions which lead
to a finite value for (\ref{eq:limit-xi-Delta}), but for simplicity, we focus
exclusively on the above support condition. If the support condition in
\eqref{eq:support-condition} holds, then by inspection we can write%
\begin{equation}
\Delta\left(  \rho_{ABC},\tau_{AC},\theta_{BC},\omega_{C}\right)  =D\left(
\rho_{ABC}\Vert\exp\left\{  \log\tau_{AC}+\log\theta_{BC}-\log\omega
_{C}\right\}  \right)  . \label{eq:general-delta}%
\end{equation}
Furthermore, observe that%
\begin{equation}
\lim_{\xi\searrow0}\Delta\left(  \rho_{ABC},\rho_{AC}+\xi I_{ABC},\rho
_{BC}+\xi I_{ABC},\rho_{C}+\xi I_{ABC}\right)  \label{eq:limit-xi-to-0-Delta}%
\end{equation}
is finite and equal to (\ref{eq:Delta-quantity})\ because the support
condition in \eqref{eq:support-condition} holds when choosing $\tau_{AC}$,
$\theta_{BC}$, and $\omega_{C}$ as the marginals of $\rho_{ABC}$ (see, e.g.,
\cite[Lemma~B.4.1]{RennerThesis}).

\begin{lemma}
\label{lem:rel-ent-CMI-helper}Let $\rho_{ABC}\in\mathcal{S}\left(
\mathcal{H}_{ABC}\right)  $, $\tau_{AC}\in\mathcal{S}\left(  \mathcal{H}%
_{AC}\right)  $, $\theta_{BC}\in\mathcal{S}\left(  \mathcal{H}_{BC}\right)  $,
and $\omega_{C}\in\mathcal{S}\left(  \mathcal{H}_{C}\right)  $ and suppose
that the support condition in \eqref{eq:support-condition}\ holds. Then%
\begin{equation}
\Delta\left(  \rho_{ABC},\tau_{AC},\theta_{BC},\omega_{C}\right)  =I\left(
A;B|C\right)  _{\rho}+D\left(  \rho_{AC}\Vert\tau_{AC}\right)  +D\left(
\rho_{BC}\Vert\theta_{BC}\right)  -D\left(  \rho_{C}\Vert\omega_{C}\right)  .
\label{eq:rel-ent-rewrite-1}%
\end{equation}

\end{lemma}

\begin{proof}
This follows simply by adding to and subtracting from $\Delta\left(
\rho_{ABC},\tau_{AC},\theta_{BC},\omega_{C}\right)  $ each of Tr$\left\{
\rho_{ABC}\log\rho_{AC}\right\}  $, Tr$\left\{  \rho_{ABC}\log\rho
_{BC}\right\}  $, and Tr$\left\{  \rho_{ABC}\log\rho_{C}\right\}  $. We then
apply the definitions of $I\left(  A;B|C\right)  _{\rho}$, $D\left(  \rho
_{AC}\Vert\tau_{AC}\right)  $, $D\left(  \rho_{BC}\Vert\theta_{BC}\right)  $,
and $D\left(  \rho_{C}\Vert\omega_{C}\right)  $.
\end{proof}

For the mutual information, there are four seemingly different ways of writing
it as a relative entropy \cite{CBR14}. However, for the conditional mutual
information, there are many ways of doing so, as summarized in the following
proposition. The significance of Proposition~\ref{thm:CMI-formulations} is
that it paves the way for designing many different R\'{e}nyi generalizations
of the conditional mutual information.

\begin{proposition}
\label{thm:CMI-formulations}Let $\rho_{ABC}\in\mathcal{S}\left(
\mathcal{H}_{ABC}\right)  $. Then%
\begin{align}
I\left(  A;B|C\right)  _{\rho}  &  =\Delta\left(  \rho_{ABC},\rho_{AC}%
,\rho_{BC},\rho_{C}\right)  =\inf_{\tau_{AC}}\Delta\left(  \rho_{ABC}%
,\tau_{AC},\rho_{BC},\rho_{C}\right) \\
&  =\inf_{\theta_{BC}}\Delta\left(  \rho_{ABC},\rho_{AC},\theta_{BC},\rho
_{C}\right)  =\sup_{\omega_{C}}\Delta\left(  \rho_{ABC},\rho_{AC},\rho
_{BC},\omega_{C}\right) \\
&  =\inf_{\tau_{AC}}\Delta\left(  \rho_{ABC},\tau_{AC},\rho_{BC},\tau
_{C}\right)  =\inf_{\tau_{AC}}\sup_{\omega_{C}}\Delta\left(  \rho_{ABC}%
,\tau_{AC},\rho_{BC},\omega_{C}\right) \\
&  =\inf_{\theta_{BC}}\Delta\left(  \rho_{ABC},\rho_{AC},\theta_{BC}%
,\theta_{C}\right)  =\inf_{\theta_{BC}}\sup_{\omega_{C}}\Delta\left(
\rho_{ABC},\rho_{AC},\theta_{BC},\omega_{C}\right) \\
&  =\inf_{\sigma_{ABC}}\Delta\left(  \rho_{ABC},\sigma_{AC},\sigma_{BC}%
,\rho_{C}\right)  =\inf_{\tau_{AC},\theta_{BC}}\Delta\left(  \rho_{ABC}%
,\tau_{AC},\theta_{BC},\rho_{C}\right) \\
&  =\inf_{\sigma_{ABC}}\Delta\left(  \rho_{ABC},\sigma_{AC},\sigma_{BC}%
,\sigma_{C}\right)  =\inf_{\tau_{AC},\theta_{BC}}\Delta\left(  \rho_{ABC}%
,\tau_{AC},\theta_{BC},\tau_{C}\right) \label{eq:CMI-rel-ent}\\
&  =\inf_{\tau_{AC},\theta_{BC}}\Delta\left(  \rho_{ABC},\tau_{AC},\theta
_{BC},\theta_{C}\right)  =\inf_{\sigma_{ABC}}\sup_{\omega_{C}}\Delta\left(
\rho_{ABC},\sigma_{AC},\sigma_{BC},\omega_{C}\right) \\
&  =\inf_{\tau_{AC},\theta_{BC}}\sup_{\omega_{C}}\Delta\left(  \rho_{ABC}%
,\tau_{AC},\theta_{BC},\omega_{C}\right)  , \label{eq:CMI-rel-ent-2}%
\end{align}
where the optimizations are over states on the indicated Hilbert spaces
obeying the support condition in \eqref{eq:support-condition} and over
$\sigma_{ABC}$ for which $\operatorname{supp}\left(  \rho_{ABC}\right)
\subseteq\operatorname{supp}\left(  \sigma_{ABC}\right)  $. The infima and
suprema can be interchanged in all of the above cases, are achieved by the
marginals of $\rho_{ABC}$, and can thus be replaced by minima and maxima.
\end{proposition}

\begin{proof}
We only prove two of these relations, noting that the rest follow from similar
ideas. We first prove (\ref{eq:CMI-rel-ent-2}). Invoking
Lemma~\ref{lem:rel-ent-CMI-helper}, we have that%
\begin{multline}
\inf_{\tau_{AC},\theta_{BC}}\sup_{\omega_{C}}\Delta\left(  \rho_{ABC}%
,\tau_{AC},\theta_{BC},\omega_{C}\right)  =I\left(  A;B|C\right)  _{\rho}\\
+\inf_{\tau_{AC}}D\left(  \rho_{AC}\Vert\tau_{AC}\right)  +\inf_{\theta_{BC}%
}D\left(  \rho_{BC}\Vert\theta_{BC}\right)  -\inf_{\omega_{C}}D\left(
\rho_{C}\Vert\omega_{C}\right)  .
\end{multline}
Invoking the fact that the relative entropy is minimized and equal to zero
when its first argument is equal to its second, we see that the right hand
side is equal to $I\left(  A;B|C\right)  _{\rho}$.

We now prove the first equality in (\ref{eq:CMI-rel-ent}). Let $\sigma_{ABC}$
denote some tripartite state for which $\operatorname{supp}\left(  \rho
_{ABC}\right)  \subseteq\operatorname{supp}\left(  \sigma_{ABC}\right)  $. By
Lemma~\ref{lem:rel-ent-CMI-helper}, we have that%
\begin{equation}
\Delta\left(  \rho_{ABC},\sigma_{AC},\sigma_{BC},\sigma_{C}\right)  =I\left(
A;B|C\right)  _{\rho}+D\left(  \rho_{AC}\Vert\sigma_{AC}\right)  +D\left(
\rho_{BC}\Vert\sigma_{BC}\right)  -D\left(  \rho_{C}\Vert\sigma_{C}\right)  .
\end{equation}
But it is known that the relative entropy is monotone under a partial trace,
so that
\begin{equation}
D\left(  \rho_{AC}\Vert\sigma_{AC}\right)  \geq D\left(  \rho_{C}\Vert
\sigma_{C}\right)  .
\end{equation}
Thus, we have that%
\begin{equation}
D\left(  \rho_{AC}\Vert\sigma_{AC}\right)  +D\left(  \rho_{BC}\Vert\sigma
_{BC}\right)  -D\left(  \rho_{C}\Vert\sigma_{C}\right)  \geq0.
\end{equation}
This implies that%
\begin{equation}
\inf_{\sigma_{ABC}}\Delta\left(  \rho_{ABC},\sigma_{AC},\sigma_{BC},\sigma
_{C}\right)  =I\left(  A;B|C\right)  _{\rho}+\inf_{\sigma_{ABC}}\left[
D\left(  \rho_{AC}\Vert\sigma_{AC}\right)  +D\left(  \rho_{BC}\Vert\sigma
_{BC}\right)  -D\left(  \rho_{C}\Vert\sigma_{C}\right)  \right]  .
\end{equation}
The three rightmost terms are non-negative (as shown above), so that we can
minimize them (to their absolute minimum of zero) by picking a state
$\sigma_{ABC}$\ such that%
\begin{equation}
\sigma_{AC}=\rho_{AC},\ \ \ \ \log\sigma_{BC}-\log\sigma_{C}=\log\rho
_{BC}-\log\rho_{C},
\end{equation}
or by symmetry, one such that%
\begin{equation}
\sigma_{BC}=\rho_{BC},\ \ \ \ \log\sigma_{AC}-\log\sigma_{C}=\log\rho
_{AC}-\log\rho_{C}.
\end{equation}
One clear choice satisfying this is $\sigma_{ABC}=\rho_{ABC}$, but there could
be others.
\end{proof}

\begin{remark}
A priori, we require infima and suprema in the above proposition because the
sets over which the optimizations occur are not compact. More explicitly,
suppose that $\rho_{ABC}=\omega_{AB}\otimes\theta_{C}$ for $\omega_{AB}%
\in\mathcal{S}\left(  \mathcal{H}_{AB}\right)  $ and $\theta_{C}\in
\mathcal{S}\left(  \mathcal{H}_{C}\right)  $. Then the sequence of states%
\begin{equation}
\omega_{AB}\left(  n\right)  \equiv\frac{1}{n}\frac{\omega_{AB}^{0}%
}{\operatorname{Tr}\left\{  \omega_{AB}^{0}\right\}  }+\left(  1-\frac{1}%
{n}\right)  \frac{I_{AB}-\omega_{AB}^{0}}{\operatorname{Tr}\left\{
I_{AB}-\omega_{AB}^{0}\right\}  },
\end{equation}
is such that $\operatorname{supp}\left(  \rho_{ABC}\right)  \subseteq
\operatorname{supp}\left(  \omega_{AB}\left(  n\right)  \right)  $ for all
$n\geq1$, but $\operatorname{supp}\left(  \rho_{ABC}\right)  \not \subseteq
\operatorname{supp}\left(  \omega_{AB}\left(  \infty\right)  \right)  $.
\end{remark}

\begin{corollary}
\label{cor:pinsker-exp-log}Let $\rho_{ABC}\in\mathcal{S}\left(  \mathcal{H}%
_{ABC}\right)  $. Then there is a Pinsker-like lower bound on the conditional
mutual information$~I\left(  A;B|C\right)  _{\rho}$:%
\begin{equation}
I\left(  A;B|C\right)  _{\rho}\geq\tfrac{1}{4}\left\Vert \rho_{ABC}%
-\exp\left\{  \log\rho_{AC}+\log\rho_{BC}-\log\rho_{C}\right\}  \right\Vert
_{1}^{2}.
\end{equation}

\end{corollary}

\begin{proof}
The corollary results from the following chain of inequalities:%
\begin{align}
&  \!\!\!\!\!\!D\left(  \rho_{ABC}\Vert\exp\left\{  \log\rho_{AC}+\log
\rho_{BC}-\log\rho_{C}\right\}  \right) \nonumber\\
&  \geq D_{1/2}\left(  \rho_{ABC}\Vert\exp\left\{  \log\rho_{AC}+\log\rho
_{BC}-\log\rho_{C}\right\}  \right) \\
&  =-2\log\text{Tr}\left\{  \sqrt{\rho_{ABC}}\sqrt{\exp\left\{  \log\rho
_{AC}+\log\rho_{BC}-\log\rho_{C}\right\}  }\right\} \\
&  \geq-2\log\left(  1-\tfrac{1}{2}\left\Vert \sqrt{\rho_{ABC}}-\sqrt
{\exp\left\{  \log\rho_{AC}+\log\rho_{BC}-\log\rho_{C}\right\}  }\right\Vert
_{2}^{2}\right) \\
&  \geq\left\Vert \sqrt{\rho_{ABC}}-\sqrt{\exp\left\{  \log\rho_{AC}+\log
\rho_{BC}-\log\rho_{C}\right\}  }\right\Vert _{2}^{2}\\
&  \geq\tfrac{1}{4}\left\Vert \rho_{ABC}-\exp\left\{  \log\rho_{AC}+\log
\rho_{BC}-\log\rho_{C}\right\}  \right\Vert _{1}^{2}.
\end{align}
The first step follows from monotonicity of the R\'{e}nyi relative entropy
with respect to the R\'{e}nyi parameter (see
(\ref{eq:monotonicity-in-alpha-intro})). The rest are from a line of reasoning
similar to that in the proofs of \cite[Theorem~2.1 and Corollary~2.2]{Z14},
which in turn follows from some of the development in \cite{CL14}.
\end{proof}

\subsection{Monotonicity of the conditional quantum mutual information under
local quantum operations}

In this section, we show that the $\Delta$ quantity in
(\ref{eq:Delta-quantity}) obeys monotonicity under tensor-product quantum
operations acting on the systems $A$ and $B$, thus establishing it as a
fundamental information measure upon which the conditional mutual information
is based. Later we also establish a R\'enyi generalization of this quantity,
which is the core quantity underlying our various R\'enyi generalizations 
of conditional mutual information.

\begin{lemma}
\label{thm:vn-rel-ent-local-mono}Let $\rho_{ABC}\in\mathcal{S}\left(
\mathcal{H}_{ABC}\right)  $, $\tau_{AC}\in\mathcal{S}\left(  \mathcal{H}%
_{AC}\right)  $, $\theta_{BC}\in\mathcal{S}\left(  \mathcal{H}_{BC}\right)  $,
and $\omega_{C}\in\mathcal{S}\left(  \mathcal{H}_{C}\right)  $ and suppose
that the support condition in \eqref{eq:support-condition}\ holds. Let
$\mathcal{N}_{A\rightarrow A^{\prime}}$ and $\mathcal{M}_{B\rightarrow
B^{\prime}}$ be CPTP\ maps acting on the systems $A$ and $B$, respectively.
Then the following monotonicity inequality holds%
\begin{equation}
\Delta\left(  \rho_{ABC},\tau_{AC},\theta_{BC},\omega_{C}\right)  \geq
\Delta\left(  \left(  \mathcal{N}_{A\rightarrow A^{\prime}}\otimes
\mathcal{M}_{B\rightarrow B^{\prime}}\right)  \left(  \rho_{ABC}\right)
,\mathcal{N}_{A\rightarrow A^{\prime}}\left(  \tau_{AC}\right)  ,\mathcal{M}%
_{B\rightarrow B^{\prime}}\left(  \theta_{BC}\right)  ,\omega_{C}\right)  .
\end{equation}

\end{lemma}

\begin{proof}
We first prove the inequality%
\begin{equation}
\Delta\left(  \rho_{ABC},\tau_{AC},\theta_{BC},\omega_{C}\right)  \geq
\Delta\left(  \mathcal{N}_{A\rightarrow A^{\prime}}\left(  \rho_{ABC}\right)
,\mathcal{N}_{A\rightarrow A^{\prime}}\left(  \tau_{AC}\right)  ,\theta
_{BC},\omega_{C}\right)  . \label{eq:vn-monotonicity-A}%
\end{equation}
To prove this, we simply expand out the terms:%
\begin{equation}
\Delta\left(  \rho_{ABC},\tau_{AC},\theta_{BC},\omega_{C}\right)  =D\left(
\rho_{ABC}\Vert\tau_{AC}\otimes I_{B}\right)  -\text{Tr}\left\{  \rho_{BC}%
\log\theta_{BC}\right\}  +\text{Tr}\left\{  \rho_{C}\log\omega_{C}\right\}  .
\end{equation}
Noting that $\operatorname{supp}\left(  \mathcal{N}_{A\rightarrow A^{\prime}%
}\left(  \rho_{ABC}\right)  \right)  \subseteq\operatorname{supp}\left(
\mathcal{N}_{A\rightarrow A^{\prime}}\left(  \tau_{AC}\right)  \right)  $ if
$\operatorname{supp}\left(  \rho_{ABC}\right)  \subseteq\operatorname{supp}%
\left(  \tau_{AC}\right)  $ (see, e.g., \cite[Lemma~B.4.2]{RennerThesis}), we
similarly have that%
\begin{multline}
\Delta\left(  \mathcal{N}_{A\rightarrow A^{\prime}}\left(  \rho_{ABC}\right)
,\mathcal{N}_{A\rightarrow A^{\prime}}\left(  \tau_{AC}\right)  ,\theta
_{BC},\omega_{C}\right)  =D\left(  \mathcal{N}_{A\rightarrow A^{\prime}%
}\left(  \rho_{ABC}\right)  \Vert\mathcal{N}_{A\rightarrow A^{\prime}}\left(
\tau_{AC}\right)  \otimes I_{B}\right) \\
-\text{Tr}\left\{  \rho_{BC}\log\theta_{BC}\right\}  +\text{Tr}\left\{
\rho_{C}\log\omega_{C}\right\}  .
\end{multline}
Then the inequality in (\ref{eq:vn-monotonicity-A}) follows from the ordinary
monotonicity of relative entropy:%
\begin{equation}
D\left(  \rho_{ABC}\Vert\tau_{AC}\otimes I_{B}\right)  \geq D\left(
\mathcal{N}_{A\rightarrow A^{\prime}}\left(  \rho_{ABC}\right)  \Vert
\mathcal{N}_{A\rightarrow A^{\prime}}\left(  \tau_{AC}\right)  \otimes
I_{B}\right)  .
\end{equation}
An essentially identical approach gives us the following inequality:%
\begin{equation}
\Delta\left(  \rho_{ABC},\tau_{AC},\theta_{BC},\omega_{C}\right)  \geq
\Delta\left(  \mathcal{M}_{B\rightarrow B^{\prime}}\left(  \rho_{ABC}\right)
,\tau_{AC},\mathcal{M}_{B\rightarrow B^{\prime}}\left(  \theta_{BC}\right)
,\omega_{C}\right)  .
\end{equation}
Combining this one with (\ref{eq:vn-monotonicity-A}) gives us the inequality
in the statement of the lemma.
\end{proof}

One of the crucial properties of the conditional quantum mutual information is
that it is monotone under CPTP\ maps acting on the systems $A$ and $B$,
respectively. That is,%
\begin{equation}
I\left(  A;B|C\right)  _{\rho}\geq I(A^{\prime};B^{\prime}|C)_{\xi},
\end{equation}
where $\xi_{A^{\prime}B^{\prime}C}\equiv\left(  \mathcal{N}_{A\rightarrow
A^{\prime}}\otimes\mathcal{M}_{B\rightarrow B^{\prime}}\right)  \left(
\rho_{ABC}\right)  $. From the statement of
Lemma~\ref{thm:vn-rel-ent-local-mono}, we can conclude with a conceptually
different proof (other than directly making use of strong subadditivity as
done in \cite[Proposition~3]{CW04}) that the conditional mutual information is
monotone under tensor-product maps acting on systems $A$ and $B$. The following theorem
is a straightforward consequence of Lemma~\ref{thm:vn-rel-ent-local-mono} and the fact that
$I\left(  A;B|C\right)  _{\rho}    =\Delta\left(  \rho_{ABC},\rho_{AC}
,\rho_{BC},\rho_{C}\right)$.

\begin{theorem}
[{\cite[Proposition~3]{CW04}}]Let $\rho_{ABC}\in\mathcal{S}\left(
\mathcal{H}_{ABC}\right)  $, $\mathcal{N}_{A\rightarrow A^{\prime}}$ and
$\mathcal{M}_{B\rightarrow B^{\prime}}$ be CPTP\ maps acting on the systems
$A$ and $B$, respectively, and $\xi_{A^{\prime}B^{\prime}C}\equiv\left(
\mathcal{N}_{A\rightarrow A^{\prime}}\otimes\mathcal{M}_{B\rightarrow
B^{\prime}}\right)  \left(  \rho_{ABC}\right)  $. Then the following
inequality holds%
\begin{equation}
I(A;B|C)_{\rho}\geq I(A^{\prime};B^{\prime}|C)_{\xi}.
\end{equation}

\end{theorem}

\subsection{Comparison with the minimum relative entropy to quantum Markov
states}

\label{sec:rel-ent-to-markov}In classical information theory, a tripartite
probability distribution $p_{A,B,C}\left(  a,b,c\right)  $ has conditional
mutual information $I\left(  A;B|C\right)  $\ equal to zero if and only if it
can be written as a Markov distribution $p_{C}\left(  c\right)  p_{A|C}\left(
a|c\right)  p_{B|C}\left(  b|c\right)  $. Equivalently, it is equal to zero if
and only if the distribution $p_{A,B,C}\left(  a,b,c\right)  $ is recoverable
after marginalizing over the random variable $A$, that is, if there exists a
classical channel $q\left(  a|c\right)  $ such that $p_{A,B,C}\left(
a,b,c\right)  =q\left(  a|c\right)  p_{B,C}\left(  b,c\right)  $. Furthermore,
the classical conditional mutual information of $p_{A,B,C}$ can be written as
the relative entropy distance between $p_{A,B,C}$ and the nearest Markov
distribution \cite[Section~II]{ILW08}.

The generalization of these ideas to quantum information theory is not so
straightforward, and we briefly review what is known from \cite{HJPW04} and
\cite{ILW08}. Our main aim in doing so is to set the stage for
establishing a R\'enyi generalization of conditional mutual information and
the subsequent discussion in Section~\ref{sec:small-CMI}.

An important class of quantum states are the quantum Markov states, introduced
in \cite{AF83} and studied for finite-dimensional tripartite states in
\cite{HJPW04}. Following \cite{HJPW04}, we define a state $\rho_{ABC}$ to be a
quantum Markov state if $I\left(  A;B|C\right)  _{\rho}=0$. Let $\mathcal{M}%
_{A-C-B}$ denote this class of states. The main result of \cite{HJPW04} is
that such a state has the following explicit form:%
\begin{equation}
\rho_{ABC}=\bigoplus\limits_{j}q\left(  j\right)  \sigma_{AC_{j}^{L}}%
\otimes\sigma_{C_{j}^{R}B},\label{eq:QMS-1}%
\end{equation}
for some probability distribution $q\left(  j\right)  $, density operators
$\{\sigma_{AC_{j}^{L}},\sigma_{C_{j}^{R}B}\}$, and a decomposition of the
Hilbert space for $C$ as
$
\mathcal{H}_{C}=\bigoplus\limits_{j}\mathcal{H}_{C_{j}^{L}}\otimes
\mathcal{H}_{C_{j}^{R}}
$.
We also know that a state $\rho_{ABC}$ is a quantum Markov state if any of the
following conditions hold \cite{Petz03,R02}:%
\begin{align}
\rho_{ABC} &  =\rho_{AC}^{1/2}\rho_{C}^{-1/2}\rho_{BC}\rho_{C}^{-1/2}\rho
_{AC}^{1/2},\label{eq:QMS-2}\\
\rho_{ABC} &  =\rho_{BC}^{1/2}\rho_{C}^{-1/2}\rho_{AC}\rho_{C}^{-1/2}\rho
_{BC}^{1/2},\label{eq:QMS-3}\\
\rho_{ABC} &  =\exp\left\{  \log\rho_{AC}+\log\rho_{BC}-\log\rho_{C}\right\}
.\label{eq:QMS-4}%
\end{align}
Interestingly, if $\rho_{C}$ is positive definite, then the map $\left(
\cdot\right)  \rightarrow\rho_{AC}^{1/2}\rho_{C}^{-1/2}\left(  \cdot\right)
\rho_{C}^{-1/2}\rho_{AC}^{1/2}$ is a quantum channel from system $C$ to $AC$,
as one can verify by observing that it is completely positive and trace
preserving. Otherwise, the map is trace non-increasing. These same statements
also obviously apply to the map $\left(  \cdot\right)  \rightarrow\rho
_{BC}^{1/2}\rho_{C}^{-1/2}\left(  \cdot\right)  \rho_{C}^{-1/2}\rho_{BC}%
^{1/2}$. See \cite{Jencova2006,JR10} for more conditions for a tripartite state to be a
quantum Markov state.

Let $M\left(  \rho_{ABC}\right)  $ denote the relative entropy
\textquotedblleft distance\textquotedblright\ to quantum Markov states
\cite{ILW08}:%
\begin{equation}
M\left(  \rho_{ABC}\right)  \equiv\inf_{\sigma_{ABC}\in\mathcal{M}_{A-C-B}%
}D\left(  \rho_{ABC}\Vert\sigma_{ABC}\right)  ,
\end{equation}
where $\mathcal{M}_{A-C-B}$ is the set of quantum Markov states defined above.
Clearly, it suffices to restrict the above infimum to the set of Markov states
$\sigma_{ABC}$\ for which $\operatorname{supp}\left(  \rho_{ABC}\right)
\subseteq\operatorname{supp}\left(  \sigma_{ABC}\right)  $. We can now easily
compare $I\left(  A;B|C\right)  $ with $M\left(  \rho_{ABC}\right)  $, as done in
\cite{ILW08}. First,
since every quantum Markov state satisfies the condition $\sigma_{ABC}%
=\exp\left\{  \log\sigma_{AC}+\log\sigma_{BC}-\log\sigma_{C}\right\}  $, we
see that this formula is equivalent to%
\begin{equation}
M\left(  \rho_{ABC}\right)  =\inf_{\sigma_{ABC}\in\mathcal{M}_{A-C-B}}D\left(
\rho_{ABC}\Vert\exp\left\{  \log\sigma_{AC}+\log\sigma_{BC}-\log\sigma
_{C}\right\}  \right)  ,
\end{equation}
from which we obtain the following inequality:%
\begin{align}
M\left(  \rho_{ABC}\right)   &  \geq\inf_{\omega_{ABC}}D\left(  \rho
_{ABC}\Vert\exp\left\{  \log\omega_{AC}+\log\omega_{BC}-\log\omega
_{C}\right\}  \right) \\
&  =\inf_{\omega_{ABC}}\Delta\left(  \rho_{ABC},\omega_{AC},\omega_{BC}%
,\omega_{C}\right) \\
&  =I\left(  A;B|C\right)  _{\rho},
\end{align}
where the infimum is over all states $\omega_{ABC}$ satisfying
$\operatorname{supp}\left(  \rho_{ABC}\right)  \subseteq\operatorname{supp}%
\left(  \omega_{ABC}\right)  $. The above inequality was already stated in
\cite[Theorem~4]{ILW08} (and with the simpler proof along the lines above
given by Jen\v cov\'a at the end of
\cite{ILW08}), but one of the main contributions of \cite{ILW08} was to show that there
are tripartite states $\omega_{ABC}$\ for which there is a strict inequality
$M\left(  \omega_{ABC}\right)  >I\left(  A;B|C\right)  _{\omega}$, and in fact
\cite[Section~VI]{ILW08} showed that the gap can be arbitrarily large.

Thus, from the results in \cite{ILW08}, we can already conclude that taking the R\'enyi relative entropy
distance to quantum Markov states will not lead to a useful
R\'enyi generalization of conditional
mutual information as one might hope. After the completion of the present paper,
we were informed that this matter was pursued independently in \cite{E14}.

\section{R\'{e}nyi conditional mutual information}

\label{sec:Renyi-CMI}In this section, we establish many R\'{e}nyi
generalizations of the conditional mutual information that bear some
properties similar to its properties. Furthermore, we can prove that some of
these generalizations converge to it in the limit as the R\'{e}nyi parameter
$\alpha\rightarrow1$. We are motivated to define a R\'{e}nyi conditional
mutual information by considering the generalized Lie-Trotter product formula
\cite{S85}:%
\begin{equation}
\exp\left\{  \log\tau_{AC}+\log\theta_{BC}-\log\omega_{C}\right\}
=\lim_{\alpha\rightarrow1}\left[  \tau_{AC}^{\left(  1-\alpha\right)
/2}\omega_{C}^{\left(  \alpha-1\right)  /2}\theta_{BC}^{1-\alpha}\omega
_{C}^{\left(  \alpha-1\right)  /2}\tau_{AC}^{\left(  1-\alpha\right)
/2}\right]  ^{1/\left(  1-\alpha\right)  }, \label{eq:lie-trotter}%
\end{equation}
where the equality holds when $\tau_{AC}\in\mathcal{S}\left(  \mathcal{H}%
_{AC}\right)  _{++}$, $\theta_{BC}\in\mathcal{S}\left(  \mathcal{H}%
_{BC}\right)  _{++}$, and $\omega_{C}\in\mathcal{S}\left(  \mathcal{H}%
_{C}\right)  _{++}$. By plugging the RHS\ above (before the limit is taken)
into the R\'{e}nyi relative entropy formula defined in (\ref{eq:Renyi-rel-ent}%
), we obtain the following expression:%
\begin{equation}
\frac{1}{\alpha-1}\log\text{Tr}\left\{  \rho_{ABC}^{\alpha}\tau_{AC}^{\left(
1-\alpha\right)  /2}\omega_{C}^{\left(  \alpha-1\right)  /2}\theta
_{BC}^{1-\alpha}\omega_{C}^{\left(  \alpha-1\right)  /2}\tau_{AC}^{\left(
1-\alpha\right)  /2}\right\}  . \label{eq:general-delta-renyi}%
\end{equation}
We can evaluate the above expression even in the case when $\tau_{AC}%
\in\mathcal{S}\left(  \mathcal{H}_{AC}\right)  $, $\theta_{BC}\in
\mathcal{S}\left(  \mathcal{H}_{BC}\right)  $, and $\omega_{C}\in
\mathcal{S}\left(  \mathcal{H}_{C}\right)  $ (considering instead the
generalized inverse mentioned in Section~\ref{sec:notation}). With this, we
consider the formula in (\ref{eq:general-delta-renyi}) to be a R\'{e}nyi
generalization of the formula in (\ref{eq:general-delta}).

The development above motivates some other core quantities that we consider in
this paper. Let $\rho_{ABC}\in\mathcal{S}\left(  \mathcal{H}_{ABC}\right)  $,
$\tau_{AC}\in\mathcal{S}\left(  \mathcal{H}_{AC}\right)  $, $\theta_{BC}%
\in\mathcal{S}\left(  \mathcal{H}_{BC}\right)  $, and $\omega_{C}%
\in\mathcal{S}\left(  \mathcal{H}_{C}\right)  $. We define the following
quantities for $\alpha\in\lbrack0,1)\cup\left(  1,\infty\right)  $:%
\begin{align}
Q_{\alpha}\left(  \rho_{ABC},\tau_{AC},\omega_{C},\theta_{BC}\right)   &
\equiv\text{Tr}\left\{  \rho_{ABC}^{\alpha}\tau_{AC}^{\left(  1-\alpha\right)
/2}\omega_{C}^{\left(  \alpha-1\right)  /2}\theta_{BC}^{1-\alpha}\omega
_{C}^{\left(  \alpha-1\right)  /2}\tau_{AC}^{\left(  1-\alpha\right)
/2}\right\}  ,\label{eq:Q_a_1}\\
\Delta_{\alpha}\left(  \rho_{ABC},\tau_{AC},\omega_{C},\theta_{BC}\right)   &
\equiv\frac{1}{\alpha-1}\log Q_{\alpha}\left(  \rho_{ABC},\tau_{AC},\omega
_{C},\theta_{BC}\right)  . \label{eq:Delta_a_1}%
\end{align}
We stress that the formula in (\ref{eq:Delta_a_1}) is to be interpreted in the
sense of generalized inverse, so that it is always finite if%
\begin{equation}
\rho_{ABC}\not \perp \left\vert \tau_{AC}^{\left(  1-\alpha\right)  /2}%
\omega_{C}^{\left(  \alpha-1\right)  /2}\theta_{BC}^{\left(  1-\alpha\right)
/2}\right\vert ^{2}. \label{eq:non-orthogonal}%
\end{equation}
The non-orthogonality condition in (\ref{eq:non-orthogonal})\ is satisfied,
e.g., if the support condition in (\ref{eq:support-condition}) holds, so that
(\ref{eq:non-orthogonal})\ is satisfied when $\tau_{AC}=\rho_{AC}$,
$\omega_{C}=\rho_{C}$, and $\theta_{BC}=\rho_{BC}$. It remains largely open to
determine support conditions under which%
\begin{equation}
\lim_{\xi\searrow0}\Delta_{\alpha}\left(  \rho_{ABC},\tau_{AC}+\xi
I_{ABC},\omega_{C}+\xi I_{ABC},\theta_{BC}+\xi I_{ABC}\right)
\label{eq:limit-xi-Delta-a-1}%
\end{equation}
is finite and equal to (\ref{eq:Delta_a_1}), with complications being due to
the fact that (\ref{eq:Q_a_1}) features the multiplication of several
non-commuting operators which can interact in non-trivial ways. We can also
consider five other different operator orderings for the last three arguments
of $Q_{\alpha}$, i.e.,%
\begin{align}
Q_{\alpha}\left(  \rho_{ABC},\theta_{BC},\omega_{C},\tau_{AC}\right)   &
\equiv\text{Tr}\left\{  \rho_{ABC}^{\alpha}\theta_{BC}^{\left(  1-\alpha
\right)  /2}\omega_{C}^{\left(  \alpha-1\right)  /2}\tau_{AC}^{1-\alpha}%
\omega_{C}^{\left(  \alpha-1\right)  /2}\theta_{BC}^{\left(  1-\alpha\right)
/2}\right\}  ,\label{eq:Q_a_2}\\
Q_{\alpha}\left(  \rho_{ABC},\omega_{C},\tau_{AC},\theta_{BC}\right)   &
\equiv\text{Tr}\left\{  \rho_{ABC}^{\alpha}\omega_{C}^{\left(  \alpha
-1\right)  /2}\tau_{AC}^{\left(  1-\alpha\right)  /2}\theta_{BC}^{1-\alpha
}\tau_{AC}^{\left(  1-\alpha\right)  /2}\omega_{C}^{\left(  \alpha-1\right)
/2}\right\}  ,\\
Q_{\alpha}\left(  \rho_{ABC},\omega_{C},\theta_{BC},\tau_{AC}\right)   &
\equiv\text{Tr}\left\{  \rho_{ABC}^{\alpha}\omega_{C}^{\left(  \alpha
-1\right)  /2}\theta_{BC}^{\left(  1-\alpha\right)  /2}\tau_{AC}^{1-\alpha
}\theta_{BC}^{\left(  1-\alpha\right)  /2}\omega_{C}^{\left(  \alpha-1\right)
/2}\right\}  ,\\
Q_{\alpha}\left(  \rho_{ABC},\tau_{AC},\theta_{BC},\omega_{C}\right)   &
\equiv\text{Tr}\left\{  \rho_{ABC}^{\alpha}\tau_{AC}^{\left(  1-\alpha\right)
/2}\theta_{BC}^{\left(  1-\alpha\right)  /2}\omega_{C}^{\alpha-1}\theta
_{BC}^{\left(  1-\alpha\right)  /2}\tau_{AC}^{\left(  1-\alpha\right)
/2}\right\}  ,\\
Q_{\alpha}\left(  \rho_{ABC},\theta_{BC},\tau_{AC},\omega_{C}\right)   &
\equiv\text{Tr}\left\{  \rho_{ABC}^{\alpha}\theta_{BC}^{\left(  1-\alpha
\right)  /2}\tau_{AC}^{\left(  1-\alpha\right)  /2}\omega_{C}^{\alpha-1}%
\tau_{AC}^{\left(  1-\alpha\right)  /2}\theta_{BC}^{\left(  1-\alpha\right)
/2}\right\}  . \label{eq:Q_a_6}%
\end{align}
In the above, we are abusing notation by always having the power $\left(
\alpha-1\right)  /2$ associated with $\omega_{C}$ and the power $\left(
1-\alpha\right)  /2$ associated with $\tau_{AC}$ and $\theta_{BC}$, but we
take the convention that the different $Q_{\alpha}$ quantities are uniquely
identified by the operator ordering of its last three arguments. These
different $Q_{\alpha}$ functions lead to different $\Delta_{\alpha}$
quantities, again uniquely identified by the operator ordering of the last
three arguments.

We can then use the above observations, the observation in
Proposition~\ref{thm:CMI-formulations}, and the definition of the R\'{e}nyi
relative entropy to define R\'{e}nyi generalizations of the conditional mutual
information. There are many definitions that we could take for a R\'{e}nyi
conditional mutual information by using the different optimizations summarized
in Proposition~\ref{thm:CMI-formulations} and the different orderings of
operators as suggested above.

In spite of the many possibilities suggested above, we choose to define the
\textit{R\'{e}nyi conditional mutual information} as the following quantity
because it obeys some additional properties (beyond those satisfied by many of
the above generalizations)\ which we would expect to hold for a R\'{e}nyi
generalization of the conditional mutual information.

\begin{definition}
\label{def:Renyi-CMI}Let $\rho_{ABC}\in\mathcal{S}\left(  \mathcal{H}%
_{ABC}\right)  $. The R\'{e}nyi conditional mutual information of $\rho_{ABC}$
is defined for $\alpha\in\lbrack0,1)\cup(1,\infty)$ as%
\begin{equation}
I_{\alpha}\left(  A;B|C\right)  _{\rho}\equiv\inf_{\sigma_{BC}}\Delta_{\alpha
}\left(  \rho_{ABC},\rho_{AC},\rho_{C},\sigma_{BC}\right)  ,
\label{eq:renyi-cmi}%
\end{equation}
where the optimization is over density operators $\sigma_{BC}$ such that
$\operatorname{supp}\left(  \rho_{ABC}\right)  \subseteq\operatorname{supp}%
\left(  \sigma_{BC}\right)  $.
\end{definition}

Note that unlike the conditional mutual information, this definition is not
symmetric with respect to $A$ and $B$. Thus one might also call it the
R\'{e}nyi information that $B$ has about $A$ from the perspective of $C$. Note
also that, for trivial $C$, the definition reduces to the usual definition of
R\'{e}nyi mutual information in~\eqref{eq:alpha_mutual}.

One advantage of the above definition is that we can identify an explicit form
for the minimizing $\sigma_{BC}$ and thus for $I_{\alpha}\left(  A;B|C\right)
_{\rho}$, as captured by the following proposition:

\begin{proposition}
\label{prop:sibson}Let $\rho_{ABC}\in\mathcal{S}\left(  \mathcal{H}%
_{ABC}\right)  $. The R\'{e}nyi conditional mutual information of $\rho_{ABC}$
has the following explicit form for $\alpha\in\left(  0,1\right)  \cup\left(
1,\infty\right)  $:%
\begin{equation}
I_{\alpha}\left(  A;B|C\right)  _{\rho}=\frac{\alpha}{\alpha-1}\log
\operatorname{Tr}\left\{  \left(  \rho_{C}^{\left(  \alpha-1\right)
/2}\operatorname{Tr}_{A}\left\{  \rho_{AC}^{\left(  1-\alpha\right)  /2}%
\rho_{ABC}^{\alpha}\rho_{AC}^{\left(  1-\alpha\right)  /2}\right\}  \rho
_{C}^{\left(  \alpha-1\right)  /2}\right)  ^{1/\alpha}\right\}  .
\end{equation}
This follows because the infimum in \eqref{eq:renyi-cmi}\ can be replaced by a
minimum and the minimum $\sigma_{BC}$\ is unique with an explicit form.
\end{proposition}

A proof of Proposition~\ref{prop:sibson} appears in Appendix~\ref{sec:sibson}.

\subsection{Limit of the R\'{e}nyi conditional mutual information as
$\alpha\rightarrow1$}

In this section, we consider the limit of the $\Delta_{\alpha}$ quantity as
the R\'{e}nyi parameter $\alpha\rightarrow1$. This allows us to prove that
some variations of the R\'{e}nyi conditional mutual information converge to
the conditional mutual information in the limit as $\alpha\rightarrow1$.

\begin{theorem}
\label{thm:conv-vN}Let $\rho_{ABC}\in\mathcal{S}\left(  \mathcal{H}%
_{ABC}\right)  $, $\tau_{AC}\in\mathcal{S}\left(  \mathcal{H}_{AC}\right)  $,
$\theta_{BC}\in\mathcal{S}\left(  \mathcal{H}_{BC}\right)  $, and $\omega
_{C}\in\mathcal{S}\left(  \mathcal{H}_{C}\right)  $ and suppose that the
support condition in \eqref{eq:support-condition}\ holds. Then%
\begin{equation}
\lim_{\alpha\rightarrow1}\Delta_{\alpha}\left(  \rho_{ABC},\tau_{AC}%
,\omega_{C},\theta_{BC}\right)  =\Delta\left(  \rho_{ABC},\tau_{AC},\omega
_{C},\theta_{BC}\right)  . \label{eq:delta-alpha-to-delta}%
\end{equation}
The same limiting relation holds for the other $\Delta_{\alpha}$ quantities
defined from \eqref{eq:Q_a_2}-\eqref{eq:Q_a_6}.
\end{theorem}

\begin{proof}
We will consider L'H\^{o}pital's rule in order to evaluate the limit of
$\Delta_{\alpha}$ as $\alpha\rightarrow1$, due to the presence of the
denominator term $\alpha-1$ in $\Delta_{\alpha}$. To this end, we compute the
following derivative with respect to $\alpha$%
\begin{multline}
\frac{d}{d\alpha}Q_{\alpha}\left(  \rho_{ABC},\tau_{AC},\omega_{C},\theta
_{BC}\right)  =\text{Tr}\left\{  \left(  \log\rho_{ABC}\right)  \rho
_{ABC}^{\alpha}\tau_{AC}^{\left(  1-\alpha\right)  /2}\omega_{C}^{\left(
\alpha-1\right)  /2}\theta_{BC}^{1-\alpha}\omega_{C}^{\left(  \alpha-1\right)
/2}\tau_{AC}^{\left(  1-\alpha\right)  /2}\right\} \\
-\frac{1}{2}\text{Tr}\left\{  \rho_{ABC}^{\alpha}\left(  \log\tau_{AC}\right)
\tau_{AC}^{\left(  1-\alpha\right)  /2}\omega_{C}^{\left(  \alpha-1\right)
/2}\theta_{BC}^{1-\alpha}\omega_{C}^{\left(  \alpha-1\right)  /2}\tau
_{AC}^{\left(  1-\alpha\right)  /2}\right\} \\
+\frac{1}{2}\text{Tr}\left\{  \rho_{ABC}^{\alpha}\tau_{AC}^{\left(
1-\alpha\right)  /2}\left(  \log\omega_{C}\right)  \omega_{C}^{\left(
\alpha-1\right)  /2}\theta_{BC}^{1-\alpha}\omega_{C}^{\left(  \alpha-1\right)
/2}\tau_{AC}^{\left(  1-\alpha\right)  /2}\right\} \\
-\text{Tr}\left\{  \rho_{ABC}^{\alpha}\tau_{AC}^{\left(  1-\alpha\right)
/2}\omega_{C}^{\left(  \alpha-1\right)  /2}\left(  \log\theta_{BC}\right)
\theta_{BC}^{1-\alpha}\omega_{C}^{\left(  \alpha-1\right)  /2}\tau
_{AC}^{\left(  1-\alpha\right)  /2}\right\} \\
+\frac{1}{2}\text{Tr}\left\{  \rho_{ABC}^{\alpha}\tau_{AC}^{\left(
1-\alpha\right)  /2}\omega_{C}^{\left(  \alpha-1\right)  /2}\theta
_{BC}^{1-\alpha}\left(  \log\omega_{C}\right)  \omega_{C}^{\left(
\alpha-1\right)  /2}\tau_{AC}^{\left(  1-\alpha\right)  /2}\right\} \\
-\frac{1}{2}\text{Tr}\left\{  \rho_{ABC}^{\alpha}\tau_{AC}^{\left(
1-\alpha\right)  /2}\omega_{C}^{\left(  \alpha-1\right)  /2}\theta
_{BC}^{1-\alpha}\omega_{C}^{\left(  \alpha-1\right)  /2}\left(  \log\tau
_{AC}\right)  \tau_{AC}^{\left(  1-\alpha\right)  /2}\right\}  .
\end{multline}
Thus, the function $Q_{\alpha}\left(  \rho_{ABC},\tau_{AC},\omega_{C}%
,\theta_{BC}\right)  $ is differentiable for $\alpha\in\left(  0,\infty
\right)  $. Applying L'H\^{o}pital's rule, we consider%
\begin{equation}
\lim_{\alpha\rightarrow1}\Delta_{\alpha}\left(  \rho_{ABC},\tau_{AC}%
,\omega_{C},\theta_{BC}\right)  =\lim_{\alpha\rightarrow1}\frac{1}{Q_{\alpha
}\left(  \rho_{ABC},\tau_{AC},\omega_{C},\theta_{BC}\right)  }\frac{d}%
{d\alpha}Q_{\alpha}\left(  \rho_{ABC},\tau_{AC},\omega_{C},\theta_{BC}\right)
. \label{eq:lhospital-renyi-cmi}%
\end{equation}
We can evaluate the limits separately to find that%
\begin{equation}
\lim_{\alpha\rightarrow1}Q_{\alpha}\left(  \rho_{ABC},\tau_{AC},\omega
_{C},\theta_{BC}\right)  =\text{Tr}\left\{  \rho_{ABC}\tau_{AC}^{0}\omega
_{C}^{0}\theta_{BC}^{0}\omega_{C}^{0}\tau_{AC}^{0}\right\}  ,
\end{equation}%
\begin{multline}
\lim_{\alpha\rightarrow1}\frac{d}{d\alpha}Q_{\alpha}\left(  \rho_{ABC}%
,\tau_{AC},\omega_{C},\theta_{BC}\right)  =\text{Tr}\left\{  \left(  \log
\rho_{ABC}\right)  \rho_{ABC}\tau_{AC}^{0}\omega_{C}^{0}\theta_{BC}^{0}%
\omega_{C}^{0}\tau_{AC}^{0}\right\} \\
-\frac{1}{2}\text{Tr}\left\{  \rho_{ABC}\left(  \log\tau_{AC}\right)
\tau_{AC}^{0}\omega_{C}^{0}\theta_{BC}^{0}\omega_{C}^{0}\tau_{AC}^{0}\right\}
+\frac{1}{2}\text{Tr}\left\{  \rho_{ABC}\tau_{AC}^{0}\left(  \log\omega
_{C}\right)  \omega_{C}^{0}\theta_{BC}^{0}\omega_{C}^{0}\tau_{AC}^{0}\right\}
\\
-\text{Tr}\left\{  \rho_{ABC}\tau_{AC}^{0}\omega_{C}^{0}\left(  \log
\theta_{BC}\right)  \theta_{BC}^{0}\omega_{C}^{0}\tau_{AC}^{0}\right\}
+\frac{1}{2}\text{Tr}\left\{  \rho_{ABC}\tau_{AC}^{0}\omega_{C}^{0}\theta
_{BC}^{0}\left(  \log\omega_{C}\right)  \omega_{C}^{0}\tau_{AC}^{0}\right\} \\
-\frac{1}{2}\text{Tr}\left\{  \rho_{ABC}\tau_{AC}^{0}\omega_{C}^{0}\theta
_{BC}^{0}\omega_{C}^{0}\left(  \log\tau_{AC}\right)  \tau_{AC}^{0}\right\}  .
\end{multline}
Since by assumption supp$\left(  \rho_{ABC}\right)  $ is contained in each of
supp$\left(  \tau_{AC}\right)  $, supp$\left(  \omega_{C}\right)  $, and
supp$\left(  \theta_{BC}\right)  $, we exploit the relations $\rho_{ABC}%
=\rho_{ABC}^{0}\rho_{ABC}\rho_{ABC}^{0}$, $\rho_{ABC}^{0}\tau_{AC}^{0}%
=\rho_{ABC}^{0}$, $\rho_{ABC}^{0}\theta_{BC}^{0}=\rho_{ABC}^{0}$, $\rho
_{ABC}^{0}\omega_{C}^{0}=\rho_{ABC}^{0}$ and their Hermitian conjugates to
find that%
\begin{align}
\lim_{\alpha\rightarrow1}Q_{\alpha}\left(  \rho_{ABC},\tau_{AC},\omega
_{C},\theta_{BC}\right)   &  =1,\\
\lim_{\alpha\rightarrow1}\frac{d}{d\alpha}Q_{\alpha}\left(  \rho_{ABC}%
,\tau_{AC},\omega_{C},\theta_{BC}\right)   &  =\Delta\left(  \rho_{ABC}%
,\tau_{AC},\omega_{C},\theta_{BC}\right)  ,
\end{align}
which when combined with (\ref{eq:lhospital-renyi-cmi}) leads to
(\ref{eq:delta-alpha-to-delta}). Essentially the same proof establishes the
limiting relation for the other $\Delta_{\alpha}$ quantities defined from \eqref{eq:Q_a_2}-\eqref{eq:Q_a_6}.
\end{proof}

\begin{corollary}
Let $\rho_{ABC}\in\mathcal{S}\left(  \mathcal{H}_{ABC}\right)  $. Then the
following limiting relation holds%
\begin{equation}
\lim_{\alpha\rightarrow1}\Delta_{\alpha}\left(  \rho_{ABC},\rho_{AC},\rho
_{C},\rho_{BC}\right)  =I\left(  A;B|C\right)  _{\rho}.
\end{equation}

\end{corollary}

\begin{proof}
This follows from the fact that supp$\left(  \rho_{ABC}\right)  \subseteq
\ $supp$\left(  \rho_{AC}\right)  $, supp$\left(  \rho_{C}\right)  $,
supp$\left(  \rho_{BC}\right)  $ (see, e.g., \cite[Lemma~B.4.1]{RennerThesis}%
), from the above theorem, and by recalling that $\Delta\left(  \rho
_{ABC},\rho_{AC},\rho_{C},\rho_{BC}\right)  =I\left(  A;B|C\right)  _{\rho}$.
\end{proof}

\begin{theorem}
\label{thm:renyi-cmi-a-to-1}Let $\rho_{ABC}\in\mathcal{S}\left(
\mathcal{H}_{ABC}\right)  _{++}$. Then the R\'{e}nyi conditional mutual
information converges to the conditional mutual information in the limit as
$\alpha\rightarrow1$:%
\begin{equation}
\lim_{\alpha\rightarrow1}I_{\alpha}\left(  A;B|C\right)  _{\rho}=I\left(
A;B|C\right)  _{\rho}.
\end{equation}

\end{theorem}

The idea behind the proof of Theorem~\ref{thm:renyi-cmi-a-to-1} is the same as
that behind the proof of Theorem~\ref{thm:conv-vN}. However, we have the
explicit form for $I_{\alpha}\left(  A;B|C\right)  _{\rho}$ from
Proposition~\ref{prop:sibson}, which allows us to evaluate the limit without
the need for uniform convergence of $\Delta_{\alpha}\left(  \rho_{ABC}%
,\tau_{AC},\omega_{C},\theta_{BC}\right)  $ in $\tau_{AC}$, $\omega_{C}$, and
$\theta_{BC}$ as $\alpha\rightarrow1$. A proof of
Theorem~\ref{thm:renyi-cmi-a-to-1} appears in
Appendix~\ref{sec:renyi-cmi-a-to-1}.

\begin{remark}
\label{conj:uniform-renyi-converge}Let $\rho_{ABC}\in\mathcal{S}\left(
\mathcal{H}_{ABC}\right)  $, $\tau_{AC}\in\mathcal{S}\left(  \mathcal{H}%
_{AC}\right)  $, $\theta_{BC}\in\mathcal{S}\left(  \mathcal{H}_{BC}\right)  $,
and $\omega_{C}\in\mathcal{S}\left(  \mathcal{H}_{C}\right)  $ and suppose
that the support condition in \eqref{eq:support-condition}\ holds. If
$\Delta_{\alpha}\left(  \rho_{ABC},\tau_{AC},\omega_{C},\theta_{BC}\right)  $
converges uniformly in $\tau_{AC}$, $\omega_{C}$, and $\theta_{BC}$\ to
$\Delta\left(  \rho_{ABC},\tau_{AC},\omega_{C},\theta_{BC}\right)  $ as
$\alpha\rightarrow1$, then we could conclude that all R\'{e}nyi
generalizations of the conditional mutual information (as proposed at the
beginning of Section~\ref{sec:Renyi-CMI})\ converge to it in the limit as
$\alpha\rightarrow1$.
\end{remark}

\subsection{Monotonicity with respect to local quantum operations on one
system}

The following lemma is the critical one which will allow us to conclude that
the R\'{e}nyi conditional mutual information is monotone non-increasing with
respect to local quantum operations acting on one system for $\alpha\in
\lbrack0,1)\cup(1,2]$.

\begin{lemma}
\label{lem:monotone-lemma}Let $\rho_{ABC}\in\mathcal{S}\left(  \mathcal{H}%
_{ABC}\right)  $, $\tau_{AC}\in\mathcal{S}\left(  \mathcal{H}_{AC}\right)  $,
$\theta_{BC}\in\mathcal{S}\left(  \mathcal{H}_{BC}\right)  $, and $\omega
_{C}\in\mathcal{S}\left(  \mathcal{H}_{C}\right)  $ and suppose that the
non-orthogonality condition in \eqref{eq:non-orthogonal} holds. Let
$\mathcal{N}_{A\rightarrow A^{\prime}}$ and $\mathcal{M}_{B\rightarrow
B^{\prime}}$\ denote quantum operations acting on systems $A$ and $B$,
respectively. Then the following monotonicity inequalities hold for $\alpha
\in\lbrack0,1)\cup(1,2]$:%
\begin{align}
\Delta_{\alpha}\left(  \rho_{ABC},\tau_{AC},\omega_{C},\theta_{BC}\right)   &
\geq\Delta_{\alpha}\left(  \mathcal{M}_{B\rightarrow B^{\prime}}\left(
\rho_{ABC}\right)  ,\tau_{AC},\omega_{C},\mathcal{M}_{B\rightarrow B^{\prime}%
}\left(  \theta_{BC}\right)  \right)  ,\label{eq:mono-1}\\
\Delta_{\alpha}\left(  \rho_{ABC},\omega_{C},\tau_{AC},\theta_{BC}\right)   &
\geq\Delta_{\alpha}\left(  \mathcal{M}_{B\rightarrow B^{\prime}}\left(
\rho_{ABC}\right)  ,\omega_{C},\tau_{AC},\mathcal{M}_{B\rightarrow B^{\prime}%
}\left(  \theta_{BC}\right)  \right)  ,\label{eq:mono-other-quant}\\
\Delta_{\alpha}\left(  \rho_{ABC},\omega_{C},\theta_{BC},\tau_{AC}\right)   &
\geq\Delta_{\alpha}\left(  \mathcal{N}_{A\rightarrow A^{\prime}}\left(
\rho_{ABC}\right)  ,\omega_{C},\theta_{BC},\mathcal{N}_{A\rightarrow
A^{\prime}}\left(  \tau_{AC}\right)  \right)  ,\\
\Delta_{\alpha}\left(  \rho_{ABC},\theta_{BC},\omega_{C},\tau_{AC}\right)   &
\geq\Delta_{\alpha}\left(  \mathcal{N}_{A\rightarrow A^{\prime}}\left(
\rho_{ABC}\right)  ,\theta_{BC},\omega_{C},\mathcal{N}_{A\rightarrow
A^{\prime}}\left(  \tau_{AC}\right)  \right)  . \label{eq:mono-other-quant-3}%
\end{align}

\end{lemma}

\begin{proof}
We begin by proving \eqref{eq:mono-1}. Consider that $Q_{\alpha}\left(
\rho_{ABC},\tau_{AC},\omega_{C},\theta_{BC}\right)  $ is jointly concave in
$\rho_{ABC}$ and $\theta_{BC}$ when $\alpha\in\lbrack0,1)$. This is a result
of Lieb's concavity theorem \cite{L73}, a special case of which is the
statement that the function%
\begin{equation}
\left(  S,R\right)  \in\mathcal{B}(\mathcal{H})_{+}\times\mathcal{B}%
(\mathcal{H})_{+}\rightarrow\text{Tr}\left\{  S^{\lambda}XR^{1-\lambda}%
X^{\dag}\right\}
\end{equation}
is jointly concave in $S$ and $R$ when $\lambda\in\left[  0,1\right]  $. (We
apply the theorem by choosing $S=\rho_{ABC}$, $R=\theta_{BC}$, and
$X=\tau_{AC}^{\left(  1-\alpha\right)  /2}\omega_{C}^{\left(  \alpha-1\right)
/2}$.) Furthermore, by an application of Ando's convexity theorem \cite{A79},
we know that $Q_{\alpha}\left(  \rho_{ABC},\tau_{AC},\omega_{C},\theta
_{BC}\right)  $ is jointly convex in $\rho_{ABC}$ and $\theta_{BC}$ when
$\alpha\in(1,2]$.

By a standard (well known)\ argument due to Uhlmann \cite{U73}, the
monotonicity inequality in (\ref{eq:mono-1}) holds. For completeness, we
detail this standard argument here for the case when $\alpha\in\lbrack0,1)$.
Note that it suffices to prove the following monotonicity under partial trace:%
\begin{equation}
Q_{\alpha}\left(  \rho_{AB_{1}B_{2}C},\tau_{AC},\omega_{C},\theta_{B_{1}%
B_{2}C}\right)  \leq Q_{\alpha}\left(  \rho_{AB_{1}C},\tau_{AC},\omega
_{C},\theta_{B_{1}C}\right)  ,
\end{equation}
because the $Q_{\alpha}$ quantity is clearly invariant under isometries acting
on system $B$ and the Stinespring representation theorem \cite{S55}\ states
that any quantum channel can be modeled as an isometry followed by a partial
trace. To this end, let $\left\{  U_{B_{2}}^{i}\right\}  _{i=0}^{d_{B_{2}}%
^{2}-1}$ denote the set of Heisenberg-Weyl operators acting on the system
$B_{2}$, with $d_{B_{2}}$ the dimension of system $B_{2}$. Then%
\begin{multline}
Q_{\alpha}\left(  \rho_{AB_{1}B_{2}C},\tau_{AC},\omega_{C},\theta_{B_{1}%
B_{2}C}\right) \\
=\frac{1}{d_{B_{2}}^{2}}\sum_{i=0}^{d_{B_{2}}^{2}-1}Q_{\alpha}\left(
U_{B_{2}}^{i}\rho_{AB_{1}B_{2}C}\left(  U_{B_{2}}^{i}\right)  ^{\dag}%
,\tau_{AC},\omega_{C},U_{B_{2}}^{i}\theta_{B_{1}B_{2}C}\left(  U_{B_{2}}%
^{i}\right)  ^{\dag}\right)  .
\end{multline}
We can then invoke the Lieb concavity theorem to conclude that%
\begin{align}
&  Q_{\alpha}\left(  \rho_{AB_{1}B_{2}C},\tau_{AC},\omega_{C},\theta
_{B_{1}B_{2}C}\right) \nonumber\\
&  \leq Q_{\alpha}\left(  \frac{1}{d_{B_{2}}^{2}}\sum_{i}U_{B_{2}}^{i}%
\rho_{AB_{1}B_{2}C}\left(  U_{B_{2}}^{i}\right)  ^{\dag},\tau_{AC},\omega
_{C},\frac{1}{d_{B_{2}}^{2}}\sum_{i}U_{B_{2}}^{i}\theta_{B_{1}B_{2}C}\left(
U_{B_{2}}^{i}\right)  ^{\dag}\right) \\
&  =Q_{\alpha}\left(  \rho_{AB_{1}C}\otimes\pi_{B_{2}},\tau_{AC},\omega
_{C},\theta_{B_{1}C}\otimes\pi_{B_{2}}\right) \\
&  =Q_{\alpha}\left(  \rho_{AB_{1}C},\tau_{AC},\omega_{C},\theta_{B_{1}%
C}\right)  ,
\end{align}
where $\pi$ is the maximally mixed state. After taking logarithms and dividing
by $\alpha-1$, we can conclude the monotonicity for $\alpha\in\lbrack0,1)$. A
similar development with Ando's convexity theorem gets the monotonicity for
$\alpha\in(1,2]$. The inequalities in (\ref{eq:mono-other-quant}%
)-(\ref{eq:mono-other-quant-3}) follow from a similar line of reasoning.
\end{proof}

\begin{remark}
Let $\rho_{ABC}\in\mathcal{S}\left(  \mathcal{H}_{ABC}\right)  $, $\tau
_{AC}\in\mathcal{S}\left(  \mathcal{H}_{AC}\right)  $, $\theta_{BC}%
\in\mathcal{S}\left(  \mathcal{H}_{BC}\right)  $, and $\omega_{C}%
\in\mathcal{S}\left(  \mathcal{H}_{C}\right)  $ and suppose that the
non-orthogonality condition in \eqref{eq:non-orthogonal} holds. It is an open
question to determine whether the $\Delta_{\alpha}$\ quantities defined from
\eqref{eq:Q_a_1}, \eqref{eq:Q_a_2}-\eqref{eq:Q_a_6} are monotone
non-increasing with respect to quantum operations acting on either systems $A$
or $B$ for $\alpha\in\lbrack0,1)\cup(1,2]$. In particular, it is an open
question to determine whether $\Delta_{\alpha}\left(  \rho_{ABC},\rho
_{AC},\rho_{C},\rho_{BC}\right)  $ and $\inf_{\theta_{BC}}\Delta_{\alpha
}\left(  \rho_{ABC},\rho_{AC},\rho_{C},\theta_{BC}\right)  $\ are monotone
non-increasing with respect to quantum operations acting on system $A$ for
$\alpha\in\lbrack0,1)\cup(1,2]$.
\end{remark}

\begin{corollary}
\label{thm:Renyi-CMI-monotone}Let $\rho_{ABC}\in\mathcal{S}\left(
\mathcal{H}_{ABC}\right)  $, $\tau_{AC}\in\mathcal{S}\left(  \mathcal{H}%
_{AC}\right)  $, $\theta_{BC}\in\mathcal{S}\left(  \mathcal{H}_{BC}\right)  $,
and $\omega_{C}\in\mathcal{S}\left(  \mathcal{H}_{C}\right)  $. All R\'{e}nyi
generalizations of the conditional mutual information derived from%
\begin{equation}
\Delta_{\alpha}\left(  \rho_{ABC},\tau_{AC},\omega_{C},\theta_{BC}\right)
,\ \ \ \ \Delta_{\alpha}\left(  \rho_{ABC},\omega_{C},\tau_{AC},\theta
_{BC}\right)  , \label{eq:mono-B-1}%
\end{equation}
are monotone non-increasing with respect to quantum operations acting on
system $B$, for $\alpha\in\lbrack0,1)\cup(1,2]$. All R\'{e}nyi generalizations
of the conditional mutual information derived from%
\begin{equation}
\Delta_{\alpha}\left(  \rho_{ABC},\omega_{C},\theta_{BC},\tau_{AC}\right)
,\ \ \ \ \Delta_{\alpha}\left(  \rho_{ABC},\theta_{BC},\omega_{C},\tau
_{AC}\right)  , \label{eq:mono-B-4}%
\end{equation}
are monotone non-increasing with respect to quantum operations acting on
system $A$, for $\alpha\in\lbrack0,1)\cup(1,2]$. The derived R\'{e}nyi
generalizations are optimized with respect to $\tau_{AC}$, $\omega_{C}$, and
$\theta_{BC}$ satisfying the support condition in \eqref{eq:support-condition}
(which implies the non-orthogonality condition in \eqref{eq:non-orthogonal}).
\end{corollary}

\begin{proof}
We prove that a variation derived from (\ref{eq:CMI-rel-ent-2}) obeys the
monotonicity (with the others mentioned above following from similar ideas).
Beginning with the inequality in Lemma~\ref{lem:monotone-lemma}, we find that%
\begin{align}
\sup_{\omega_{C}}\Delta_{\alpha}\left(  \rho_{ABC},\tau_{AC},\omega_{C}%
,\theta_{BC}\right)   &  \geq\sup_{\omega_{C}}\Delta_{\alpha}\left(
\mathcal{M}_{B\rightarrow B^{\prime}}\left(  \rho_{ABC}\right)  ,\tau
_{AC},\omega_{C},\mathcal{M}_{B\rightarrow B^{\prime}}\left(  \theta
_{BC}\right)  \right)  ,\\
&  \geq\inf_{\tau_{AC}^{\prime},\theta_{BC}^{\prime}}\sup_{\omega_{C}}%
\Delta_{\alpha}\left(  \mathcal{M}_{B\rightarrow B^{\prime}}\left(  \rho
_{ABC}\right)  ,\tau_{AC}^{\prime},\omega_{C},\theta_{BC}^{\prime}\right)  .
\end{align}
Since this inequality holds for all $\tau_{AC}$ and $\theta_{BC}$, it holds in
particular for the infimum of the first line over all such states,
establishing monotonicity for the R\'{e}nyi generalization of the conditional
mutual information derived from (\ref{eq:CMI-rel-ent-2}).
\end{proof}

\begin{corollary}
\label{cor:Renyi-CMI-positive}We can employ the monotonicity inequalities from
Lemma~\ref{lem:monotone-lemma} to conclude that some R\'{e}nyi generalizations
of the conditional mutual information derived from
\eqref{eq:mono-B-1}-\eqref{eq:mono-B-4} and
Proposition~\ref{thm:CMI-formulations} are non-negative for all $\alpha
\in\lbrack0,1)\cup(1,2]$. This includes $\Delta_{\alpha}\left(  \rho
_{ABC},\rho_{AC},\rho_{C},\rho_{BC}\right)  $ and the one from
Definition~\ref{def:Renyi-CMI}.
\end{corollary}

\begin{proof}
Let $\rho_{ABC}\in\mathcal{S}\left(  \mathcal{H}_{ABC}\right)  $, $\tau
_{AC}\in\mathcal{S}\left(  \mathcal{H}_{AC}\right)  $, $\theta_{BC}%
\in\mathcal{S}\left(  \mathcal{H}_{BC}\right)  $, and $\omega_{C}%
\in\mathcal{S}\left(  \mathcal{H}_{C}\right)  $ and suppose that the support
condition in (\ref{eq:support-condition}) holds. A common proof technique
applies to reach the conclusions stated above. We illustrate with an example
for%
\begin{equation}
\inf_{\theta_{BC}}\sup_{\omega_{C}}\Delta_{\alpha}\left(  \rho_{ABC},\rho
_{AC},\omega_{C},\theta_{BC}\right)  . \label{eq:renyi-gen-non-neg}%
\end{equation}
We apply Lemma~\ref{lem:monotone-lemma}, choosing the local map on system $B$
to be a trace-out map, to conclude that%
\begin{equation}
\Delta_{\alpha}\left(  \rho_{ABC},\rho_{AC},\omega_{C},\theta_{BC}\right)
\geq\Delta_{\alpha}\left(  \rho_{AC},\rho_{AC},\omega_{C},\theta_{C}\right)  .
\end{equation}
Then, we can conclude that%
\begin{align}
\sup_{\omega_{C}}\Delta_{\alpha}\left(  \rho_{ABC},\rho_{AC},\omega_{C}%
,\theta_{BC}\right)   &  \geq\sup_{\omega_{C}}\Delta_{\alpha}\left(  \rho
_{AC},\rho_{AC},\omega_{C},\theta_{C}\right) \\
&  \geq\Delta_{\alpha}\left(  \rho_{AC},\rho_{AC},\theta_{C},\theta_{C}\right)
\\
&  =\frac{1}{\alpha-1}\log\text{Tr}\left\{  \rho_{AC}^{\alpha}\rho
_{AC}^{\left(  1-\alpha\right)  /2}\theta_{C}^{\left(  \alpha-1\right)
/2}\theta_{C}^{1-\alpha}\theta_{C}^{\left(  \alpha-1\right)  /2}\rho
_{AC}^{\left(  1-\alpha\right)  /2}\right\} \\
&  =\frac{1}{\alpha-1}\log\text{Tr}\left\{  \rho_{AC}\theta_{C}^{0}\right\} \\
&  =0,
\end{align}
with the last inequality following from the support condition supp$\left(
\rho_{ABC}\right)  \subseteq\ $supp$\left(  \theta_{BC}\right)  $ implying the
support condition supp$\left(  \rho_{AC}\right)  \subseteq\ $supp$\left(
\theta_{C}\right)  $ \cite[Lemma~B.4.2]{RennerThesis}. Since the inequality
holds for all $\theta_{BC}$ satisfying the support condition, we can conclude
that the quantity in (\ref{eq:renyi-gen-non-neg}) is non-negative. A similar
technique can be used to conclude that other R\'{e}nyi generalizations of the
conditional mutual information are non-negative (including the one in
Definition~\ref{def:Renyi-CMI}).
\end{proof}

\begin{remark}
If the system $C$ is classical, then the R\'{e}nyi conditional mutual
information given in Definition~\ref{def:Renyi-CMI} is monotone with respect
to local operations on both $A$ and $B$. This is because the optimizing state
is classical on system $C$ and then we have the commutation%
\begin{equation}
\rho_{AC}^{\left(  1-\alpha\right)  /2}\rho_{C}^{\left(  \alpha-1\right)
/2}\sigma_{BC}^{1-\alpha}\rho_{C}^{\left(  \alpha-1\right)  /2}\rho
_{AC}^{\left(  1-\alpha\right)  /2}=\sigma_{BC}^{\left(  1-\alpha\right)
/2}\rho_{C}^{\left(  \alpha-1\right)  /2}\rho_{AC}^{1-\alpha}\rho_{C}^{\left(
\alpha-1\right)  /2}\sigma_{BC}^{\left(  1-\alpha\right)  /2}.
\end{equation}

\end{remark}

\begin{remark}
It is an open question to determine whether all R\'{e}nyi generalizations of
the conditional mutual information designed from the different optimizations
in Proposition~\ref{thm:CMI-formulations}\ and the different orderings in
\eqref{eq:Q_a_1}, \eqref{eq:Q_a_2}-\eqref{eq:Q_a_6} are non-negative for
$\alpha\in\lbrack0,1)\cup(1,2]$.
\end{remark}

\section{Sandwiched R\'{e}nyi conditional mutual information}

\label{sec:sandwiched-Renyi-CMI}As in the previous section, there are many
ways in which we can define a sandwiched R\'{e}nyi conditional mutual
information. Let $\rho_{ABC}\in\mathcal{S}\left(  \mathcal{H}_{ABC}\right)  $,
$\tau_{AC}\in\mathcal{S}\left(  \mathcal{H}_{AC}\right)  $, $\theta_{BC}%
\in\mathcal{S}\left(  \mathcal{H}_{BC}\right)  $, and $\omega_{C}%
\in\mathcal{S}\left(  \mathcal{H}_{C}\right)  $. We define the following core
quantities for $\alpha\in\left(  0,1\right)  \cup\left(  1,\infty\right)  $:%
\begin{align}
\widetilde{Q}_{\alpha}\left(  \rho_{ABC},\tau_{AC},\omega_{C},\theta
_{BC}\right)   &  \equiv\text{Tr}\left\{  \left(  \rho_{ABC}^{1/2}\tau
_{AC}^{\left(  1-\alpha\right)  /2\alpha}\omega_{C}^{\left(  \alpha-1\right)
/2\alpha}\theta_{BC}^{\left(  1-\alpha\right)  /\alpha}\omega_{C}^{\left(
\alpha-1\right)  /2\alpha}\tau_{AC}^{\left(  1-\alpha\right)  /2\alpha}%
\rho_{ABC}^{1/2}\right)  ^{\alpha}\right\}  ,\label{eq:Q_SW_a_1}\\
\widetilde{\Delta}_{\alpha}\left(  \rho_{ABC},\tau_{AC},\omega_{C},\theta
_{BC}\right)   &  \equiv\frac{1}{\alpha-1}\log\widetilde{Q}_{\alpha}\left(
\rho_{ABC},\tau_{AC},\omega_{C},\theta_{BC}\right)  . \label{eq:Delta_SW_a_1}%
\end{align}
We stress again that the formula above is to be interpreted in terms of
generalized inverses. By employing (\ref{eq:a-norm}) and (\ref{eq:Q_SW_a_1}),
we can write%
\begin{equation}
\widetilde{Q}_{\alpha}\left(  \rho_{ABC},\tau_{AC},\omega_{C},\theta
_{BC}\right)  =\left\Vert \rho_{ABC}^{1/2}\tau_{AC}^{\left(  1-\alpha\right)
/2\alpha}\omega_{C}^{\left(  \alpha-1\right)  /2\alpha}\theta_{BC}^{\left(
1-\alpha\right)  /2\alpha}\right\Vert _{2\alpha}^{2\alpha},
\end{equation}
and we see that $\widetilde{Q}_{\alpha}\left(  \rho_{ABC},\tau_{AC},\omega
_{C},\theta_{BC}\right)  =0$ if and only if%
\begin{equation}
\rho_{ABC}^{1/2}\tau_{AC}^{\left(  1-\alpha\right)  /2\alpha}\omega
_{C}^{\left(  \alpha-1\right)  /2\alpha}\theta_{BC}^{\left(  1-\alpha\right)
/2\alpha}=0.
\end{equation}
So $\widetilde{Q}_{\alpha}\left(  \rho_{ABC},\tau_{AC},\omega_{C},\theta
_{BC}\right)  >0$ if%
\begin{equation}
\rho_{ABC}^{1/2}\not \perp \tau_{AC}^{\left(  1-\alpha\right)  /2\alpha}%
\omega_{C}^{\left(  \alpha-1\right)  /2\alpha}\theta_{BC}^{\left(
1-\alpha\right)  /2\alpha}. \label{eq:non-orthogonality-SW}%
\end{equation}
The non-orthogonality condition in (\ref{eq:non-orthogonality-SW})\ is
satisfied, e.g., if the support condition in (\ref{eq:support-condition})
holds, so that (\ref{eq:non-orthogonality-SW})\ is satisfied when $\tau
_{AC}=\rho_{AC}$, $\omega_{C}=\rho_{C}$, and $\theta_{BC}=\rho_{BC}$. It
remains largely open to determine support conditions under which%
\begin{equation}
\lim_{\xi\searrow0}\widetilde{\Delta}_{\alpha}\left(  \rho_{ABC},\tau_{AC}+\xi
I_{ABC},\omega_{C}+\xi I_{ABC},\theta_{BC}+\xi I_{ABC}\right)
\label{eq:finite-Delta-SW}%
\end{equation}
is finite and equal to (\ref{eq:Delta_SW_a_1}), with complications being due
to the fact that (\ref{eq:Q_SW_a_1}) features the multiplication of several
non-commuting operators which can interact in non-trivial ways. As before, we
define five other different $\widetilde{Q}_{\alpha}$ quantities, again
uniquely identified by the order of the last three arguments:%
\begin{align}
\widetilde{Q}_{\alpha}\left(  \rho_{ABC},\theta_{BC},\omega_{C},\tau
_{AC}\right)   &  \equiv\left\Vert \rho_{ABC}^{1/2}\theta_{BC}^{\left(
1-\alpha\right)  /2\alpha}\omega_{C}^{\left(  \alpha-1\right)  /2\alpha}%
\tau_{AC}^{\left(  1-\alpha\right)  /2\alpha}\right\Vert _{2\alpha}^{2\alpha
},\label{eq:Q_SW_a_2}\\
\widetilde{Q}_{\alpha}\left(  \rho_{ABC},\omega_{C},\tau_{AC},\theta
_{BC}\right)   &  \equiv\left\Vert \rho_{ABC}^{1/2}\omega_{C}^{\left(
\alpha-1\right)  /2\alpha}\tau_{AC}^{\left(  1-\alpha\right)  /2\alpha}%
\theta_{BC}^{\left(  1-\alpha\right)  /2\alpha}\right\Vert _{2\alpha}%
^{2\alpha},\\
\widetilde{Q}_{\alpha}\left(  \rho_{ABC},\omega_{C},\theta_{BC},\tau
_{AC}\right)   &  \equiv\left\Vert \rho_{ABC}^{1/2}\omega_{C}^{\left(
\alpha-1\right)  /2\alpha}\theta_{BC}^{\left(  1-\alpha\right)  /2\alpha}%
\tau_{AC}^{\left(  1-\alpha\right)  /2\alpha}\right\Vert _{2\alpha}^{2\alpha
},\\
\widetilde{Q}_{\alpha}\left(  \rho_{ABC},\tau_{AC},\theta_{BC},\omega
_{C}\right)   &  \equiv\left\Vert \rho_{ABC}^{1/2}\tau_{AC}^{\left(
1-\alpha\right)  /2\alpha}\theta_{BC}^{\left(  1-\alpha\right)  /2\alpha
}\omega_{C}^{\left(  \alpha-1\right)  /2\alpha}\right\Vert _{2\alpha}%
^{2\alpha},\\
\widetilde{Q}_{\alpha}\left(  \rho_{ABC},\theta_{BC},\tau_{AC},\omega
_{C}\right)   &  \equiv\left\Vert \rho_{ABC}^{1/2}\theta_{BC}^{\left(
1-\alpha\right)  /2\alpha}\tau_{AC}^{\left(  1-\alpha\right)  /2\alpha}%
\omega_{C}^{\left(  \alpha-1\right)  /2\alpha}\right\Vert _{2\alpha}^{2\alpha
}. \label{eq:Q_SW_a_6}%
\end{align}
These then lead to different $\widetilde{\Delta}_{\alpha}$ quantities. We call
the quantities above \textquotedblleft sandwiched\textquotedblright\ because
they can be viewed as having their root in the sandwiched R\'{e}nyi relative
entropy, i.e., for $\rho_{ABC}\in\mathcal{S(H}_{ABC})_{++}$, $\tau_{AC}%
\in\mathcal{S(H}_{AC})_{++}$, $\theta_{BC}\in\mathcal{S(H}_{BC})_{++}$, and
$\omega_{C}\in\mathcal{S(H}_{C})_{++}$:%
\begin{equation}
\widetilde{\Delta}_{\alpha}\left(  \rho_{ABC},\tau_{AC},\omega_{C},\theta
_{BC}\right)  =\widetilde{D}_{\alpha}\left(  \rho_{ABC}\middle\Vert\left[
\tau_{AC}^{\left(  1-\alpha\right)  /2\alpha}\omega_{C}^{\left(
\alpha-1\right)  /2\alpha}\theta_{BC}^{\left(  1-\alpha\right)  /\alpha}%
\omega_{C}^{\left(  \alpha-1\right)  /2\alpha}\tau_{AC}^{\left(
1-\alpha\right)  /2\alpha}\right]  ^{\alpha/\left(  1-\alpha\right)  }\right)
.
\end{equation}
\ 

Although there are many different possible sandwiched R\'{e}nyi
generalizations of the conditional mutual information, found by combining the
different $\widetilde{\Delta}_{\alpha}$ quantities discussed above with the
different optimizations summarized in Proposition~\ref{thm:CMI-formulations},
we choose the definition given below because it obeys many of the properties
that the conditional mutual information does.

\begin{definition}
\label{def:renyicmi_sand}Let $\rho_{ABC}\in\mathcal{S}\left(  \mathcal{H}%
_{ABC}\right)  $. The sandwiched R\'{e}nyi conditional mutual information is
defined as%
\begin{equation}
\widetilde{I}_{\alpha}\left(  A;B|C\right)  _{\rho}\equiv\inf_{\sigma_{BC}%
}\sup_{\omega_{C}}\widetilde{\Delta}_{\alpha}\left(  \rho_{ABC},\rho
_{AC},\omega_{C},\sigma_{BC}\right)  , \label{eq:sw-Renyi-CMI}%
\end{equation}
where the optimizations are over states obeying the support conditions in \eqref{eq:support-condition}.
\end{definition}

Again, unlike the conditional mutual information, this definition is not
symmetric with respect to $A$ and $B$.Thus one might also call it the
sandwiched R\'{e}nyi information that $B$ has about $A$ from the perspective
of $C$. Also, for trivial $C$, the definition reduces to the usual definition
of sandwiched R\'{e}nyi mutual information (see, e.g., \cite{WWY13,GW13,CMW14}).

\subsection{Limit of the sandwiched R\'{e}nyi conditional mutual information
as $\alpha\rightarrow1$}

This section considers the limit of the $\widetilde{\Delta}_{\alpha}$
quantities as $\alpha\rightarrow1$. For technical reasons, we restrict the
development to positive definite density operators. It remains open to
determine whether the following theorems hold under less restrictive conditions.

\begin{theorem}
\label{thm:conv-vN-sandwiched}Let $\rho_{ABC}\in\mathcal{S}\left(
\mathcal{H}_{ABC}\right)  _{++}$, $\tau_{AC}\in\mathcal{S}\left(
\mathcal{H}_{AC}\right)  _{++}$, $\theta_{BC}\in\mathcal{S}\left(
\mathcal{H}_{BC}\right)  _{++}$, and $\omega_{C}\in\mathcal{S}\left(
\mathcal{H}_{C}\right)  _{++}$. Then%
\begin{equation}
\lim_{\alpha\rightarrow1}\widetilde{\Delta}_{\alpha}\left(  \rho_{ABC}%
,\tau_{AC},\omega_{C},\theta_{BC}\right)  =\Delta\left(  \rho_{ABC},\tau
_{AC},\omega_{C},\theta_{BC}\right)  .
\end{equation}
The same limiting relation holds for the other $\widetilde{\Delta}_{\alpha}$
quantities defined from \eqref{eq:Q_SW_a_2}-\eqref{eq:Q_SW_a_6}.
\end{theorem}

The proof of Theorem~\ref{thm:conv-vN-sandwiched} is very similar to the proof
of Theorem~\ref{thm:conv-vN} and presented in
Appendix~\ref{sec:sandwiched-renyi-a-to-1}.

\begin{corollary}
Let $\rho_{ABC}\in\mathcal{S}\left(  \mathcal{H}_{ABC}\right)  _{++}$. The
following limiting relation holds%
\begin{equation}
\lim_{\alpha\rightarrow1}\widetilde{\Delta}_{\alpha}\left(  \rho_{ABC}%
,\rho_{AC},\rho_{C},\rho_{BC}\right)  =I\left(  A;B|C\right)  _{\rho}.
\end{equation}

\end{corollary}

\begin{proof}
This follows from the fact that supp$\left(  \rho_{ABC}\right)  \subseteq
\ $supp$\left(  \rho_{AC}\right)  $, supp$\left(  \rho_{C}\right)  $,
supp$\left(  \rho_{BC}\right)  $ (see, e.g., \cite[Lemma~B.4.1]{RennerThesis}%
), Theorem~\ref{thm:conv-vN-sandwiched}, and by recalling that $\Delta\left(
\rho_{ABC},\rho_{AC},\rho_{C},\rho_{BC}\right)  =I\left(  A;B|C\right)
_{\rho}$.
\end{proof}

\begin{remark}
\label{conj:uniform-SW-renyi-converge}Let $\rho_{ABC}\in\mathcal{S}\left(
\mathcal{H}_{ABC}\right)  _{++}$, $\tau_{AC}\in\mathcal{S}\left(
\mathcal{H}_{AC}\right)  _{++}$, $\theta_{BC}\in\mathcal{S}\left(
\mathcal{H}_{BC}\right)  _{++}$, and $\omega_{C}\in\mathcal{S}\left(
\mathcal{H}_{C}\right)  _{++}$. If $\widetilde{\Delta}_{\alpha}\left(
\rho_{ABC},\tau_{AC},\omega_{C},\theta_{BC}\right)  $ converges uniformly in
$\tau_{AC}$, $\omega_{C}$, $\theta_{BC}$\ to $\Delta\left(  \rho_{ABC}%
,\tau_{AC},\omega_{C},\theta_{BC}\right)  $ as $\alpha\rightarrow1$, then we
could conclude that all sandwiched R\'{e}nyi generalizations of the
conditional mutual information (as proposed at the beginning of
Section~\ref{sec:sandwiched-Renyi-CMI})\ converge to it in the limit as
$\alpha\rightarrow1$. In particular, uniform convergence implies that
$\widetilde{I}_{\alpha}\left(  A;B|C\right)  _{\rho}$ converges to $I\left(
A;B|C\right)  _{\rho}$ as $\alpha\rightarrow1$.
\end{remark}

\subsection{Monotonicity under local quantum operations on one system}

This section considers monotonicity of the $\widetilde{\Delta}_{\alpha}$
quantities under local quantum operations. For technical reasons, we restrict
the development to positive definite density operators. It remains open to
determine whether the following theorems hold under less restrictive conditions.

\begin{lemma}
\label{lem:monotone-lemma-sandwiched}Let $\rho_{ABC}\in\mathcal{S(H}%
_{ABC})_{++}$, $\tau_{AC}\in\mathcal{S(H}_{AC})_{++}$, $\theta_{BC}%
\in\mathcal{S(H}_{BC})_{++}$, and $\omega_{C}\in\mathcal{S(H}_{C})_{++}$. Let
$\mathcal{N}_{A\rightarrow A^{\prime}}$ and $\mathcal{M}_{B\rightarrow
B^{\prime}}$\ denote quantum operations acting on systems $A$ and $B$,
respectively. Then the following monotonicity inequalities hold for all
$\alpha\in\lbrack1/2,1)\cup(1,\infty)$:%
\begin{align}
\widetilde{\Delta}_{\alpha}\left(  \rho_{ABC},\tau_{AC},\omega_{C},\theta
_{BC}\right)   &  \geq\widetilde{\Delta}_{\alpha}\left(  \mathcal{M}%
_{B\rightarrow B^{\prime}}\left(  \rho_{ABC}\right)  ,\tau_{AC},\omega
_{C},\mathcal{M}_{B\rightarrow B^{\prime}}\left(  \theta_{BC}\right)  \right)
,\label{eq:mono-sandwiched-ops-on-B}\\
\widetilde{\Delta}_{\alpha}\left(  \rho_{ABC},\omega_{C},\tau_{AC},\theta
_{BC}\right)   &  \geq\widetilde{\Delta}_{\alpha}\left(  \mathcal{M}%
_{B\rightarrow B^{\prime}}\left(  \rho_{ABC}\right)  ,\omega_{C},\tau
_{AC},\mathcal{M}_{B\rightarrow B^{\prime}}\left(  \theta_{BC}\right)
\right)  ,\label{eq:other-mono-sandwiched}\\
\widetilde{\Delta}_{\alpha}\left(  \rho_{ABC},\omega_{C},\theta_{BC},\tau
_{AC}\right)   &  \geq\widetilde{\Delta}_{\alpha}\left(  \mathcal{N}%
_{A\rightarrow A^{\prime}}\left(  \rho_{ABC}\right)  ,\omega_{C},\theta
_{BC},\mathcal{N}_{A\rightarrow A^{\prime}}\left(  \tau_{AC}\right)  \right)
,\\
\widetilde{\Delta}_{\alpha}\left(  \rho_{ABC},\theta_{BC},\omega_{C},\tau
_{AC}\right)   &  \geq\widetilde{\Delta}_{\alpha}\left(  \mathcal{N}%
_{A\rightarrow A^{\prime}}\left(  \rho_{ABC}\right)  ,\theta_{BC},\omega
_{C},\mathcal{N}_{A\rightarrow A^{\prime}}\left(  \tau_{AC}\right)  \right)  .
\label{eq:other-mono-sandwiched-3}%
\end{align}

\end{lemma}

\begin{proof}
We first focus on establishing the inequality in
\eqref{eq:mono-sandwiched-ops-on-B} for $\alpha\in\lbrack1/2,1)$. From part 1)
of \cite[Theorem~1.1]{Hiai20131568}, we know that the following function is
jointly concave in $S$ and $T$:%
\begin{equation}
\left(  S,T\right)  \in\mathcal{B}(\mathcal{H})_{++}\times\mathcal{B}%
(\mathcal{H})_{++}\mapsto\text{Tr}\left\{  \left[  \Phi\left(  S^{p}\right)
^{1/2}\Psi\left(  T^{q}\right)  \Phi\left(  S^{p}\right)  ^{1/2}\right]
^{s}\right\}  ,
\end{equation}
for strictly positive maps $\Phi\left(  \cdot\right)  $ and $\Psi\left(
\cdot\right)  $, $0<p,q\leq1$, and $1/2\leq s\leq1/\left(  p+q\right)  $. We
can then see that $\widetilde{Q}_{\alpha}\left(  \rho_{ABC},\tau_{AC}%
,\omega_{C},\theta_{BC}\right)  $ is of this form, with%
\begin{align}
\Psi &  =\tau_{AC}^{\left(  1-\alpha\right)  /2\alpha}\omega_{C}^{\left(
\alpha-1\right)  /2\alpha}\left(  \cdot\right)  \omega_{C}^{\left(
\alpha-1\right)  /2\alpha}\tau_{AC}^{\left(  1-\alpha\right)  /2\alpha},\\
q  &  =\frac{1-\alpha}{\alpha},\\
\Phi\left(  \cdot\right)   &  =\text{id},\\
p  &  =1,\\
s  &  =\alpha.
\end{align}
For the range $\alpha\in\lbrack1/2,1)$, we have that $p\in(0,1]$ and
$1/\left(  p+q\right)  =\alpha$, so that the conditions of part 1) of
\cite[Theorem~1.1]{Hiai20131568} are satisfied. We conclude that
$\widetilde{Q}_{\alpha}\left(  \rho_{ABC},\tau_{AC},\omega_{C},\theta
_{BC}\right)  $ is jointly concave in $\theta_{BC}$ and $\rho_{ABC}$. From
this, we can conclude the monotonicity in \eqref{eq:mono-sandwiched-ops-on-B}
for $\alpha\in\lbrack1/2,1)$. A similar proof establishes the inequalities in
(\ref{eq:other-mono-sandwiched})-(\ref{eq:other-mono-sandwiched-3}) for
$\alpha\in\lbrack1/2,1)$.

The proof of \eqref{eq:mono-sandwiched-ops-on-B} for $\alpha\in(1,\infty)$ is
a straightforward generalization of the technique used for
\cite[Proposition~3]{FL13}. To prove \eqref{eq:mono-sandwiched-ops-on-B}, it
suffices to prove that the following function%
\begin{equation}
\left(  \rho_{ABC},\theta_{BC}\right)  \in\mathcal{S}(\mathcal{H}_{ABC}%
)_{++}\times\mathcal{S}(\mathcal{H}_{ABC})_{++}\mapsto\text{Tr}\left\{
\left[  \rho_{ABC}^{1/2}K\left(  \alpha\right)  \rho_{ABC}^{1/2}\right]
^{\alpha}\right\}  \label{eq:trace-function-sw-a1}%
\end{equation}
is jointly convex for $\alpha\in\left(  1,\infty\right)  $, where%
\begin{equation}
K\left(  \alpha\right)  \equiv\tau_{AC}^{\left(  1-\alpha\right)  /2\alpha
}\omega_{C}^{\left(  \alpha-1\right)  /2\alpha}\theta_{BC}^{\left(
1-\alpha\right)  /\alpha}\omega_{C}^{\left(  \alpha-1\right)  /2\alpha}%
\tau_{AC}^{\left(  1-\alpha\right)  /2\alpha}.
\end{equation}
To this end, consider that we can write the trace function in
(\ref{eq:trace-function-sw-a1}) as%
\begin{equation}
\text{Tr}\left\{  \left[  \rho_{ABC}^{1/2}K\left(  \alpha\right)  \rho
_{ABC}^{1/2}\right]  ^{\alpha}\right\}  =\sup_{H\geq0}\alpha\text{Tr}\left\{
H\rho_{ABC}\right\}  -\left(  \alpha-1\right)  \text{Tr}\left\{  \left[
H^{1/2}L\left(  \alpha\right)  H^{1/2}\right]  ^{\alpha/\left(  \alpha
-1\right)  }\right\}  , \label{eq:sup-form}%
\end{equation}
where%
\begin{equation}
L\left(  \alpha\right)  \equiv\tau_{AC}^{\left(  \alpha-1\right)  /2\alpha
}\omega_{C}^{\left(  1-\alpha\right)  /2\alpha}\theta_{BC}^{\left(
\alpha-1\right)  /\alpha}\omega_{C}^{\left(  1-\alpha\right)  /2\alpha}%
\tau_{AC}^{\left(  \alpha-1\right)  /2\alpha},
\end{equation}
so that $\left[  L\left(  \alpha\right)  \right]  ^{-1}=K\left(
\alpha\right)  $. From the fact that the following map%
\begin{equation}
S\in\mathcal{B}(\mathcal{H})_{+}\mapsto\text{Tr}\left\{  \left[  T^{\dag}%
S^{p}T\right]  ^{1/p}\right\}
\end{equation}
is concave in $S$ for a fixed $T\in\mathcal{B}(\mathcal{H})$ and for $-1\leq
p\leq1$ \cite[Lemma~5]{FL13}\ and the representation formula given in
(\ref{eq:sup-form}), we can then conclude that the function in
(\ref{eq:trace-function-sw-a1}) is jointly convex in $\rho_{ABC}$ and
$\theta_{BC}$ for $\alpha\in\left(  1,\infty\right)  $.

So it remains to prove the representation formula in (\ref{eq:sup-form}).
Recall from the alternative proof of \cite[Lemma~4]{FL13}\ that for positive
semi-definite operators $X$ and $Y$ and $1<p,q<\infty$ with $1/p+1/q=1$, the
following inequality holds%
\begin{equation}
\text{Tr}\left\{  XY\right\}  \leq\frac{1}{p}\text{Tr}\left\{  X^{p}\right\}
+\frac{1}{q}\text{Tr}\left\{  Y^{q}\right\}  , \label{eq:holder-like-ineq}%
\end{equation}
with equality holding if $X^{p}=Y^{q}$. To apply the inequality in
(\ref{eq:holder-like-ineq}), we set%
\begin{align}
X  &  =K\left(  \alpha\right)  ^{1/2}\rho_{ABC}K\left(  \alpha\right)
^{1/2},\\
Y  &  =L\left(  \alpha\right)  ^{1/2}HL\left(  \alpha\right)  ^{1/2},\\
p  &  =\alpha,\\
q  &  =\frac{\alpha}{\alpha-1}.
\end{align}
Applying (\ref{eq:holder-like-ineq}), we find that%
\begin{equation}
\text{Tr}\left\{  H\rho_{ABC}\right\}  \leq\frac{1}{\alpha}\text{Tr}\left\{
\left[  \rho_{ABC}^{1/2}K\left(  \alpha\right)  \rho_{ABC}^{1/2}\right]
^{\alpha}\right\}  +\frac{\alpha-1}{\alpha}\text{Tr}\left\{  \left[
H^{1/2}L\left(  \alpha\right)  H^{1/2}\right]  ^{\alpha/\left(  \alpha
-1\right)  }\right\}  ,
\end{equation}
which can be rewritten as%
\begin{equation}
\alpha\text{Tr}\left\{  H\rho_{ABC}\right\}  -\left(  \alpha-1\right)
\text{Tr}\left\{  \left[  H^{1/2}L\left(  \alpha\right)  H^{1/2}\right]
^{\alpha/\left(  \alpha-1\right)  }\right\}  \leq\text{Tr}\left\{  \left[
\rho_{ABC}^{1/2}K\left(  \alpha\right)  \rho_{ABC}^{1/2}\right]  ^{\alpha
}\right\}  .
\end{equation}
From the equality condition $X^{p}=Y^{q}$, we can see that the optimal $H$
attaining equality is%
\begin{equation}
L\left(  \alpha\right)  ^{-1/2}\left[  K\left(  \alpha\right)  ^{1/2}%
\ \rho_{ABC}\ K\left(  \alpha\right)  ^{1/2}\right]  ^{\alpha-1}L\left(
\alpha\right)  ^{-1/2}.
\end{equation}
This proves the representation formula in (\ref{eq:sup-form}). A proof similar
to the above one demonstrates (\ref{eq:other-mono-sandwiched}%
)-(\ref{eq:other-mono-sandwiched-3}) for $\alpha\in\left(  1,\infty\right)  $.
\end{proof}

\begin{remark}
It is open to determine whether Lemma~\ref{lem:monotone-lemma-sandwiched}
applies to $\rho_{ABC}\in\mathcal{S}\left(  \mathcal{H}_{ABC}\right)  $,
$\tau_{AC}\in\mathcal{S}\left(  \mathcal{H}_{AC}\right)  $, $\theta_{BC}%
\in\mathcal{S}\left(  \mathcal{H}_{BC}\right)  $, and $\omega_{C}%
\in\mathcal{S}\left(  \mathcal{H}_{C}\right)  $. That is, it is not clear to
us whether Lemma~\ref{lem:monotone-lemma-sandwiched} can be extended by a
straightforward continuity argument as was the case in \cite[Proposition~3]%
{FL13}, due to the fact that $\widetilde{\Delta}_{\alpha}$ features many
non-commutative matrix multiplications which can interact in non-trivial ways.
\end{remark}

\begin{remark}
Let $\rho_{ABC}\in\mathcal{S}\left(  \mathcal{H}_{ABC}\right)  _{++}$,
$\tau_{AC}\in\mathcal{S}\left(  \mathcal{H}_{AC}\right)  _{++}$, $\theta
_{BC}\in\mathcal{S}\left(  \mathcal{H}_{BC}\right)  _{++}$, and $\omega_{C}%
\in\mathcal{S}\left(  \mathcal{H}_{C}\right)  _{++}$. It is an open question
to determine whether the $\widetilde{\Delta}_{\alpha}$ quantities defined from
\eqref{eq:Q_SW_a_1}, \eqref{eq:Q_SW_a_2}-\eqref{eq:Q_SW_a_6} are monotone
non-increasing with respect to quantum operations acting on either systems $A$
or $B$ for $\alpha\in\lbrack1/2,1)\cup(1,\infty)$. It is also an open question
to determine whether $\widetilde{I}_{\alpha}\left(  A;B|C\right)  _{\rho}$ is
monotone non-increasing with respect to local quantum operations acting on the
system $A$ for $\alpha\in\lbrack1/2,1)\cup(1,\infty)$.
\end{remark}

\begin{corollary}
\label{thm:SW-Renyi-CMI-monotone}Let $\rho_{ABC}\in\mathcal{S}\left(
\mathcal{H}_{ABC}\right)  _{++}$, $\tau_{AC}\in\mathcal{S}\left(
\mathcal{H}_{AC}\right)  _{++}$, $\theta_{BC}\in\mathcal{S}\left(
\mathcal{H}_{BC}\right)  _{++}$, and $\omega_{C}\in\mathcal{S}\left(
\mathcal{H}_{C}\right)  _{++}$. All sandwiched R\'{e}nyi generalizations of
the conditional mutual information derived from%
\begin{equation}
\widetilde{\Delta}_{\alpha}\left(  \rho_{ABC},\tau_{AC},\omega_{C},\theta
_{BC}\right)  ,\ \ \ \ \widetilde{\Delta}_{\alpha}\left(  \rho_{ABC}%
,\omega_{C},\tau_{AC},\theta_{BC}\right)  , \label{eq:mono-B-1-SW}%
\end{equation}
are monotone non-increasing with respect to quantum operations on system $B$,
for $\alpha\in\lbrack1/2,1)\cup(1,\infty)$. All sandwiched R\'{e}nyi
generalizations of the conditional mutual information derived from%
\begin{equation}
\widetilde{\Delta}_{\alpha}\left(  \rho_{ABC},\omega_{C},\theta_{BC},\tau
_{AC}\right)  ,\ \ \ \widetilde{\Delta}_{\alpha}\left(  \rho_{ABC},\theta
_{BC},\omega_{C},\tau_{AC}\right)  , \label{eq:mono-B-4-SW}%
\end{equation}
are monotone non-increasing with respect to quantum operations on system $A$,
for $\alpha\in\lbrack1/2,1)\cup(1,\infty)$.
\end{corollary}

\begin{proof}
The argument is exactly the same as that in the proof of
Corollary~\ref{thm:Renyi-CMI-monotone}.
\end{proof}

\begin{corollary}
\label{cor:SW-Renyi-CMI-positive}We can employ the monotonicity inequalities
from Lemma~\ref{lem:monotone-lemma} to conclude that some R\'{e}nyi
generalizations of the conditional mutual information derived from
\eqref{eq:mono-B-1-SW}-\eqref{eq:mono-B-4-SW} and
Proposition~\ref{thm:CMI-formulations} are non-negative for all $\alpha
\in\lbrack1/2,1)\cup(1,\infty)$. This includes $\widetilde{\Delta}_{\alpha
}\left(  \rho_{ABC},\rho_{AC},\rho_{C},\rho_{BC}\right)  $ and the one from \eqref{eq:sw-Renyi-CMI}.
\end{corollary}

\begin{proof}
The argument proceeds similarly to that in the proof of
Corollary~\ref{cor:Renyi-CMI-positive}.
\end{proof}

\begin{remark}
It is an open question to determine whether all sandwiched R\'{e}nyi
generalizations of the conditional mutual information designed from the
different optimizations in Proposition~\ref{thm:CMI-formulations}\ and the
different orderings in \eqref{eq:Q_SW_a_1},
\eqref{eq:Q_SW_a_2}-\eqref{eq:Q_SW_a_6} are non-negative for $\alpha\in
\lbrack1/2,1)\cup(1,\infty)$.
\end{remark}

\subsection{Max- and min-conditional mutual information}

Let $\rho_{ABC}\in\mathcal{S}\left(  \mathcal{H}_{ABC}\right)  _{++}$,
$\tau_{AC}\in\mathcal{S}\left(  \mathcal{H}_{AC}\right)  _{++}$, $\theta
_{BC}\in\mathcal{S}\left(  \mathcal{H}_{BC}\right)  _{++}$, and $\omega_{C}%
\in\mathcal{S}\left(  \mathcal{H}_{C}\right)  _{++}$. In this section, we
define a max- and min-conditional mutual information from the following two
core quantities:%
\begin{align}
\Delta_{\max}\left(  \rho_{ABC},\tau_{AC},\omega_{C},\theta_{BC}\right)   &
\equiv\log\left\Vert \rho_{ABC}^{1/2}\tau_{AC}^{-1/2}\omega_{C}^{1/2}%
\theta_{BC}^{-1}\omega_{C}^{1/2}\tau_{AC}^{-1/2}\rho_{ABC}^{1/2}\right\Vert
_{\infty}\\
&  =\inf\left\{  \lambda:\rho_{ABC}\leq\exp\left(  \lambda\right)  \tau
_{AC}^{1/2}\omega_{C}^{-1/2}\theta_{BC}\omega_{C}^{-1/2}\tau_{AC}%
^{1/2}\right\}  ,\\
\Delta_{\min}\left(  \rho_{ABC},\tau_{AC},\omega_{C},\theta_{BC}\right)   &
\equiv\widetilde{\Delta}_{1/2}\left(  \rho_{ABC},\tau_{AC},\omega_{C}%
,\theta_{BC}\right) \\
&  =-\log\left\Vert \sqrt{\rho_{ABC}}\sqrt{\tau_{AC}^{1/2}\omega_{C}%
^{-1/2}\theta_{BC}\omega_{C}^{-1/2}\tau_{AC}^{1/2}}\right\Vert _{1}^{2}\\
&  =-\log F\left(  \rho_{ABC},\tau_{AC}^{1/2}\omega_{C}^{-1/2}\theta
_{BC}\omega_{C}^{-1/2}\tau_{AC}^{1/2}\right)  .
\end{align}
Also, the fidelity between $P\in\mathcal{B}\left(  \mathcal{H}\right)  _{+}$
and $Q\in\mathcal{B}\left(  \mathcal{H}\right)  _{+}$ is defined as $F\left(
P,Q\right)  \equiv\Vert\sqrt{P}\sqrt{Q}\Vert_{1}^{2}$. These quantities are
inspired by the max-relative entropy from \cite{D09}, defined as%
\begin{equation}
D_{\max}\left(  \rho\Vert\sigma\right)  \equiv\inf\left\{  \lambda:\rho
\leq\exp\left(  \lambda\right)  \sigma\right\}  ,
\end{equation}
when $\operatorname{supp}\left(  \rho\right)  \subseteq\operatorname{supp}%
\left(  \sigma\right)  $ and $+\infty$ otherwise, and the min-relative entropy
from \cite{KRS09}, defined as%
\begin{equation}
D_{\min}\left(  \rho\Vert\sigma\right)  \equiv\widetilde{D}_{1/2}\left(
\rho\Vert\sigma\right)  =-\log F\left(  \rho,\sigma\right)  .
\end{equation}

We first state a generalization of the result that $\lim_{\alpha
\rightarrow\infty}\widetilde{D}_{\alpha}\left(  \rho\Vert\sigma\right)
=D_{\max}\left(  \rho\Vert\sigma\right)  $ \cite[Theorem~5]{MDSFT13}:

\begin{proposition}
\label{prop:convergence-Delta_max}Let $\rho_{ABC}\in\mathcal{S}\left(
\mathcal{H}_{ABC}\right)  _{++}$, $\tau_{AC}\in\mathcal{S}\left(
\mathcal{H}_{AC}\right)  _{++}$, $\theta_{BC}\in\mathcal{S}\left(
\mathcal{H}_{BC}\right)  _{++}$, and $\omega_{C}\in\mathcal{S}\left(
\mathcal{H}_{C}\right)  _{++}$. Then%
\begin{equation}
\lim_{\alpha\rightarrow\infty}\widetilde{\Delta}_{\alpha}\left(  \rho
_{ABC},\tau_{AC},\omega_{C},\theta_{BC}\right)  =\Delta_{\max}\left(
\rho_{ABC},\tau_{AC},\omega_{C},\theta_{BC}\right)  .
\end{equation}

\end{proposition}

The idea for the proof is the same as that for the proof of \cite[Theorem~5]%
{MDSFT13}, and we provide it in Appendix~\ref{sec:convergence-Delta_max}.
Next, we turn to monotonicity of $\Delta_{\max}$ under local quantum operations:

\begin{proposition}
Let $\rho_{ABC}\in\mathcal{S}\left(  \mathcal{H}_{ABC}\right)  _{++}$,
$\tau_{AC}\in\mathcal{S}\left(  \mathcal{H}_{AC}\right)  _{++}$, $\theta
_{BC}\in\mathcal{S}\left(  \mathcal{H}_{BC}\right)  _{++}$, and $\omega_{C}%
\in\mathcal{S}\left(  \mathcal{H}_{C}\right)  _{++}$. Let $\mathcal{N}%
_{A\rightarrow A^{\prime}}$ and $\mathcal{M}_{B\rightarrow B^{\prime}}%
$\ denote local quantum operations acting on systems $A$ and $B$,
respectively. Then the following monotonicity inequalities hold:%
\begin{align}
\Delta_{\max}\left(  \rho_{ABC},\tau_{AC},\omega_{C},\theta_{BC}\right)   &
\geq\Delta_{\max}\left(  \mathcal{M}_{B\rightarrow B^{\prime}}\left(
\rho_{ABC}\right)  ,\tau_{AC},\omega_{C},\mathcal{M}_{B\rightarrow B^{\prime}%
}\left(  \theta_{BC}\right)  \right)  ,\label{eq:Delta_max_mono}\\
\Delta_{\max}\left(  \rho_{ABC},\omega_{C},\tau_{AC},\theta_{BC}\right)   &
\geq\Delta_{\max}\left(  \mathcal{M}_{B\rightarrow B^{\prime}}\left(
\rho_{ABC}\right)  ,\omega_{C},\tau_{AC},\mathcal{M}_{B\rightarrow B^{\prime}%
}\left(  \theta_{BC}\right)  \right)  ,\label{eq:Delta_max_mono-2}\\
\Delta_{\max}\left(  \rho_{ABC},\omega_{C},\theta_{BC},\tau_{AC}\right)   &
\geq\Delta_{\max}\left(  \mathcal{N}_{A\rightarrow A^{\prime}}\left(
\rho_{ABC}\right)  ,\omega_{C},\theta_{BC},\mathcal{N}_{A\rightarrow
A^{\prime}}\left(  \tau_{AC}\right)  \right)  ,\\
\Delta_{\max}\left(  \rho_{ABC},\theta_{BC},\omega_{C},\tau_{AC}\right)   &
\geq\Delta_{\max}\left(  \mathcal{N}_{A\rightarrow A^{\prime}}\left(
\rho_{ABC}\right)  ,\theta_{BC},\omega_{C},\mathcal{N}_{A\rightarrow
A^{\prime}}\left(  \tau_{AC}\right)  \right)  . \label{eq:Delta_max_mono-4}%
\end{align}

\end{proposition}

\begin{proof}
We begin by establishing\ (\ref{eq:Delta_max_mono}). Let $\lambda^{\ast
}=\Delta_{\max}\left(  \rho_{ABC},\tau_{AC},\omega_{C},\theta_{BC}\right)  $,
so that%
\begin{equation}
\rho_{ABC}\leq\exp\left(  \lambda^{\ast}\right)  \tau_{AC}^{1/2}\omega
_{C}^{-1/2}\theta_{BC}\omega_{C}^{-1/2}\tau_{AC}^{1/2}%
.\label{eq:start-op-ineq-mono}%
\end{equation}
For any CPTP\ map $\mathcal{M}_{B\rightarrow B^{\prime}}$, the inequality in
(\ref{eq:start-op-ineq-mono}) implies the following operator inequality%
\begin{equation}
\mathcal{M}_{B\rightarrow B^{\prime}}\left(  \rho_{ABC}\right)  \leq
\exp\left(  \lambda^{\ast}\right)  \tau_{AC}^{1/2}\omega_{C}^{-1/2}%
\mathcal{M}_{B\rightarrow B^{\prime}}\left(  \theta_{BC}\right)  \omega
_{C}^{-1/2}\tau_{AC}^{1/2}.
\end{equation}
From the definition of $\Delta_{\max}$, we can conclude that%
\begin{equation}
\lambda^{\ast}\geq\Delta_{\max}\left(  \mathcal{M}_{B\rightarrow B^{\prime}%
}\left(  \rho_{ABC}\right)  ,\tau_{AC},\omega_{C},\mathcal{M}_{B\rightarrow
B^{\prime}}\left(  \theta_{BC}\right)  \right)  ,
\end{equation}
which is equivalent to (\ref{eq:Delta_max_mono}). The inequalities in
(\ref{eq:Delta_max_mono-2})-(\ref{eq:Delta_max_mono-4}) follow from a similar
line of reasoning.
\end{proof}

We define a max-conditional mutual information as follows:%
\begin{equation}
I_{\max}\left(  A;B|C\right)  _{\rho|\rho}\equiv\Delta_{\max}\left(
\rho_{ABC},\rho_{AC},\rho_{C},\rho_{BC}\right)  . \label{eq:Imax}%
\end{equation}
This generalizes the max-mutual information, defined in \cite{BCR09}, and its
variations~\cite{CBR14}. We define a min-conditional mutual information as
follows:%
\begin{equation}
I_{\min}\left(  A;B|C\right)  _{\rho|\rho}\equiv\Delta_{\min}\left(
\rho_{ABC},\rho_{AC},\rho_{C},\rho_{BC}\right)  . \label{eq:Imin}%
\end{equation}
The forms given above seem quite natural, as the operators $\rho_{AC}%
^{1/2}\rho_{C}^{-1/2}\rho_{BC}\rho_{C}^{-1/2}\rho_{AC}^{1/2}$ appear in our
review of quantum Markov states in Section~\ref{sec:notation} (however, note
again that this operator is not a Markov state unless $\rho_{ABC}=\rho
_{AC}^{1/2}\rho_{C}^{-1/2}\rho_{BC}\rho_{C}^{-1/2}\rho_{AC}^{1/2}$). Note that
other min- and max-conditional mutual information quantities are possible by
considering the other orderings and optimizations for the last three arguments
to $\Delta_{\max}$ and $\Delta_{\min}$, but it is our impression that the
above choice is natural.

\section{Duality}

\label{sec:duality}A fundamental property of the conditional mutual
information is a duality relation:\ For a four-party pure state $\psi_{ABCD}$,
the following equality holds%
\begin{equation}
I\left(  A;B|C\right)  _{\psi}=I\left(  A;B|D\right)  _{\psi}.
\label{eq:duality}%
\end{equation}
This can easily be verified by considering Schmidt decompositions of
$\psi_{ABCD}$ for the different possible bipartite cuts of $ABCD$ (see
\cite{DY08,YD09} for an operational interpretation of this duality in terms of
the state redistribution protocol). Furthermore, since the conditional mutual
information is symmetric under the exchange of $A$ and $B$, we have the
following equalities:%
\begin{equation}
I\left(  B;A|C\right)  _{\psi}=I\left(  A;B|C\right)  _{\psi}=I\left(
A;B|D\right)  _{\psi}=I\left(  B;A|D\right)  _{\psi}.
\end{equation}

In this section, we prove that the R\'{e}nyi conditional mutual information in
Definition~\ref{def:Renyi-CMI}\ and the sandwiched quantity in
Definition~\ref{def:renyicmi_sand} obey a duality relation of the above form.
However, note that other (but not all) variations satisfy duality as well. In
order to prove these results, we make use of the following standard lemma:

\begin{lemma}
\label{lem:dual_simple}For any bipartite pure state $\psi_{AB}$, any Hermitian
operator $M_{A}$ acting on system $A$, and the maximally entangled vector
$\left\vert \Gamma\right\rangle _{AB}\equiv\sum_{j}\left\vert j\right\rangle
_{A}\left\vert j\right\rangle _{B}$ (with $\left\{  \left\vert j\right\rangle
_{A}\right\}  $ and $\left\{  \left\vert j\right\rangle _{B}\right\}  $
orthonormal bases), we have that%
\begin{align}
\left(  M_{A}\otimes I_{B}\right)  \left\vert \Gamma\right\rangle _{AB}  &
=\left(  I_{A}\otimes M_{B}^{T}\right)  \left\vert \Gamma\right\rangle
_{AB},\label{eq:mes_trick}\\
\psi_{A}\left\vert \psi\right\rangle _{AB}  &  =\psi_{B}\left\vert
\psi\right\rangle _{AB},\label{eq:for-duality}\\
\langle\psi|M_{A}\otimes I_{B}|\psi\rangle_{AB}  &  =\langle\psi|I_{A}\otimes
M_{B}^{T}|\psi\rangle_{AB}, \label{eq:transpose-trick-var}%
\end{align}
where the transpose is with respect to the Schmidt basis.
\end{lemma}

\begin{theorem}
\label{thm:duality-renyi}The following duality relation holds for all
$\alpha\in(0,1)\cup(1,\infty)$ for a pure four-party state $\psi_{ABCD}$:%
\begin{equation}
I_{\alpha}\left(  A;B|C\right)  _{\psi}=I_{\alpha}\left(  B;A|D\right)
_{\psi}.
\end{equation}

\end{theorem}

\begin{proof}
Our proof exploits ideas used in the proof of \cite[Lemma~6]{TCR09} and
\cite[Theorem~2]{TBH13}. We know from Proposition~\ref{prop:sibson}\ that%
\begin{align}
I_{\alpha}\left(  A;B|C\right)  _{\psi}  &  =\frac{\alpha}{\alpha-1}%
\log\text{Tr}\left\{  \left(  \text{Tr}_{A}\left\{  \psi_{C}^{\left(
\alpha-1\right)  /2}\psi_{AC}^{\left(  \alpha-1\right)  /2}\psi_{ABC}^{\alpha
}\psi_{AC}^{\left(  \alpha-1\right)  /2}\psi_{C}^{\left(  \alpha-1\right)
/2}\right\}  \right)  ^{1/\alpha}\right\}  ,\\
I_{\alpha}\left(  B;A|D\right)  _{\psi}  &  =\frac{\alpha}{\alpha-1}%
\log\text{Tr}\left\{  \left(  \text{Tr}_{B}\left\{  \psi_{D}^{\left(
\alpha-1\right)  /2}\psi_{BD}^{\left(  \alpha-1\right)  /2}\psi_{ABD}^{\alpha
}\psi_{BD}^{\left(  \alpha-1\right)  /2}\psi_{D}^{\left(  \alpha-1\right)
/2}\right\}  \right)  ^{1/\alpha}\right\}  .
\end{align}
Thus, we will have proved the theorem if we can show that the eigenvalues of%
\begin{equation}
\text{Tr}_{A}\left\{  \psi_{C}^{\left(  \alpha-1\right)  /2}\psi_{AC}^{\left(
\alpha-1\right)  /2}\psi_{ABC}^{\alpha}\psi_{AC}^{\left(  \alpha-1\right)
/2}\psi_{C}^{\left(  \alpha-1\right)  /2}\right\}
\end{equation}
and%
\begin{equation}
\text{Tr}_{B}\left\{  \psi_{D}^{\left(  \alpha-1\right)  /2}\psi_{BD}^{\left(
\alpha-1\right)  /2}\psi_{ABD}^{\alpha}\psi_{BD}^{\left(  \alpha-1\right)
/2}\psi_{D}^{\left(  \alpha-1\right)  /2}\right\}
\end{equation}
are the same. To show this, consider that%
\begin{align}
&  \text{Tr}_{A}\left\{  \psi_{C}^{\left(  \alpha-1\right)  /2}\psi
_{AC}^{\left(  \alpha-1\right)  /2}\psi_{ABC}^{\alpha}\psi_{AC}^{\left(
\alpha-1\right)  /2}\psi_{C}^{\left(  \alpha-1\right)  /2}\right\} \nonumber\\
&  =\text{Tr}_{A}\left\{  \psi_{C}^{\left(  \alpha-1\right)  /2}\psi
_{AC}^{\left(  \alpha-1\right)  /2}\psi_{ABC}^{\left(  \alpha-1\right)
/2}\psi_{ABC}\psi_{ABC}^{\left(  \alpha-1\right)  /2}\psi_{AC}^{\left(
\alpha-1\right)  /2}\psi_{C}^{\left(  \alpha-1\right)  /2}\right\} \\
&  =\text{Tr}_{AD}\left\{  \psi_{C}^{\left(  \alpha-1\right)  /2}\psi
_{AC}^{\left(  \alpha-1\right)  /2}\psi_{ABC}^{\left(  \alpha-1\right)
/2}\psi_{ABCD}\psi_{ABC}^{\left(  \alpha-1\right)  /2}\psi_{AC}^{\left(
\alpha-1\right)  /2}\psi_{C}^{\left(  \alpha-1\right)  /2}\right\}  .
\end{align}
The eigenvalues of the operator in the last line are the same as those of the
operator in the first line of what follows (from the Schmidt decomposition):%
\begin{align}
&  \text{Tr}_{BC}\left\{  \psi_{C}^{\left(  \alpha-1\right)  /2}\psi
_{AC}^{\left(  \alpha-1\right)  /2}\psi_{ABC}^{\left(  \alpha-1\right)
/2}\psi_{ABCD}\psi_{ABC}^{\left(  \alpha-1\right)  /2}\psi_{AC}^{\left(
\alpha-1\right)  /2}\psi_{C}^{\left(  \alpha-1\right)  /2}\right\} \nonumber\\
&  =\text{Tr}_{BC}\left\{  \psi_{C}^{\left(  \alpha-1\right)  /2}\psi
_{AC}^{\left(  \alpha-1\right)  /2}\psi_{D}^{\left(  \alpha-1\right)  /2}%
\psi_{ABCD}\psi_{D}^{\left(  \alpha-1\right)  /2}\psi_{AC}^{\left(
\alpha-1\right)  /2}\psi_{C}^{\left(  \alpha-1\right)  /2}\right\} \\
&  =\text{Tr}_{BC}\left\{  \psi_{D}^{\left(  \alpha-1\right)  /2}\psi
_{C}^{\left(  \alpha-1\right)  /2}\psi_{AC}^{\left(  \alpha-1\right)  /2}%
\psi_{ABCD}\psi_{AC}^{\left(  \alpha-1\right)  /2}\psi_{C}^{\left(
\alpha-1\right)  /2}\psi_{D}^{\left(  \alpha-1\right)  /2}\right\} \\
&  =\text{Tr}_{BC}\left\{  \psi_{D}^{\left(  \alpha-1\right)  /2}\psi
_{C}^{\left(  \alpha-1\right)  /2}\psi_{BD}^{\left(  \alpha-1\right)  /2}%
\psi_{ABCD}\psi_{BD}^{\left(  \alpha-1\right)  /2}\psi_{C}^{\left(
\alpha-1\right)  /2}\psi_{D}^{\left(  \alpha-1\right)  /2}\right\} \\
&  =\text{Tr}_{BC}\left\{  \psi_{D}^{\left(  \alpha-1\right)  /2}\psi
_{BD}^{\left(  \alpha-1\right)  /2}\psi_{C}^{\left(  \alpha-1\right)  /2}%
\psi_{ABCD}\psi_{C}^{\left(  \alpha-1\right)  /2}\psi_{BD}^{\left(
\alpha-1\right)  /2}\psi_{D}^{\left(  \alpha-1\right)  /2}\right\} \\
&  =\text{Tr}_{BC}\left\{  \psi_{D}^{\left(  \alpha-1\right)  /2}\psi
_{BD}^{\left(  \alpha-1\right)  /2}\psi_{ABD}^{\left(  \alpha-1\right)
/2}\psi_{ABCD}\psi_{ABD}^{\left(  \alpha-1\right)  /2}\psi_{BD}^{\left(
\alpha-1\right)  /2}\psi_{D}^{\left(  \alpha-1\right)  /2}\right\} \\
&  =\text{Tr}_{B}\left\{  \psi_{D}^{\left(  \alpha-1\right)  /2}\psi
_{BD}^{\left(  \alpha-1\right)  /2}\psi_{ABD}^{\left(  \alpha-1\right)
/2}\psi_{ABD}\psi_{ABD}^{\left(  \alpha-1\right)  /2}\psi_{BD}^{\left(
\alpha-1\right)  /2}\psi_{D}^{\left(  \alpha-1\right)  /2}\right\} \\
&  =\text{Tr}_{B}\left\{  \psi_{D}^{\left(  \alpha-1\right)  /2}\psi
_{BD}^{\left(  \alpha-1\right)  /2}\psi_{ABD}^{\alpha}\psi_{BD}^{\left(
\alpha-1\right)  /2}\psi_{D}^{\left(  \alpha-1\right)  /2}\right\}  .
\end{align}
In the above, we have applied (\ref{eq:for-duality}) several times.
\end{proof}

\begin{theorem}
The following duality relation holds for all $\alpha\in(0,1)\cup(1,\infty)$
for a pure four-party state $\psi_{ABCD}$:
\begin{equation}
\widetilde{I}_{\alpha}\left(  A;B|C\right)  _{\psi}=\widetilde{I}_{\alpha
}\left(  B;A|D\right)  _{\psi}.
\end{equation}

\end{theorem}

\begin{proof}
Our proof uses ideas similar to those in the proof of~\cite[Theorem
10]{MDSFT13}. We start by considering the case $\alpha>1$. We recall that it
is possible to express the $\alpha$-norm with its dual norm (see, e.g.,
\cite[Lemma 12]{MDSFT13}):
\begin{multline}
\inf_{\sigma_{BC}}\sup_{\omega_{C}}\left\Vert \psi_{ABC}^{1/2}\psi
_{AC}^{\left(  1-\alpha\right)  /2\alpha}\omega_{C}^{\left(  \alpha-1\right)
/2\alpha}\sigma_{BC}^{\left(  1-\alpha\right)  /\alpha}\omega_{C}^{\left(
\alpha-1\right)  /2\alpha}\psi_{AC}^{\left(  1-\alpha\right)  /2\alpha}%
\psi_{ABC}^{1/2}\right\Vert _{\alpha}=\\
\inf_{\sigma_{BC}}\sup_{\omega_{C}}\sup_{\tau_{ABC}}\text{Tr}\left\{
\psi_{ABC}^{1/2}\psi_{AC}^{\left(  1-\alpha\right)  /2\alpha}\omega
_{C}^{\left(  \alpha-1\right)  /2\alpha}\sigma_{BC}^{\left(  1-\alpha\right)
/\alpha}\omega_{C}^{\left(  \alpha-1\right)  /2\alpha}\psi_{AC}^{\left(
1-\alpha\right)  /2\alpha}\psi_{ABC}^{1/2}\tau_{ABC}^{\left(  \alpha-1\right)
/\alpha}\right\}  .
\end{multline}
So it suffices to prove the following relation:%
\begin{multline}
\inf_{\sigma_{BC}}\sup_{\omega_{C}}\sup_{\tau_{ABC}}\text{Tr}\left\{
\psi_{ABC}^{1/2}\psi_{AC}^{\left(  1-\alpha\right)  /2\alpha}\omega
_{C}^{\left(  \alpha-1\right)  /2\alpha}\sigma_{BC}^{\left(  1-\alpha\right)
/\alpha}\omega_{C}^{\left(  \alpha-1\right)  /2\alpha}\psi_{AC}^{\left(
1-\alpha\right)  /2\alpha}\psi_{ABC}^{1/2}\tau_{ABC}^{\left(  \alpha-1\right)
/\alpha}\right\}  =\label{eq:duality-eq-to-prove}\\
\inf_{\sigma_{AD}}\sup_{\tau_{D}}\sup_{\omega_{ABD}}\text{Tr}\left\{
\psi_{ABD}^{1/2}\psi_{BD}^{\left(  1-\alpha\right)  /2\alpha}\tau_{D}^{\left(
\alpha-1\right)  /2\alpha}\sigma_{AD}^{\left(  1-\alpha\right)  /\alpha}%
\tau_{D}^{\left(  \alpha-1\right)  /2\alpha}\psi_{AC}^{\left(  1-\alpha
\right)  /2\alpha}\psi_{ABD}^{1/2}\omega_{ABD}^{\left(  \alpha-1\right)
/\alpha}\right\}  ,
\end{multline}
because%
\begin{multline}
\widetilde{I}_{\alpha}\left(  B;A|D\right)  _{\psi}\\
=\inf_{\sigma_{AD}}\sup_{\tau_{D}}\sup_{\omega_{ABD}}\frac{\alpha}{\alpha
-1}\log\text{Tr}\left\{  \psi_{ABD}^{1/2}\psi_{BD}^{\left(  1-\alpha\right)
/2\alpha}\tau_{D}^{\left(  \alpha-1\right)  /2\alpha}\sigma_{AD}^{\left(
1-\alpha\right)  /\alpha}\tau_{D}^{\left(  \alpha-1\right)  /2\alpha}\psi
_{AC}^{\left(  1-\alpha\right)  /2\alpha}\psi_{ABD}^{1/2}\omega_{ABD}^{\left(
\alpha-1\right)  /\alpha}\right\}  .
\end{multline}
Indeed, we will prove that%
\begin{multline}
\text{Tr}\left\{  \psi_{ABC}^{1/2}\psi_{AC}^{\left(  1-\alpha\right)
/2\alpha}\omega_{C}^{\left(  \alpha-1\right)  /2\alpha}\sigma_{BC}^{\left(
1-\alpha\right)  /\alpha}\omega_{C}^{\left(  \alpha-1\right)  /2\alpha}%
\psi_{AC}^{\left(  1-\alpha\right)  /2\alpha}\psi_{ABC}^{1/2}\tau
_{ABC}^{\left(  \alpha-1\right)  /\alpha}\right\} \\
=\text{Tr}\left\{  \psi_{ABD}^{1/2}\psi_{BD}^{\left(  1-\alpha\right)
/2\alpha}\left(  \tau_{D}^{T}\right)  ^{\left(  \alpha-1\right)  /2\alpha
}\left(  \sigma_{AD}^{T}\right)  ^{\left(  1-\alpha\right)  /\alpha}\left(
\tau_{D}^{T}\right)  ^{\left(  \alpha-1\right)  /2\alpha}\psi_{BD}^{\left(
1-\alpha\right)  /2\alpha}\psi_{ABD}^{1/2}\left(  \omega_{ABD}^{T}\right)
^{\left(  \alpha-1\right)  /\alpha}\right\}  ,
\end{multline}
from which one can conclude (\ref{eq:duality-eq-to-prove}), which has the optimizations.

Proceeding, we observe that
\begin{align}
&  \text{Tr}\left\{  \psi_{ABC}^{1/2}\psi_{AC}^{\left(  1-\alpha\right)
/2\alpha}\omega_{C}^{\left(  \alpha-1\right)  /2\alpha}\sigma_{BC}^{\left(
1-\alpha\right)  /\alpha}\omega_{C}^{\left(  \alpha-1\right)  /2\alpha}%
\psi_{AC}^{\left(  1-\alpha\right)  /2\alpha}\psi_{ABC}^{1/2}\tau
_{ABC}^{\left(  \alpha-1\right)  /\alpha}\right\} \nonumber\\
&  =\left\langle \Gamma\right\vert \psi_{ABC}^{1/2}\psi_{AC}^{\left(
1-\alpha\right)  /2\alpha}\omega_{C}^{\left(  \alpha-1\right)  /2\alpha}%
\sigma_{BC}^{\left(  1-\alpha\right)  /\alpha}\omega_{C}^{\left(
\alpha-1\right)  /2\alpha}\psi_{AC}^{\left(  1-\alpha\right)  /2\alpha}%
\psi_{ABC}^{1/2}\tau_{ABC}^{\left(  \alpha-1\right)  /\alpha}\left\vert
\Gamma\right\rangle _{ABC|D}\\
&  =\left\langle \Gamma\right\vert \psi_{ABC}^{1/2}\psi_{AC}^{\left(
1-\alpha\right)  /2\alpha}\omega_{C}^{\left(  \alpha-1\right)  /2\alpha}%
\sigma_{BC}^{\left(  1-\alpha\right)  /\alpha}\omega_{C}^{\left(
\alpha-1\right)  /2\alpha}\psi_{AC}^{\left(  1-\alpha\right)  /2\alpha}%
\psi_{ABC}^{1/2}\left(  \tau_{D}^{T}\right)  ^{\left(  \alpha-1\right)
/\alpha}\left\vert \Gamma\right\rangle _{ABC|D}\\
&  =\left\langle \psi\right\vert \psi_{AC}^{\left(  1-\alpha\right)  /2\alpha
}\omega_{C}^{\left(  \alpha-1\right)  /2\alpha}\left(  \tau_{D}^{T}\right)
^{\left(  \alpha-1\right)  /2\alpha}\sigma_{BC}^{\left(  1-\alpha\right)
/\alpha}\left(  \tau_{D}^{T}\right)  ^{\left(  \alpha-1\right)  /2\alpha
}\omega_{C}^{\left(  \alpha-1\right)  /2\alpha}\psi_{AC}^{\left(
1-\alpha\right)  /2\alpha}\left\vert \psi\right\rangle _{ABCD}\\
&  =\left\langle \psi\right\vert \psi_{BD}^{\left(  1-\alpha\right)  /2\alpha
}\omega_{C}^{\left(  \alpha-1\right)  /2\alpha}\left(  \tau_{D}^{T}\right)
^{\left(  \alpha-1\right)  /2\alpha}\sigma_{BC}^{\left(  1-\alpha\right)
/\alpha}\left(  \tau_{D}^{T}\right)  ^{\left(  \alpha-1\right)  /2\alpha
}\omega_{C}^{\left(  \alpha-1\right)  /2\alpha}\psi_{BD}^{\left(
1-\alpha\right)  /2\alpha}\left\vert \psi\right\rangle _{ABCD}\\
&  =\left\langle \psi\right\vert \omega_{C}^{\left(  \alpha-1\right)
/2\alpha}\psi_{BD}^{\left(  1-\alpha\right)  /2\alpha}\left(  \tau_{D}%
^{T}\right)  ^{\left(  \alpha-1\right)  /2\alpha}\sigma_{BC}^{\left(
1-\alpha\right)  /\alpha}\left(  \tau_{D}^{T}\right)  ^{\left(  \alpha
-1\right)  /2\alpha}\psi_{BD}^{\left(  1-\alpha\right)  /2\alpha}\omega
_{C}^{\left(  \alpha-1\right)  /2\alpha}\left\vert \psi\right\rangle _{ABCD}\\
&  =\left\langle \Gamma\right\vert \psi_{ABD}^{1/2}\omega_{C}^{\left(
\alpha-1\right)  /2\alpha}\psi_{BD}^{\frac{1-\alpha}{2\alpha}}\left(  \tau
_{D}^{T}\right)  ^{\frac{\alpha-1}{2\alpha}}\sigma_{BC}^{\frac{1-\alpha
}{\alpha}}\left(  \tau_{D}^{T}\right)  ^{\frac{\alpha-1}{2\alpha}}\psi
_{BD}^{\frac{1-\alpha}{2\alpha}}\omega_{C}^{\left(  \alpha-1\right)  /2\alpha
}\psi_{ABD}^{1/2}\left\vert \Gamma\right\rangle _{ABD|C}\\
&  =\left\langle \Gamma\right\vert \omega_{C}^{\left(  \alpha-1\right)
/2\alpha}\psi_{ABD}^{1/2}\psi_{BD}^{\frac{1-\alpha}{2\alpha}}\left(  \tau
_{D}^{T}\right)  ^{\frac{\alpha-1}{2\alpha}}\sigma_{BC}^{\frac{1-\alpha
}{\alpha}}\left(  \tau_{D}^{T}\right)  ^{\frac{\alpha-1}{2\alpha}}\psi
_{BD}^{\frac{1-\alpha}{2\alpha}}\psi_{ABD}^{1/2}\omega_{C}^{\left(
\alpha-1\right)  /2\alpha}\left\vert \Gamma\right\rangle _{ABD|C}\\
&  =\left\langle \Gamma\right\vert \left(  \omega_{ABD}^{T}\right)
^{\frac{\alpha-1}{2\alpha}}\psi_{ABD}^{1/2}\psi_{BD}^{\frac{1-\alpha}{2\alpha
}}\left(  \tau_{D}^{T}\right)  ^{\frac{\alpha-1}{2\alpha}}\sigma_{BC}%
^{\frac{1-\alpha}{\alpha}}\left(  \tau_{D}^{T}\right)  ^{\frac{\alpha
-1}{2\alpha}}\psi_{BD}^{\frac{1-\alpha}{2\alpha}}\psi_{ABD}^{1/2}\left(
\omega_{ABD}^{T}\right)  ^{\frac{\alpha-1}{2\alpha}}\left\vert \Gamma
\right\rangle _{ABD|C},
\end{align}
where we used the standard transpose trick~\eqref{eq:mes_trick} for the
maximally entangled vector $\left\vert \Gamma\right\rangle _{ABD|C}$ and the
first identity from Lemma~\ref{lem:dual_simple}. For the vector
\begin{equation}
\left\vert \varphi\right\rangle _{ABCD}\equiv\left(  \tau_{D}^{T}\right)
^{\left(  \alpha-1\right)  /2\alpha}\psi_{BD}^{\left(  1-\alpha\right)
/2\alpha}\psi_{ABD}^{1/2}\left(  \omega_{ABD}^{T}\right)  ^{\left(
\alpha-1\right)  /2\alpha}\left\vert \Gamma\right\rangle _{ABD|C},
\end{equation}
we get from the second identity in Lemma~\ref{lem:dual_simple} that
\begin{align}
&  \left\langle \Gamma\right\vert \left(  \omega_{ABD}^{T}\right)
^{\frac{\alpha-1}{2\alpha}}\psi_{ABD}^{1/2}\psi_{BD}^{\frac{1-\alpha}{2\alpha
}}\left(  \tau_{D}^{T}\right)  ^{\frac{\alpha-1}{2\alpha}}\sigma_{BC}%
^{\frac{1-\alpha}{\alpha}}\left(  \tau_{D}^{T}\right)  ^{\frac{\alpha
-1}{2\alpha}}\psi_{BD}^{\frac{1-\alpha}{2\alpha}}\psi_{ABD}^{1/2}\left(
\omega_{ABD}^{T}\right)  ^{\frac{\alpha-1}{2\alpha}}\left\vert \Gamma
\right\rangle _{ABD|C}\nonumber\\
&  =\left\langle \varphi\right\vert \sigma_{BC}^{\frac{1-\alpha}{\alpha}%
}\left\vert \varphi\right\rangle _{ABCD}\\
&  =\left\langle \varphi\right\vert (\sigma_{AD}^{T})^{\frac{1-\alpha}{\alpha
}}\left\vert \varphi\right\rangle _{ABCD}\\
&  =\left\langle \Gamma\right\vert \left(  \omega_{ABD}^{T}\right)
^{\frac{\alpha-1}{2\alpha}}\psi_{ABD}^{1/2}\psi_{BD}^{\frac{1-\alpha}{2\alpha
}}\left(  \tau_{D}^{T}\right)  ^{\frac{\alpha-1}{2\alpha}}(\sigma_{AD}%
^{T})^{\frac{1-\alpha}{\alpha}}\left(  \tau_{D}^{T}\right)  ^{\frac{\alpha
-1}{2\alpha}}\psi_{BD}^{\frac{1-\alpha}{2\alpha}}\psi_{ABD}^{1/2}\left(
\omega_{ABD}^{T}\right)  ^{\frac{\alpha-1}{2\alpha}}\left\vert \Gamma
\right\rangle _{ABD|C}\\
&  =\text{Tr}\left\{  \left(  \omega_{ABD}^{T}\right)  ^{\frac{\alpha
-1}{2\alpha}}\psi_{ABD}^{1/2}\psi_{BD}^{\frac{1-\alpha}{2\alpha}}\left(
\tau_{D}^{T}\right)  ^{\frac{\alpha-1}{2\alpha}}\left(  \sigma_{AD}%
^{T}\right)  ^{(1-\alpha)/\alpha}\left(  \tau_{D}^{T}\right)  ^{\frac
{\alpha-1}{2\alpha}}\psi_{BD}^{\frac{1-\alpha}{2\alpha}}\psi_{ABD}%
^{1/2}\left(  \omega_{ABD}^{T}\right)  ^{\frac{\alpha-1}{2\alpha}}\right\} \\
&  =\text{Tr}\left\{  \psi_{ABD}^{1/2}\psi_{BD}^{\left(  1-\alpha\right)
/2\alpha}\left(  \tau_{D}^{T}\right)  ^{\left(  \alpha-1\right)  /2\alpha
}\left(  \sigma_{AD}^{T}\right)  ^{\left(  1-\alpha\right)  /\alpha}\left(
\tau_{D}^{T}\right)  ^{\left(  \alpha-1\right)  /2\alpha}\psi_{BD}^{\left(
1-\alpha\right)  /2\alpha}\psi_{ABD}^{1/2}\left(  \omega_{ABD}^{T}\right)
^{\left(  \alpha-1\right)  /\alpha}\right\}  .
\end{align}
For the case $\alpha\in(0,1)$ the proof is similar, where we also use
\cite[Lemma 12]{MDSFT13}. We omit the details for this case.
\end{proof}

\section{Monotonicity in $\alpha$}

\label{sec:monotonicity-in-alpha}From numerical evidence and proofs for some
special cases, we think it is natural to put forward the following conjecture:

\begin{conjecture}
\label{conj:monotone-alpha}Let $\rho_{ABC}\in\mathcal{S}(\mathcal{H}%
_{ABC})_{++}$, $\tau_{AC}\in\mathcal{S}(\mathcal{H}_{AC})_{++}$, $\theta
_{BC}\in\mathcal{S}(\mathcal{H}_{BC})_{++}$, and $\omega_{C}\in\mathcal{S}%
(\mathcal{H}_{C})_{++}$. Then all of the R\'{e}nyi core quantities
$\Delta_{\alpha}$ and $\widetilde{\Delta}_{\alpha}$ derived from
\eqref{eq:Q_a_1}, \eqref{eq:Q_a_2}-\eqref{eq:Q_a_6} and \eqref{eq:Q_SW_a_1},
\eqref{eq:Q_SW_a_2}-\eqref{eq:Q_SW_a_6}, respectively, are monotone
non-decreasing in$~\alpha$. That is, for $0\leq\alpha\leq\beta$, the following
inequalities hold%
\begin{align}
\Delta_{\alpha}\left(  \rho_{ABC},\tau_{AC},\omega_{C},\theta_{BC}\right)   &
\leq\Delta_{\beta}\left(  \rho_{ABC},\tau_{AC},\omega_{C},\theta_{BC}\right)
,\label{eq:conj-1}\\
\widetilde{\Delta}_{\alpha}\left(  \rho_{ABC},\tau_{AC},\omega_{C},\theta
_{BC}\right)   &  \leq\widetilde{\Delta}_{\beta}\left(  \rho_{ABC},\tau
_{AC},\omega_{C},\theta_{BC}\right)  , \label{eq:conj-3}%
\end{align}
and similar inequalities hold for all orderings of the last three arguments of
$\Delta_{\alpha}$ and $\widetilde{\Delta}_{\alpha}$.
\end{conjecture}

If Conjecture~\ref{conj:monotone-alpha}\ is true, we could conclude that all
non-sandwiched and sandwiched R\'{e}nyi generalizations of the conditional
mutual information are monotone non-decreasing in $\alpha$ for positive
definite operators. Another implication of monotonicity in $\alpha\geq1/2$ for
$\widetilde{\Delta}_{\alpha}\left(  \rho_{ABC},\rho_{AC},\rho_{C},\rho
_{BC}\right)  $ would be that a tripartite quantum state $\rho_{ABC}$ is a
quantum Markov state if and only if%
\begin{equation}
\widetilde{\Delta}_{\alpha}\left(  \rho_{ABC},\rho_{AC},\rho_{C},\rho
_{BC}\right)  =0
\end{equation}
(with $\alpha\geq1/2$). This would generalize the results from~\cite{HJPW04}
to the case $\alpha\neq1$.

Note that this conjecture does not follow straightforwardly from the following
monotonicity%
\begin{align}
D_{\alpha}\left(  \rho\Vert\sigma\right)   &  \leq D_{\beta}\left(  \rho
\Vert\sigma\right)  ,\\
\widetilde{D}_{\alpha}\left(  \rho\Vert\sigma\right)   &  \leq\widetilde
{D}_{\beta}\left(  \rho\Vert\sigma\right)  , \label{eq:sandwiched-alpha-mono}%
\end{align}
which holds for $0\leq\alpha\leq\beta$ \cite{TCR09,MDSFT13}. However, for
classical states $\rho_{ABC}$, the conjecture is clearly true for
$\Delta_{\alpha}\left(  \rho_{ABC},\rho_{AC},\rho_{C},\rho_{BC}\right)  $ and
$\widetilde{\Delta}_{\alpha}\left(  \rho_{ABC},\rho_{AC},\rho_{C},\rho
_{BC}\right)  $\ by appealing to the above known inequalities.

Observe that some of the conjectured inequalities are redundant. For example,
if%
\begin{equation}
\Delta_{\alpha}\left(  \rho_{ABC},\tau_{AC},\theta_{BC},\omega_{C}\right)
\leq\Delta_{\beta}\left(  \rho_{ABC},\tau_{AC},\theta_{BC},\omega_{C}\right)
\end{equation}
holds for all $\rho_{ABC}\in\mathcal{S}(\mathcal{H}_{ABC})_{++}$, $\tau
_{AC}\in\mathcal{S}(\mathcal{H}_{AC})_{++}$, $\theta_{BC}\in\mathcal{S}%
(\mathcal{H}_{BC})_{++}$, and $\omega_{C}\in\mathcal{S}(\mathcal{H}_{C})_{++}%
$, then the following monotonicity holds as well%
\begin{equation}
\Delta_{\alpha}\left(  \rho_{ABC},\theta_{BC},\tau_{AC},\omega_{C}\right)
\leq\Delta_{\beta}\left(  \rho_{ABC},\theta_{BC},\tau_{AC},\omega_{C}\right)
,
\end{equation}
due to a symmetry under the exchange of systems $A$ and $B$. Similar
statements apply to other pairs of inequalities, so that it suffices to prove
only six of the 12 monotonicities discussed above in order to establish the
other six. However, as we will see below, a single proof of the monotonicity
for each kind of R\'{e}nyi conditional mutual information (non-sandwiched and
sandwiched) should suffice because we think one could easily generalize such a
proof to the other cases.

\subsection{Approaches for proving the conjecture}

\label{sec:proof-approach-outline}We briefly outline some approaches for
proving the conjecture. One idea is to follow a proof technique from
\cite[Lemma~3]{TCR09} and \cite[Theorem~7]{MDSFT13}. If the derivative of
$\Delta_{\alpha}\left(  \rho_{ABC},\tau_{AC},\omega_{C},\theta_{BC}\right)  $
and $\widetilde{\Delta}_{\alpha}\left(  \rho_{ABC},\tau_{AC},\omega_{C}%
,\theta_{BC}\right)  $ with respect to $\alpha$ is non-negative, then we can
conclude that these functions are monotone increasing with $\alpha$. It is
possible to prove that the derivatives are non-negative when $\alpha$ is in a
neighborhood of one, by computing Taylor expansions of these functions. We
explore this approach further in Appendix~\ref{sec:proof-approach}.

\subsection{Numerical evidence}

To test the conjecture in (\ref{eq:conj-1}) and its variations, we conducted
several numerical experiments. First, we selected states $\rho_{ABC}$,
$\tau_{AC}$, $\omega_{C}$, $\theta_{BC}$ at random \cite{CubittMatlab}, with
the dimensions of the local systems never exceeding six. We then computed the
numerator in (\ref{eq:derivative-exp}) for values of $\gamma$ ranging from
$-0.99$ to $10$ with a step size of $0.05$ (so that $\alpha=\gamma+1$ goes
from $0.01$ to $11$). For each value of $\gamma$, we conducted 1000 numerical
experiments. The result was that the numerator in (\ref{eq:derivative-exp})
was always non-negative. We then conducted the same set of experiments for the
various operator orderings and always found the numerator to be non-negative.

To test the conjecture in (\ref{eq:conj-3}) and its variations, we conducted
similar numerical experiments. First, we selected states $\rho_{ABC}$,
$\tau_{AC}$, $\omega_{C}$, $\theta_{BC}$, $\mu_{ABC}$ at random
\cite{CubittMatlab}, with the dimensions of the local systems never exceeding
six. We then computed the numerator in (\ref{eq:derivative-exp-1}) for values
of $\gamma$ ranging from $-10$ to $0.99$ with a step size of 0.05 (so that
$\alpha=1/\left(  1-\gamma\right)  $ goes from $\approx0.091$ to $\approx
100$). For each value of $\gamma$, we conducted 1000 numerical experiments.
The result was that the numerator in (\ref{eq:derivative-exp-1}) was always
non-negative. We then conducted the same set of experiments for the various
operator orderings and always found the numerator to be non-negative.

\subsection{Special cases of the conjecture}

We can prove that the conjecture is true in a number of cases, due to the
special form that the R\'{e}nyi conditional mutual information takes in these
cases. Let $\rho_{ABC}\in\mathcal{S}(\mathcal{H}_{ABC})_{++}$. We define the
following quantities, which are the same as (\ref{eq:renyi-gen-marginals}) and
(\ref{eq:renyi-sandwiched-marginals}), respectively:%
\begin{align}
I_{\alpha}\left(  A;B|C\right)  _{\rho|\rho}  &  \equiv\frac{1}{\alpha-1}%
\log\text{Tr}\left\{  \rho_{ABC}^{\alpha}\rho_{AC}^{\left(  1-\alpha\right)
/2}\rho_{C}^{\left(  \alpha-1\right)  /2}\rho_{BC}^{1-\alpha}\rho_{C}^{\left(
\alpha-1\right)  /2}\rho_{AC}^{\left(  1-\alpha\right)  /2}\right\}  ,\\
\widetilde{I}_{\alpha}\left(  A;B|C\right)  _{\rho|\rho}  &  \equiv
\frac{\alpha}{\alpha-1}\log\left\Vert \rho_{ABC}^{1/2}\rho_{AC}^{\left(
1-\alpha\right)  /2\alpha}\rho_{C}^{\left(  \alpha-1\right)  /2\alpha}%
\rho_{BC}^{\left(  1-\alpha\right)  /\alpha}\rho_{C}^{\left(  \alpha-1\right)
/2\alpha}\rho_{AC}^{\left(  1-\alpha\right)  /2\alpha}\rho_{ABC}%
^{1/2}\right\Vert _{\alpha}, \label{eq:max_end}%
\end{align}
so that%
\begin{align}
I_{0}\left(  A;B|C\right)  _{\rho|\rho}  &  =-\log\text{Tr}\left\{  \rho
_{ABC}^{0}\rho_{AC}^{1/2}\rho_{C}^{-1/2}\rho_{BC}\rho_{C}^{-1/2}\rho
_{AC}^{1/2}\right\}  ,\\
I_{2}\left(  A;B|C\right)  _{\rho|\rho}  &  =\log\text{Tr}\left\{  \rho
_{ABC}^{2}\left(  \rho_{AC}^{1/2}\rho_{C}^{-1/2}\rho_{BC}\rho_{C}^{-1/2}%
\rho_{AC}^{1/2}\right)  ^{-1}\right\}  .
\end{align}
Recall that the following inequality holds for all $\alpha\in(0,1)\cup
(1,\infty)$ \cite{DL13}:%
\begin{equation}
\widetilde{D}_{\alpha}\left(  \rho\Vert\sigma\right)  \leq D_{\alpha}\left(
\rho\Vert\sigma\right)  .
\end{equation}
Using the monotonicity given in (\ref{eq:sandwiched-alpha-mono}) and the above
inequality, we can conclude that%
\begin{align}
I_{0}\left(  A;B|C\right)  _{\rho|\rho}  &  \leq I_{2}\left(  A;B|C\right)
_{\rho|\rho},\label{eq:mono-obs-1}\\
I_{\min}\left(  A;B|C\right)  _{\rho|\rho}  &  \leq I_{\max}\left(
A;B|C\right)  _{\rho|\rho},\label{eq:mono-obs-2}\\
I_{\min}\left(  A;B|C\right)  _{\rho|\rho}  &  \leq I_{2}\left(  A;B|C\right)
_{\rho|\rho}, \label{eq:mono-obs-3}%
\end{align}
where $I_{\max}\left(  A;B|C\right)  _{\rho|\rho}$ and $I_{\min}\left(
A;B|C\right)  _{\rho|\rho}$ are defined in (\ref{eq:Imax}) and (\ref{eq:Imin}%
), respectively. However, we cannot relate to the (von Neumann entropy based)
conditional mutual information because its representation in terms of the
relative entropy does not feature the operator $\rho_{AC}^{1/2}\rho_{C}%
^{-1/2}\rho_{BC}\rho_{C}^{-1/2}\rho_{AC}^{1/2}$ as its second argument but
instead has $\exp\left\{  \log\rho_{BC}+\log\rho_{AC}-\log\rho_{C}\right\}  $.

Let $\rho_{ABC}\in\mathcal{S}(\mathcal{H}_{ABC})_{++}$, $\tau_{AC}%
\in\mathcal{S}(\mathcal{H}_{AC})_{++}$, $\omega_{C}\in\mathcal{S}%
(\mathcal{H}_{C})_{++}$, and $\theta_{BC}\in\mathcal{S}(\mathcal{H}_{BC}%
)_{++}$. Tomamichel has informed us that the inequality in (\ref{eq:conj-3})
and its variations are true for $0\leq\alpha\leq\beta$ and such that
$1/\alpha+1/\beta=2$ \cite{T14}. This is because in such a case, we have that
$\alpha/\left(  1-\alpha\right)  =-\beta\left(  1-\beta\right)  $, so that%
\begin{multline}
\left[  \tau_{AC}^{\left(  1-\alpha\right)  /2\alpha}\omega_{C}^{\left(
\alpha-1\right)  /2\alpha}\theta_{BC}^{\left(  1-\alpha\right)  /\alpha}%
\omega_{C}^{\left(  \alpha-1\right)  /2\alpha}\tau_{AC}^{\left(
1-\alpha\right)  /2\alpha}\right]  ^{\alpha/\left(  1-\alpha\right)  }\\
=\left[  \tau_{AC}^{\left(  1-\beta\right)  /2\beta}\omega_{C}^{\left(
\beta-1\right)  /2\beta}\theta_{BC}^{\left(  1-\beta\right)  /\beta}\omega
_{C}^{\left(  \beta-1\right)  /2\beta}\tau_{AC}^{\left(  1-\beta\right)
/2\beta}\right]  ^{\beta/\left(  1-\beta\right)  },
\end{multline}
and similar equalities hold for the five other operator orderings. Since this
is the case, the monotonicity follows directly from the ordinary monotonicity
of the sandwiched R\'{e}nyi relative entropy. By a similar line of reasoning,
the inequality in (\ref{eq:conj-1}) and its variations are true for
$0\leq\alpha\leq\beta$ and such that $\alpha+\beta=2$. Similarly, in such a
case, we have that $1-\alpha=-\left(  1-\beta\right)  $, so that%
\begin{equation}
\left[  \tau_{AC}^{\left(  1-\alpha\right)  /2}\omega_{C}^{\left(
\alpha-1\right)  /2}\theta_{BC}^{1-\alpha}\omega_{C}^{\left(  \alpha-1\right)
/2}\tau_{AC}^{\left(  1-\alpha\right)  /2}\right]  ^{1/\left(  1-\alpha
\right)  }=\left[  \tau_{AC}^{\left(  1-\beta\right)  /2}\omega_{C}^{\left(
\beta-1\right)  /2}\theta_{BC}^{1-\beta}\omega_{C}^{\left(  \beta-1\right)
/2}\tau_{AC}^{\left(  1-\beta\right)  /2}\right]  ^{1/\left(  1-\beta\right)
},
\end{equation}
and similar equalities hold for the five other operator orderings. Then the
monotonicity again follows from the ordinary monotonicity of the R\'{e}nyi
relative entropy. The observations in (\ref{eq:mono-obs-1}%
)-(\ref{eq:mono-obs-2}) are then special cases of the above observations.

\subsection{Implications for tripartite states with small conditional mutual
information}

\label{sec:small-CMI}It has been an open question since the work in
\cite{HJPW04}\ to characterize tripartite quantum states $\rho_{ABC}$\ with
small conditional mutual information $I\left(  A;B|C\right)  _{\rho}$. That
is, given that the various quantum Markov state conditions in (\ref{eq:QMS-1})
and (\ref{eq:QMS-2})-(\ref{eq:QMS-4}) are equivalent to $I\left(
A;B|C\right)  _{\rho}$ being equal to zero, we would like to understand what
happens when we perturb these various conditions. In this section, we pursue this direction
and explicitly show how Conjecture~\ref{conj:monotone-alpha} could be used to address this
important question. 

Several researchers have already considered what happens when perturbing the
quantum Markov state condition in (\ref{eq:QMS-1}), but we include a
discussion here for completeness. To begin with, we know that if there exists
a quantum Markov state $\mu_{ABC}\in\mathcal{M}_{A-C-B}$ such that%
\begin{equation}
\left\Vert \rho_{ABC}-\mu_{ABC}\right\Vert _{1}\leq\varepsilon
\end{equation}
then%
\begin{align}
I\left(  A;B|C\right)  _{\mu}  &  =0,\\
I\left(  A;B|C\right)  _{\rho}  &  \leq8\varepsilon\log\min\left\{
d_{A},d_{B}\right\}  +4h_{2}\left(  \varepsilon\right)  ,
\end{align}
where%
\begin{equation}
h_{2}\left(  x\right)  \equiv-x\log x-\left(  1-x\right)  \log\left(
1-x\right)  \label{eq:bin-entropy}%
\end{equation}
is the binary entropy, which obeys%
\begin{equation}
\lim_{\varepsilon\searrow0}h_{2}\left(  \varepsilon\right)  =0.
\end{equation}
The first line is by definition and the second follows from an application of
the Alicki-Fannes inequality \cite{AF04}. However, the example in \cite{CSW12}
and the subsequent development in \cite{E14} exclude a particular converse of
the above bound. That is, by \cite[Lemma 6]{CSW12}, there exists a sequence of
states $\rho_{ABC}^{d}$ such that%
\begin{equation}
I\left(  A;B|C\right)  _{\rho^{d}}=2\log\left(  \left(  d+2\right)  /d\right)
, \label{eq:counter-example-CMI}%
\end{equation}
which goes to zero as $d\rightarrow\infty$. However, for this same sequence of
states, the following constant lower bound is known%
\begin{equation}
\min_{\mu_{ABC}\in\mathcal{M}_{A-C-B}}D_{0}\left(  \rho_{ABC}^{d}%
\middle\Vert\mu_{ABC}\right)  \geq\log\sqrt{4/3}, \label{eq:erker-bound}%
\end{equation}
by \cite[Theorem~1]{E14}. By employing monotonicity of the R\'{e}nyi relative
entropy with respect to the R\'{e}nyi parameter, so that $D_{1/2}\geq D_{0}$,
and the well-known relation $1-\left\Vert \omega-\tau\right\Vert _{1}/2\leq
\ $Tr$\{\sqrt{\omega}\sqrt{\tau}\}$ for $\omega,\tau\in\mathcal{S}\left(
\mathcal{H}\right)  $ (see, e.g., \cite[Equation~(22)]{CMMAB08}), we can
readily translate the bound in (\ref{eq:erker-bound})\ to a constant lower
bound on the trace distance of $\rho_{ABC}^{d}$ to the set of quantum Markov
states:%
\begin{equation}
\left\Vert \rho_{ABC}^{d}-\mathcal{M}_{A-C-B}\right\Vert _{1}\equiv\min
_{\mu_{ABC}\in\mathcal{M}_{A-C-B}}\left\Vert \rho_{ABC}^{d}-\mu_{ABC}%
\right\Vert _{1}\geq2\left(  1-\left(  3/4\right)  ^{1/4}\right)
\approx0.139. \label{eq:counter-example-TD}%
\end{equation}
So (\ref{eq:counter-example-CMI}) and (\ref{eq:counter-example-TD}) imply that
a Pinsker-like bound of the form $I\left(  A;B|C\right)  _{\rho}\geq
K\left\Vert \rho_{ABC}-\mathcal{M}_{A-C-B}\right\Vert _{1}^{2}$ cannot hold in
general, with $K$ a dimension-independent constant.

We now focus on a perturbation of the conditions in (\ref{eq:QMS-2}%
)-(\ref{eq:QMS-3}). It appears that these cases will be promising for
applications if Conjecture~\ref{conj:monotone-alpha} is true. The following
proposition states that the conditional mutual information is small if it is
possible to recover the system $A$ from system $C$ alone (or by symmetry, if
one can get $B$ from $C$ alone). We note that \eqref{eq:cmi-rho} was proven
independently in \cite[Eq.~(8)]{FR14}.

\begin{proposition}
\label{prop:small-CMI}Let $\rho_{ABC}\in\mathcal{S}\left(  \mathcal{H}%
_{ABC}\right)  $, $\mathcal{R}_{C\rightarrow AC}$ be a CPTP\ \textquotedblleft
recovery\textquotedblright\ map, and $\varepsilon\in\left[  0,1\right]  $.
Suppose that it is possible to recover the system $A$ from system $C$ alone,
in the following sense%
\begin{equation}
\left\Vert \rho_{ABC}-\omega_{ABC}\right\Vert _{1}\leq\varepsilon,
\label{eq:recoverable-state}%
\end{equation}
where%
\begin{equation}
\omega_{ABC}\equiv\mathcal{R}_{C\rightarrow AC}\left(  \rho_{BC}\right)  .
\end{equation}
Then the conditional mutual informations $I\left(  A;B|C\right)  _{\rho}$ and
$I\left(  A;B|C\right)  _{\omega}$\ obey the following bounds:%
\begin{align}
I\left(  A;B|C\right)  _{\rho}  &  \leq4\varepsilon\log d_{B}+2h_{2}\left(
\varepsilon\right)  , \label{eq:cmi-rho}\\
I\left(  A;B|C\right)  _{\omega}  &  \leq4\varepsilon\log d_{B}+2h_{2}\left(
\varepsilon\right)  ,
\end{align}
where $d_{B}$ is the dimension of the $B$ system and $h_{2}\left(
\varepsilon\right)  $ is defined in \eqref{eq:bin-entropy}. By symmetry, a
related bound holds if one can recover system $B$ from system $C$ alone.
\end{proposition}

\begin{proof}
Consider that%
\begin{align}
I\left(  A;B|C\right)  _{\rho}  &  =H\left(  B|C\right)  _{\rho}-H\left(
B|AC\right)  _{\rho}\\
&  \leq H\left(  B|AC\right)  _{\omega}-H\left(  B|AC\right)  _{\rho}\\
&  \leq H\left(  B|AC\right)  _{\omega}-H\left(  B|AC\right)  _{\omega
}+4\varepsilon\log d_{B}+2h_{2}\left(  \varepsilon\right) \\
&  =4\varepsilon\log d_{B}+2h_{2}\left(  \varepsilon\right)  .
\end{align}
The first inequality follows because the conditional entropy is monotone
increasing under quantum operations on the conditioning system (the map
$\mathcal{R}_{C\rightarrow AC}$ is applied to the system $C$ of state
$\rho_{ABC}$ to produce $\omega_{ABC}$ and the conditional entropy only
increases under such processing). The second inequality is a result of
(\ref{eq:recoverable-state}) and the Alicki-Fannes inequality \cite{AF04}
(continuity of conditional entropy). Similarly, consider that%
\begin{align}
I\left(  A;B|C\right)  _{\omega}  &  =H\left(  B|C\right)  _{\omega}-H\left(
B|AC\right)  _{\omega}\\
&  \leq H\left(  B|C\right)  _{\rho}-H\left(  B|AC\right)  _{\omega
}+4\varepsilon\log d_{B}+2h_{2}\left(  \varepsilon\right) \\
&  \leq H\left(  B|AC\right)  _{\omega}-H\left(  B|AC\right)  _{\omega
}+4\varepsilon\log d_{B}+2h_{2}\left(  \varepsilon\right) \\
&  =4\varepsilon\log d_{B}+2h_{2}\left(  \varepsilon\right)  .
\end{align}
The first inequality is from the fact that (\ref{eq:recoverable-state})
implies that
\begin{equation}
\left\Vert \rho_{BC}-\omega_{BC}\right\Vert _{1}\leq\varepsilon
\end{equation}
and the Alicki-Fannes' inequality. The second is again from monotonicity of
conditional entropy.
\end{proof}

The implications of Conjecture~\ref{conj:monotone-alpha}\ are nontrivial. For
example, if it were true, then we could conclude a converse of
Proposition~\ref{prop:small-CMI}, that if the conditional mutual information
is small, then it is possible to recover the system $A$ from system $C$ alone
(or by symmetry, that one can get $B$ from $C$ alone). That is, the following
relation would hold for $\rho_{ABC}\in\mathcal{S}(\mathcal{H}_{ABC})_{++}$:%
\begin{align}
I\left(  A;B|C\right)  _{\rho}  &  \geq I_{\min}\left(  A;B|C\right)
_{\rho|\rho}\label{eq:conq1}\\
&  =-\log F\left(  \rho_{ABC},\rho_{AC}^{1/2}\rho_{C}^{-1/2}\rho_{BC}\rho
_{C}^{-1/2}\rho_{AC}^{1/2}\right) \\
&  =-\log F\left(  \rho_{ABC},\mathcal{R}_{C\rightarrow AC}^{P}\left(
\rho_{BC}\right)  \right) \\
&  \geq-\log\left[  1-\left(  \frac{1}{2}\left\Vert \rho_{ABC}-\mathcal{R}%
_{C\rightarrow AC}^{P}\left(  \rho_{BC}\right)  \right\Vert _{1}\right)
^{2}\right] \\
&  \geq\frac{1}{4}\left\Vert \rho_{ABC}-\mathcal{R}_{C\rightarrow AC}%
^{P}\left(  \rho_{BC}\right)  \right\Vert _{1}^{2},
\end{align}
where $\mathcal{R}_{C\rightarrow AC}^{P}$ is Petz's transpose map discussed
in~\cite{HJPW04}%
\begin{equation}
\mathcal{R}_{C\rightarrow AC}^{P}(\cdot)\equiv\rho_{AC}^{1/2}\rho_{C}%
^{-1/2}(\cdot)\rho_{C}^{-1/2}\rho_{AC}^{1/2}.
\end{equation}
In the above, the first inequality would follow from
Conjecture~\ref{conj:monotone-alpha}, the second is a result of well known
relations between trace distance and fidelity \cite{FG98}, and the last is a
consequence of the inequality $-\log\left(  1-x\right)  \geq x$, valid for
$x\leq1$. Thus, the truth of Conjecture~\ref{conj:monotone-alpha} would
establish the truth of an open conjecture from \cite{K13conj} (up to a
constant). As pointed out in \cite{K13conj}, this would then imply that for
tripartite states $\rho_{ABC}$ with conditional mutual information
$I(A;B|C)_{\rho}$ small (i.e., states that fulfill strong subadditivity with
near equality), Petz's transpose map for the partial trace over $A$ is good
for recovering $\rho_{ABC}$ from $\rho_{BC}$. Hence, even though $\rho_{ABC}$
does not have to be close to a quantum Markov state if $I(A;B|C)_{\rho}$ is
small (as discussed above), $A$ would still be nearly independent of $B$ from
the perspective of $C$ in the sense that $\rho_{ABC}$ could be approximately
recovered from $\rho_{BC}$ alone. This would give an operationally useful
characterization of states that fulfill strong subadditivity with near
equality and would be helpful for answering some open questions concerning
squashed entanglement, as discussed in \cite{Winterconj}.

For the quantum Markov state condition in (\ref{eq:QMS-4}), for simplicity we
consider instead the \textquotedblleft relative entropy
distance\textquotedblright\ between $\rho_{ABC}$ and $\varsigma_{ABC}$, where%
\begin{equation}
\varsigma_{ABC}\equiv\exp\left\{  \log\rho_{AC}+\log\rho_{BC}-\log\rho
_{C}\right\}  .
\end{equation}
So if%
\begin{equation}
D\left(  \rho_{ABC}\Vert\varsigma_{ABC}\right)  \leq\varepsilon,
\end{equation}
then we can conclude that%
\begin{equation}
I\left(  A;B|C\right)  _{\rho}=D\left(  \rho_{ABC}\Vert\varsigma_{ABC}\right)
\leq\varepsilon.
\end{equation}
If desired, one can also obtain an $\varepsilon$-dependent upper bound on
$I\left(  A;B|C\right)  _{\varsigma^{\prime}}$, where $\varsigma_{ABC}%
^{\prime}\equiv\varsigma_{ABC}/$Tr$\left\{  \varsigma_{ABC}\right\}  $, which
vanishes in the limit as $\varepsilon$ goes to zero. This can be accomplished
by employing the bound in Corollary~\ref{cor:pinsker-exp-log} and by bounding
Tr$\left\{  \varsigma_{ABC}\right\}  $ from below by $1-\left\Vert \rho
_{ABC}-\varsigma_{ABC}\right\Vert _{1}$. The bound in
Corollary~\ref{cor:pinsker-exp-log} also serves as a converse of these bounds:
if the conditional mutual information is small, then the trace distance
between $\rho_{ABC}$ and $\varsigma_{ABC}$ is small. However, it is not clear
that a perturbation of the quantum Markov state condition in (\ref{eq:QMS-4})
will be as useful in applications as a perturbation of (\ref{eq:QMS-2}%
)-(\ref{eq:QMS-3}) would be, mainly because the map $\rho_{ABC}\rightarrow
\exp\left\{  \log\rho_{AC}+\log\rho_{BC}-\log\rho_{C}\right\}  $ is non-linear
(as discussed in \cite{K14a}).

\section{Discussion}

\label{sec:conclusion}This paper has defined several R\'{e}nyi generalizations
of the conditional quantum mutual information (CQMI) quantities that satisfy
properties that should find use in applications. Namely, we showed that these
generalizations are non-negative and are monotone under local quantum
operations on one of the systems $A$ or $B$. An important open question is to
prove that they are monotone under local quantum operations on both systems.
Some of the R\'{e}nyi generalizations satisfy a generalization of the duality
relation $I(A;B|C)=I(A;B|D)$, which holds for a four-party pure state
$\psi_{ABCD}$. We conjecture that these R\'{e}nyi generalizations of the
CQMI\ are monotone non-decreasing\ in the R\'{e}nyi parameter $\alpha$, and we
have proved that this conjecture is true when $\alpha$ is in a neighborhood of
one and in some other special cases. The truth of this conjecture in general
would have implications in condensed matter physics, as detailed in
\cite{K13conj}, and quantum communication complexity, as mentioned in
\cite{Touchette14}.

Based on the fact that the conditional mutual information can be written as%
\begin{equation}
I\left(  A;B|C\right)  _{\rho}   =
D\left(  \rho_{ABC}\Vert\exp\left\{  \log\rho_{AC}+\log\rho_{BC}-\log
\rho_{C}\right\}  \right)  ,
\end{equation}
one could consider another R\'{e}nyi generalization of the conditional mutual
information, such as%
\begin{equation}
  D_{\alpha}\left(  \rho_{ABC}\Vert\exp\left\{  \log\rho_{AC}+\log\rho
_{BC}-\log\rho_{C}\right\}  \right)  ,\label{eq:gen-3-editor}%
\end{equation}
or with the sandwiched variant. However, it is unclear to us whether (\ref{eq:gen-3-editor}) is
monotone under local operations, which we have argued is an important
property for a R\'{e}nyi generalization of conditional mutual information.

There are many directions to consider going forward from this paper. First,
one could improve many of the results here on a technical level. It would be
interesting to understand in depth the limits in (\ref{eq:limit-xi-to-0-Delta}%
), (\ref{eq:limit-xi-Delta-a-1}), and (\ref{eq:finite-Delta-SW}) in order to
establish the most general support conditions for the $\Delta$, $\Delta
_{\alpha}$, and $\widetilde{\Delta}_{\alpha}$ quantities, respectively, as has
been done for the quantum and R\'{e}nyi relative entropies, as recalled in
(\ref{eq:rel-ent-consistency}), (\ref{eq:supp-renyi-1}), and
(\ref{eq:supp-renyi-2}). Next, if one could establish uniform convergence of
the $\Delta_{\alpha}$ and $\widetilde{\Delta}_{\alpha}$ quantities as $\alpha$
goes to one, then we could conclude that the optimized versions of these
quantities converge to the conditional mutual information in this limit. One
might also attempt to extend Theorem~\ref{thm:renyi-cmi-a-to-1},
Theorem~\ref{thm:conv-vN-sandwiched}, and
Lemma~\ref{lem:monotone-lemma-sandwiched} to hold for positive semi-definite
density operators.

As far as applications are concerned, one could explore a R\'{e}nyi squashed
entanglement and determine if several properties hold which are analogous to
the squashed entanglement \cite{CW04}. Such a quantity might be helpful in
strengthening \cite[Proposition~10]{CW04}, so that the squashed entanglement
could be interpreted as a strong converse upper bound on distillable
entanglement. More generally, it might be helpful in strengthening the main
result of \cite{TGW13}, so that the upper bound established on the two-way
assisted quantum capacity could be interpreted as a strong converse rate. The
quantities defined here might be useful in the context of one-shot information
theory, for example, to establish a one-shot state redistribution protocol as
an extension of the main result of \cite{DY08}. Preliminary results on
R\'{e}nyi squashed entanglement and discord are discussed in our follow-up
paper \cite{SBW14}. One could also explore applications of the R\'{e}nyi
conditional mutual informations in the context of condensed matter physics or
high energy physics, as the R\'{e}nyi entropy has been employed extensively in
these contexts \cite{CC09}.

Finally, these potential applications in information theory and physics should
help in singling out some of our many possible definitions for R\'{e}nyi
conditional mutual information.

\bigskip

\textbf{Acknowledgments.} KS acknowledges support ffrom the DARPA Quiness
Program through US Army Research Office award W31P4Q-12-1-0019 and the
Graduate school, Louisiana State University. MMW\ is grateful to the Institute
for Quantum Information and Matter at Caltech for hospitality during a
research visit in July 2014. MMW\ acknowledges startup funds from the
Department of Physics and Astronomy at LSU, support from the NSF\ under Award
No.~CCF-1350397, and support from the DARPA Quiness Program through US Army
Research Office award W31P4Q-12-1-0019.

\appendix

\section{Sibson identity for the R\'{e}nyi conditional mutual information}

\label{sec:sibson}The R\'{e}nyi conditional mutual information in
Definition~\ref{def:Renyi-CMI}\ has an explicit form, much like other
R\'{e}nyi information quantities \cite{KW09,SW12,GW13,TBH13}. We prove this in
two steps, first by proving the following Sibson identity \cite{S69}.

\begin{lemma}
\label{lem:unique-min}The following quantum Sibson identity holds when
$\operatorname{supp}\left(  \rho_{ABC}\right)  \subseteq\operatorname{supp}%
\left(  \sigma_{BC}\right)  $ and for $\alpha\in(0,1)\cup(1,\infty)$:%
\begin{equation}
\Delta_{\alpha}\left(  \rho_{ABC},\rho_{AC},\rho_{C},\sigma_{BC}\right)
=\Delta_{\alpha}\left(  \rho_{ABC},\rho_{AC},\rho_{C},\sigma_{BC}^{\ast
}\right)  +D_{\alpha}\left(  \sigma_{BC}^{\ast}\Vert\sigma_{BC}\right)  ,
\end{equation}
with the state $\sigma_{BC}^{\ast}$ having the form%
\begin{equation}
\sigma_{BC}^{\ast}\equiv\frac{\left(  \operatorname{Tr}_{A}\left\{  \rho
_{C}^{\left(  \alpha-1\right)  /2}\rho_{AC}^{\left(  1-\alpha\right)  /2}%
\rho_{ABC}^{\alpha}\rho_{AC}^{\left(  1-\alpha\right)  /2}\rho_{C}^{\left(
\alpha-1\right)  /2}\right\}  \right)  ^{1/\alpha}}{\operatorname{Tr}\left\{
\left(  \operatorname{Tr}_{A}\left\{  \rho_{C}^{\left(  \alpha-1\right)
/2}\rho_{AC}^{\left(  1-\alpha\right)  /2}\rho_{ABC}^{\alpha}\rho
_{AC}^{\left(  1-\alpha\right)  /2}\rho_{C}^{\left(  \alpha-1\right)
/2}\right\}  \right)  ^{1/\alpha}\right\}  }. \label{eq:optimal-sibson}%
\end{equation}

\end{lemma}

\begin{proof}
The relation for $\sigma_{BC}^{\ast}$ implies that%
\begin{multline}
\left[  \sigma_{BC}^{\ast}\text{Tr}\left\{  \left(  \text{Tr}_{A}\left\{
\rho_{C}^{\left(  \alpha-1\right)  /2}\rho_{AC}^{\left(  1-\alpha\right)
/2}\rho_{ABC}^{\alpha}\rho_{AC}^{\left(  1-\alpha\right)  /2}\rho_{C}^{\left(
\alpha-1\right)  /2}\right\}  \right)  ^{1/\alpha}\right\}  \right]  ^{\alpha
}\\
=\text{Tr}_{A}\left\{  \rho_{C}^{\left(  \alpha-1\right)  /2}\rho
_{AC}^{\left(  1-\alpha\right)  /2}\rho_{ABC}^{\alpha}\rho_{AC}^{\left(
1-\alpha\right)  /2}\rho_{C}^{\left(  \alpha-1\right)  /2}\right\}  .
\end{multline}
Then consider that%
\begin{align}
&  \Delta_{\alpha}\left(  \rho_{ABC},\rho_{AC},\rho_{C},\sigma_{BC}\right)
\nonumber\\
&  =\frac{1}{\alpha-1}\log\text{Tr}\left\{  \rho_{ABC}^{\alpha}\rho
_{AC}^{\left(  1-\alpha\right)  /2}\rho_{C}^{\left(  \alpha-1\right)
/2}\sigma_{BC}^{1-\alpha}\rho_{C}^{\left(  \alpha-1\right)  /2}\rho
_{AC}^{\left(  1-\alpha\right)  /2}\right\} \\
&  =\frac{1}{\alpha-1}\log\text{Tr}\left\{  \rho_{C}^{\left(  \alpha-1\right)
/2}\rho_{AC}^{\left(  1-\alpha\right)  /2}\rho_{ABC}^{\alpha}\rho
_{AC}^{\left(  1-\alpha\right)  /2}\rho_{C}^{\left(  \alpha-1\right)
/2}\sigma_{BC}^{1-\alpha}\right\} \\
&  =\frac{1}{\alpha-1}\log\text{Tr}\left\{  \text{Tr}_{A}\left\{  \rho
_{C}^{\left(  \alpha-1\right)  /2}\rho_{AC}^{\left(  1-\alpha\right)  /2}%
\rho_{ABC}^{\alpha}\rho_{AC}^{\left(  1-\alpha\right)  /2}\rho_{C}^{\left(
\alpha-1\right)  /2}\right\}  \sigma_{BC}^{1-\alpha}\right\} \\
&  =\frac{1}{\alpha-1}\log\text{Tr}\left\{  \left[  \sigma_{BC}^{\ast}\right]
^{\alpha}\sigma_{BC}^{1-\alpha}\right\}  +\frac{\alpha}{\alpha-1}\log
\text{Tr}\left\{  \left(  \text{Tr}_{A}\left\{  \rho_{C}^{\left(
\alpha-1\right)  /2}\rho_{AC}^{\left(  1-\alpha\right)  /2}\rho_{ABC}^{\alpha
}\rho_{AC}^{\left(  1-\alpha\right)  /2}\rho_{C}^{\left(  \alpha-1\right)
/2}\right\}  \right)  ^{1/\alpha}\right\}  .
\end{align}
Now consider expanding the following:%
\begin{align}
&  \Delta_{\alpha}\left(  \rho_{ABC},\rho_{AC},\rho_{C},\sigma_{BC}^{\ast
}\right) \nonumber\\
&  =\frac{1}{\alpha-1}\log\text{Tr}\left\{  \rho_{C}^{\left(  \alpha-1\right)
/2}\rho_{AC}^{\left(  1-\alpha\right)  /2}\rho_{ABC}^{\alpha}\rho
_{AC}^{\left(  1-\alpha\right)  /2}\rho_{C}^{\left(  \alpha-1\right)
/2}\left[  \sigma_{BC}^{\ast}\right]  ^{1-\alpha}\right\} \\
&  =\frac{1}{\alpha-1}\log\text{Tr}\left\{  \text{Tr}_{A}\left\{  \rho
_{C}^{\left(  \alpha-1\right)  /2}\rho_{AC}^{\left(  1-\alpha\right)  /2}%
\rho_{ABC}^{\alpha}\rho_{AC}^{\left(  1-\alpha\right)  /2}\rho_{C}^{\left(
\alpha-1\right)  /2}\right\}  \left[  \sigma_{BC}^{\ast}\right]  ^{1-\alpha
}\right\} \\
&  =\frac{1}{\alpha-1}\log\text{Tr}\left\{  \left[  \text{Tr}_{A}\left\{
\rho_{C}^{\left(  \alpha-1\right)  /2}\rho_{AC}^{\left(  1-\alpha\right)
/2}\rho_{ABC}^{\alpha}\rho_{AC}^{\left(  1-\alpha\right)  /2}\rho_{C}^{\left(
\alpha-1\right)  /2}\right\}  \right]  ^{1/\alpha}\right\} \\
&  +\log\text{Tr}\left\{  \left(  \text{Tr}_{A}\left\{  \rho_{C}^{\left(
\alpha-1\right)  /2}\rho_{AC}^{\left(  1-\alpha\right)  /2}\rho_{ABC}^{\alpha
}\rho_{AC}^{\left(  1-\alpha\right)  /2}\rho_{C}^{\left(  \alpha-1\right)
/2}\right\}  \right)  ^{1/\alpha}\right\} \\
&  =\frac{\alpha}{\alpha-1}\log\text{Tr}\left\{  \left(  \text{Tr}_{A}\left\{
\rho_{C}^{\left(  \alpha-1\right)  /2}\rho_{AC}^{\left(  1-\alpha\right)
/2}\rho_{ABC}^{\alpha}\rho_{AC}^{\left(  1-\alpha\right)  /2}\rho_{C}^{\left(
\alpha-1\right)  /2}\right\}  \right)  ^{1/\alpha}\right\}  .
\end{align}
Putting everything together, we can conclude the statement of the lemma.
\end{proof}

\begin{corollary}
\label{cor:sibson}The R\'{e}nyi conditional mutual information has the
following explicit form for $\alpha\in(0,1)\cup(1,\infty)$:%
\begin{equation}
I_{\alpha}\left(  A;B|C\right)  _{\rho}=\frac{\alpha}{\alpha-1}\log
\operatorname{Tr}\left\{  \left(  \rho_{C}^{\left(  \alpha-1\right)
/2}\operatorname{Tr}_{A}\left\{  \rho_{AC}^{\left(  1-\alpha\right)  /2}%
\rho_{ABC}^{\alpha}\rho_{AC}^{\left(  1-\alpha\right)  /2}\right\}  \rho
_{C}^{\left(  \alpha-1\right)  /2}\right)  ^{1/\alpha}\right\}  .
\end{equation}
The infimum in $I_{\alpha}\left(  A;B|C\right)  _{\rho}$ is achieved uniquely
by the state in \eqref{eq:optimal-sibson}, so that it can be replaced by a minimum.
\end{corollary}

\begin{proof}
This follows from the previous lemma:%
\begin{align}
I_{\alpha}\left(  A;B|C\right)  _{\rho}  &  =\inf_{\sigma_{BC}}\Delta_{\alpha
}\left(  \rho_{ABC},\rho_{AC},\rho_{C},\sigma_{BC}\right) \\
&  =\inf_{\sigma_{BC}}\left[  \Delta_{\alpha}\left(  \rho_{ABC},\rho_{AC}%
,\rho_{C},\sigma_{BC}^{\ast}\right)  +D_{\alpha}\left(  \sigma_{BC}^{\ast
}\Vert\sigma_{BC}\right)  \right] \\
&  =\Delta_{\alpha}\left(  \rho_{ABC},\rho_{AC},\rho_{C},\sigma_{BC}^{\ast
}\right) \\
&  =\frac{\alpha}{\alpha-1}\log\text{Tr}\left\{  \left(  \text{Tr}_{A}\left\{
\rho_{C}^{\left(  \alpha-1\right)  /2}\rho_{AC}^{\left(  1-\alpha\right)
/2}\rho_{ABC}^{\alpha}\rho_{AC}^{\left(  1-\alpha\right)  /2}\rho_{C}^{\left(
\alpha-1\right)  /2}\right\}  \right)  ^{1/\alpha}\right\}  .
\end{align}

\end{proof}

Other Sibson identities hold for other variations of the R\'{e}nyi conditional
mutual information (whenever the innermost operator is optimized over and the
others are the marginals of $\rho_{ABC}$). The proof for this is the same as
given above.

\section{Convergence of the R\'{e}nyi conditional mutual information}

\label{sec:renyi-cmi-a-to-1}Before giving a proof of
Theorem~\ref{thm:renyi-cmi-a-to-1}, we first establish the following lemma,
which is a slight extension of \cite[Proposition~15]{MDSFT13}.

\begin{lemma}
\label{lem:matrix-der-lemma}Let $Z\left(  \alpha\right)  \in\mathcal{B}\left(
\mathcal{H}\right)  _{++}$ be an operator-valued function and let $f\left(
\alpha\right)  $ be a function, both continuously differentiable in $\alpha$
for all $\alpha\in(0,\infty)$. Then the derivative $\frac{d}{d\alpha
}\operatorname{Tr}\{Z\left(  \alpha\right)  ^{f\left(  \alpha\right)  }\}$
exists and is equal to%
\begin{equation}
\frac{d}{d\alpha}\operatorname{Tr}\left\{  Z\left(  \alpha\right)  ^{f\left(
\alpha\right)  }\right\}  =\left(  \frac{d}{d\alpha}f\left(  \alpha\right)
\right)  \operatorname{Tr}\left\{  Z\left(  \alpha\right)  ^{f\left(
\alpha\right)  }\log Z\left(  \alpha\right)  \right\}  +f\left(
\alpha\right)  \operatorname{Tr}\left\{  Z\left(  \alpha\right)  ^{f\left(
\alpha\right)  -1}\frac{d}{d\alpha}Z\left(  \alpha\right)  \right\}  .
\end{equation}

\end{lemma}

\begin{proof}
We proceed as in \cite[Theorem~2.7]{OZ99} or \cite[Proposition~15]{MDSFT13}.
Consider that%
\begin{align}
&  Z\left(  \alpha+h\right)  ^{f\left(  \alpha+h\right)  }-Z\left(
\alpha\right)  ^{f\left(  \alpha\right)  }\nonumber\\
&  =\int_{0}^{1}ds\frac{d}{ds}\left[  Z\left(  \alpha+h\right)  ^{sf\left(
\alpha+h\right)  }Z\left(  \alpha\right)  ^{\left(  1-s\right)  f\left(
\alpha\right)  }\right] \\
&  =\int_{0}^{1}ds\ Z\left(  \alpha+h\right)  ^{sf\left(  \alpha+h\right)
}\left[  \log Z\left(  \alpha+h\right)  ^{f\left(  \alpha+h\right)  }-\log
Z\left(  \alpha\right)  ^{f\left(  \alpha\right)  }\right]  Z\left(
\alpha\right)  ^{\left(  1-s\right)  f\left(  \alpha\right)  }.
\end{align}
Taking the trace, we get%
\begin{multline}
\text{Tr}\left\{  Z\left(  \alpha+h\right)  ^{f\left(  \alpha+h\right)
}\right\}  -\text{Tr}\left\{  Z\left(  \alpha\right)  ^{f\left(
\alpha\right)  }\right\} \\
=f\left(  \alpha+h\right)  \int_{0}^{1}ds\ \text{Tr}\left\{  Z\left(
\alpha\right)  ^{\left(  1-s\right)  f\left(  \alpha\right)  }Z\left(
\alpha+h\right)  ^{sf\left(  \alpha+h\right)  }\left[  \log Z\left(
\alpha+h\right)  -\log Z\left(  \alpha\right)  \right]  \right\} \\
\left(  f\left(  \alpha+h\right)  -f\left(  \alpha\right)  \right)  \int
_{0}^{1}ds\ \text{Tr}\left\{  Z\left(  \alpha\right)  ^{\left(  1-s\right)
f\left(  \alpha\right)  }Z\left(  \alpha+h\right)  ^{sf\left(  \alpha
+h\right)  }\log Z\left(  \alpha\right)  \right\}  .
\end{multline}
Dividing by $h$ and taking the limit as $h\rightarrow0$, we find%
\begin{multline}
\lim_{h\rightarrow0}\frac{1}{h}\left[  \text{Tr}\left\{  Z\left(
\alpha+h\right)  ^{f\left(  \alpha+h\right)  }\right\}  -\text{Tr}\left\{
Z\left(  \alpha\right)  ^{f\left(  \alpha\right)  }\right\}  \right] \\
=f\left(  \alpha\right)  \int_{0}^{1}ds\ \text{Tr}\left\{  Z\left(
\alpha\right)  ^{\left(  1-s\right)  f\left(  \alpha\right)  }Z\left(
\alpha\right)  ^{sf\left(  \alpha\right)  }\lim_{h\rightarrow0}\frac{1}%
{h}\left[  \log Z\left(  \alpha+h\right)  -\log Z\left(  \alpha\right)
\right]  \right\} \\
+\lim_{h\rightarrow0}\frac{f\left(  \alpha+h\right)  -f\left(  \alpha\right)
}{h}\int_{0}^{1}ds\ \text{Tr}\left\{  Z\left(  \alpha\right)  ^{\left(
1-s\right)  f\left(  \alpha\right)  }Z\left(  \alpha\right)  ^{sf\left(
\alpha\right)  }\log Z\left(  \alpha\right)  \right\}  ,
\end{multline}
which is equal to%
\begin{equation}
f\left(  \alpha\right)  \text{Tr}\left\{  Z\left(  \alpha\right)  ^{f\left(
\alpha\right)  }\frac{d}{d\alpha}\left[  \log Z\left(  \alpha\right)  \right]
\right\}  +\left(  \frac{d}{d\alpha}f\left(  \alpha\right)  \right)
\text{Tr}\left\{  Z\left(  \alpha\right)  ^{f\left(  \alpha\right)  }\log
Z\left(  \alpha\right)  \right\}  .
\end{equation}
Carrying out the same arguments as in \cite[Theorem~2.7]{OZ99} or
\cite[Proposition~15]{MDSFT13} in order to compute $\frac{d}{d\alpha}\left[
\log Z\left(  \alpha\right)  \right]  $, we recover the formula in the
statement of the lemma.
\end{proof}

\bigskip

We now provide a proof of Theorem~\ref{thm:renyi-cmi-a-to-1}. The idea is
similar to that in the proof of Theorem~\ref{thm:conv-vN}. To this end, we
again invoke L'H\^{o}pital's rule. We begin by defining%
\begin{equation}
G\left(  \alpha\right)  \equiv\rho_{C}^{\left(  \alpha-1\right)  /2}%
\text{Tr}_{A}\left\{  \rho_{AC}^{\left(  1-\alpha\right)  /2}\rho
_{ABC}^{\alpha}\rho_{AC}^{\left(  1-\alpha\right)  /2}\right\}  \rho
_{C}^{\left(  \alpha-1\right)  /2},
\end{equation}
which implies that%
\begin{equation}
I_{\alpha}\left(  A;B|C\right)  _{\rho}=\frac{1}{1-\frac{1}{\alpha}}%
\log\text{Tr}\left\{  G\left(  \alpha\right)  ^{1/\alpha}\right\}  .
\end{equation}
Applying Lemma~\ref{lem:matrix-der-lemma}\ to $G\left(  \alpha\right)  $ and
the function $1/\alpha$, we find that%
\begin{equation}
\frac{d}{d\alpha}\text{Tr}\left\{  G\left(  \alpha\right)  ^{1/\alpha
}\right\}  =-\frac{1}{\alpha^{2}}\text{Tr}\left\{  G\left(  \alpha\right)
^{1/\alpha}\log G\left(  \alpha\right)  \right\}  +\frac{1}{\alpha}%
\text{Tr}\left\{  G\left(  \alpha\right)  ^{\left(  1-\alpha\right)  /\alpha
}\frac{d}{d\alpha}G\left(  \alpha\right)  \right\}  .
\end{equation}
Also, we have that%
\begin{multline}
\frac{d}{d\alpha}G\left(  \alpha\right)  =\frac{d}{d\alpha}\left[  \rho
_{C}^{\left(  \alpha-1\right)  /2}\text{Tr}_{A}\left\{  \rho_{AC}^{\left(
1-\alpha\right)  /2}\rho_{ABC}^{\alpha}\rho_{AC}^{\left(  1-\alpha\right)
/2}\right\}  \rho_{C}^{\left(  \alpha-1\right)  /2}\right] \\
=\frac{1}{2}\left(  \log\rho_{C}\right)  \rho_{C}^{\left(  \alpha-1\right)
/2}\text{Tr}_{A}\left\{  \rho_{AC}^{\left(  1-\alpha\right)  /2}\rho
_{ABC}^{\alpha}\rho_{AC}^{\left(  1-\alpha\right)  /2}\right\}  \rho
_{C}^{\left(  \alpha-1\right)  /2}\\
-\frac{1}{2}\rho_{C}^{\left(  \alpha-1\right)  /2}\text{Tr}_{A}\left\{
\left(  \log\rho_{AC}\right)  \rho_{AC}^{\left(  1-\alpha\right)  /2}%
\rho_{ABC}^{\alpha}\rho_{AC}^{\left(  1-\alpha\right)  /2}\right\}  \rho
_{C}^{\left(  \alpha-1\right)  /2}\\
+\rho_{C}^{\left(  \alpha-1\right)  /2}\text{Tr}_{A}\left\{  \rho
_{AC}^{\left(  1-\alpha\right)  /2}\left(  \log\rho_{ABC}\right)  \rho
_{ABC}^{\alpha}\rho_{AC}^{\left(  1-\alpha\right)  /2}\right\}  \rho
_{C}^{\left(  \alpha-1\right)  /2}\\
-\frac{1}{2}\rho_{C}^{\left(  \alpha-1\right)  /2}\text{Tr}_{A}\left\{
\rho_{AC}^{\left(  1-\alpha\right)  /2}\rho_{ABC}^{\alpha}\left(  \log
\rho_{AC}\right)  \rho_{AC}^{\left(  1-\alpha\right)  /2}\right\}  \rho
_{C}^{\left(  \alpha-1\right)  /2}\\
+\frac{1}{2}\rho_{C}^{\left(  \alpha-1\right)  /2}\text{Tr}_{A}\left\{
\rho_{AC}^{\left(  1-\alpha\right)  /2}\rho_{ABC}^{\alpha}\rho_{AC}^{\left(
1-\alpha\right)  /2}\right\}  \left(  \log\rho_{C}\right)  \rho_{C}^{\left(
\alpha-1\right)  /2}.
\end{multline}
Applying L'H\^{o}pital's rule gives%
\begin{equation}
\lim_{\alpha\rightarrow1}I_{\alpha}\left(  A;B|C\right)  _{\rho}=\lim
_{\alpha\rightarrow1}\frac{-\text{Tr}\left\{  G\left(  \alpha\right)
^{1/\alpha}\log G\left(  \alpha\right)  \right\}  +\alpha\text{Tr}\left\{
G\left(  \alpha\right)  ^{\left(  1-\alpha\right)  /\alpha}\frac{d}{d\alpha
}G\left(  \alpha\right)  \right\}  }{\text{Tr}\left\{  G\left(  \alpha\right)
^{1/\alpha}\right\}  }. \label{eq:limit-a-1-sibson}%
\end{equation}
Consider that%
\begin{align}
\lim_{\alpha\rightarrow1}G\left(  \alpha\right)  ^{\left(  1-\alpha\right)
/\alpha}  &  =\left[  \rho_{C}^{0}\text{Tr}_{A}\left\{  \rho_{AC}^{0}%
\rho_{ABC}\rho_{AC}^{0}\right\}  \rho_{C}^{0}\right]  ^{0}\\
&  =\rho_{BC}^{0}.
\end{align}
Evaluating the limits above one at a time and using that supp$\left(
\rho_{ABC}\right)  \subseteq\ $supp$\left(  \rho_{AC}\right)  \subseteq
\ $supp$\left(  \rho_{C}\right)  $ (see, e.g., \cite[Lemma~B.4.1]%
{RennerThesis}), we find that%
\begin{align}
\lim_{\alpha\rightarrow1}\frac{1}{\text{Tr}\left\{  G\left(  \alpha\right)
^{1/\alpha}\right\}  }  &  =\frac{1}{\text{Tr}\left\{  \rho_{C}^{0}%
\text{Tr}_{A}\left\{  \rho_{AC}^{0}\rho_{ABC}\rho_{AC}^{0}\right\}  \rho
_{C}^{0}\right\}  }\\
&  =1,\\
\lim_{\alpha\rightarrow1}-\text{Tr}\left\{  G\left(  \alpha\right)
^{1/\alpha}\log G\left(  \alpha\right)  \right\}   &  =-\text{Tr}\left\{
\left[  \rho_{C}^{0}\text{Tr}_{A}\left\{  \rho_{AC}^{0}\rho_{ABC}\rho_{AC}%
^{0}\right\}  \rho_{C}^{0}\right]  \log\left[  \rho_{C}^{0}\text{Tr}%
_{A}\left\{  \rho_{AC}^{0}\rho_{ABC}\rho_{AC}^{0}\right\}  \rho_{C}%
^{0}\right]  \right\} \\
&  =-\text{Tr}\left\{  \rho_{BC}\log\rho_{BC}\right\}  ,
\end{align}%
\begin{multline}
\lim_{\alpha\rightarrow1}\frac{d}{d\alpha}G\left(  \alpha\right)  =\frac{1}%
{2}\left(  \log\rho_{C}\right)  \rho_{C}^{0}\text{Tr}_{A}\left\{  \rho
_{AC}^{0}\rho_{ABC}\rho_{AC}^{0}\right\}  \rho_{C}^{0}-\frac{1}{2}\rho_{C}%
^{0}\text{Tr}_{A}\left\{  \left(  \log\rho_{AC}\right)  \rho_{AC}^{0}%
\rho_{ABC}\rho_{AC}^{0}\right\}  \rho_{C}^{0}\\
+\rho_{C}^{0}\text{Tr}_{A}\left\{  \rho_{AC}^{0}\left(  \log\rho_{ABC}\right)
\rho_{ABC}\rho_{AC}^{0}\right\}  \rho_{C}^{0}-\frac{1}{2}\rho_{C}^{0}%
\text{Tr}_{A}\left\{  \rho_{AC}^{0}\rho_{ABC}\left(  \log\rho_{AC}\right)
\rho_{AC}^{0}\right\}  \rho_{C}^{0}\\
+\frac{1}{2}\rho_{C}^{0}\text{Tr}_{A}\left\{  \rho_{AC}^{0}\rho_{ABC}\rho
_{AC}^{0}\right\}  \left(  \log\rho_{C}\right)  \rho_{C}^{0}.
\end{multline}
Putting all of this together, we can see that the limit in
(\ref{eq:limit-a-1-sibson}) evaluates to%
\begin{align}
\lim_{\alpha\rightarrow1}I_{\alpha}\left(  A;B|C\right)  _{\rho}  &
=\Delta\left(  \rho_{ABC},\rho_{AC},\rho_{C},\rho_{BC}\right) \\
&  =I\left(  A;B|C\right)  _{\rho}.
\end{align}

\section{Convergence of the $\widetilde{\Delta}_{\alpha}$ quantities}

\label{sec:sandwiched-renyi-a-to-1}This section presents a proof of
Theorem~\ref{thm:conv-vN-sandwiched}. We will consider L'H\^{o}pital's rule in
order to evaluate the limit of $\widetilde{\Delta}_{\alpha}$ as $\alpha
\rightarrow1$, due to the presence of the denominator term $\alpha-1$ in
$\widetilde{\Delta}_{\alpha}$. Consider that%
\begin{equation}
\widetilde{Q}_{\alpha}\left(  \rho_{ABC},\tau_{AC},\omega_{C},\theta
_{BC}\right)  =\text{Tr}\left\{  \left[  Z_{ABC}\left(  \alpha\right)
\right]  ^{\alpha}\right\}  ,
\end{equation}
where%
\begin{equation}
Z_{ABC}\left(  \alpha\right)  \equiv\rho_{ABC}^{1/2}\tau_{AC}^{\left(
1-\alpha\right)  /2\alpha}\omega_{C}^{\left(  \alpha-1\right)  /2\alpha}%
\theta_{BC}^{\left(  1-\alpha\right)  /\alpha}\omega_{C}^{\left(
\alpha-1\right)  /2\alpha}\tau_{AC}^{\left(  1-\alpha\right)  /2\alpha}%
\rho_{ABC}^{1/2}.
\end{equation}
We begin by computing%
\begin{multline}
\frac{d}{d\alpha}Z_{ABC}\left(  \alpha\right)  =\left(  -\frac{1}{\alpha^{2}%
}\right)  \Bigg[\frac{1}{2}\rho_{ABC}^{1/2}\left(  \log\tau_{AC}\right)
\tau_{AC}^{\left(  1-\alpha\right)  /2\alpha}\omega_{C}^{\left(
\alpha-1\right)  /2\alpha}\theta_{BC}^{\left(  1-\alpha\right)  /\alpha}%
\omega_{C}^{\left(  \alpha-1\right)  /2\alpha}\tau_{AC}^{\left(
1-\alpha\right)  /2\alpha}\rho_{ABC}^{1/2}\\
-\frac{1}{2}\rho_{ABC}^{1/2}\tau_{AC}^{\left(  1-\alpha\right)  /2\alpha
}\left(  \log\omega_{C}\right)  \omega_{C}^{\left(  \alpha-1\right)  /2\alpha
}\theta_{BC}^{\left(  1-\alpha\right)  /\alpha}\omega_{C}^{\left(
\alpha-1\right)  /2\alpha}\tau_{AC}^{\left(  1-\alpha\right)  /2\alpha}%
\rho_{ABC}^{1/2}\\
+\rho_{ABC}^{1/2}\tau_{AC}^{\left(  1-\alpha\right)  /2\alpha}\omega
_{C}^{\left(  \alpha-1\right)  /2\alpha}\left(  \log\theta_{BC}\right)
\theta_{BC}^{\left(  1-\alpha\right)  /\alpha}\omega_{C}^{\left(
\alpha-1\right)  /2\alpha}\tau_{AC}^{\left(  1-\alpha\right)  /2\alpha}%
\rho_{ABC}^{1/2}\\
-\frac{1}{2}\rho_{ABC}^{1/2}\tau_{AC}^{\left(  1-\alpha\right)  /2\alpha
}\omega_{C}^{\left(  \alpha-1\right)  /2\alpha}\theta_{BC}^{\left(
1-\alpha\right)  /\alpha}\left(  \log\omega_{C}\right)  \omega_{C}^{\left(
\alpha-1\right)  /2\alpha}\tau_{AC}^{\left(  1-\alpha\right)  /2\alpha}%
\rho_{ABC}^{1/2}\\
+\frac{1}{2}\rho_{ABC}^{1/2}\tau_{AC}^{\left(  1-\alpha\right)  /2\alpha
}\omega_{C}^{\left(  \alpha-1\right)  /2\alpha}\theta_{BC}^{\left(
1-\alpha\right)  /\alpha}\omega_{C}^{\left(  \alpha-1\right)  /2\alpha}\left(
\log\tau_{AC}\right)  \tau_{AC}^{\left(  1-\alpha\right)  /2\alpha}\rho
_{ABC}^{1/2}\Bigg].
\end{multline}
Applying Lemma~\ref{lem:matrix-der-lemma}\ to $Z_{ABC}\left(  \alpha\right)  $
and the function $\alpha$, we find that%
\begin{equation}
\frac{d}{d\alpha}\text{Tr}\left\{  \left[  Z_{ABC}\left(  \alpha\right)
\right]  ^{\alpha}\right\}  =\text{Tr}\left\{  \left[  Z_{ABC}\left(
\alpha\right)  \right]  ^{\alpha}\log Z_{ABC}\left(  \alpha\right)  \right\}
+\alpha\text{Tr}\left\{  \left[  Z_{ABC}\left(  \alpha\right)  \right]
^{\alpha-1}\frac{d}{d\alpha}Z_{ABC}\left(  \alpha\right)  \right\}  ,
\end{equation}
and%
\begin{align}
\lim_{\alpha\rightarrow1}\alpha\left[  Z_{ABC}\left(  \alpha\right)  \right]
^{\alpha-1}  &  =\left[  \rho_{ABC}^{1/2}\tau_{AC}^{0}\omega_{C}^{0}%
\theta_{BC}^{0}\omega_{C}^{0}\tau_{AC}^{0}\rho_{ABC}^{1/2}\right]  ^{0}\\
&  \equiv\left[  Z_{ABC}\left(  1\right)  \right]  ^{0},
\end{align}
we find that%
\begin{multline}
\lim_{\alpha\rightarrow1}\frac{d}{d\alpha}\widetilde{Q}_{\alpha}\left(
\rho_{ABC},\tau_{AC},\omega_{C},\theta_{BC}\right) \\
=\text{Tr}\left\{  \rho_{ABC}^{1/2}\tau_{AC}^{0}\omega_{C}^{0}\theta_{BC}%
^{0}\omega_{C}^{0}\tau_{AC}^{0}\rho_{ABC}^{1/2}\log\rho_{ABC}^{1/2}\tau
_{AC}^{0}\omega_{C}^{0}\theta_{BC}^{0}\omega_{C}^{0}\tau_{AC}^{0}\rho
_{ABC}^{1/2}\right\} \\
-\frac{1}{2}\text{Tr}\left\{  \left[  Z_{ABC}\left(  1\right)  \right]
^{0}\rho_{ABC}^{1/2}\left(  \log\tau_{AC}\right)  \tau_{AC}^{0}\omega_{C}%
^{0}\theta_{BC}^{0}\omega_{C}^{0}\tau_{AC}^{0}\rho_{ABC}^{1/2}\right\} \\
+\frac{1}{2}\text{Tr}\left\{  \left[  Z_{ABC}\left(  1\right)  \right]
^{0}\rho_{ABC}^{1/2}\tau_{AC}^{0}\left(  \log\omega_{C}\right)  \omega_{C}%
^{0}\theta_{BC}^{0}\omega_{C}^{0}\tau_{AC}^{0}\rho_{ABC}^{1/2}\right\} \\
-\text{Tr}\left\{  \left[  Z_{ABC}\left(  1\right)  \right]  ^{0}\rho
_{ABC}^{1/2}\tau_{AC}^{0}\omega_{C}^{0}\left(  \log\theta_{BC}\right)
\theta_{BC}^{0}\omega_{C}^{0}\tau_{AC}^{0}\rho_{ABC}^{1/2}\right\} \\
+\frac{1}{2}\text{Tr}\left\{  \left[  Z_{ABC}\left(  1\right)  \right]
^{0}\rho_{ABC}^{1/2}\tau_{AC}^{0}\omega_{C}^{0}\theta_{BC}^{0}\left(
\log\omega_{C}\right)  \omega_{C}^{0}\tau_{AC}^{0}\rho_{ABC}^{1/2}\right\} \\
-\frac{1}{2}\text{Tr}\left\{  \left[  Z_{ABC}\left(  1\right)  \right]
^{0}\rho_{ABC}^{1/2}\tau_{AC}^{0}\omega_{C}^{0}\theta_{BC}^{0}\omega_{C}%
^{0}\left(  \log\tau_{AC}\right)  \tau_{AC}^{0}\rho_{ABC}^{1/2}\right\}  .
\end{multline}
Since we assume that supp$\left(  \rho_{ABC}\right)  $ is contained in each of
supp$\left(  \tau_{AC}\right)  $, supp$\left(  \omega_{C}\right)  $, and
supp$\left(  \theta_{BC}\right)  $, we can see that%
\begin{align}
\lim_{\alpha\rightarrow1}\frac{d}{d\alpha}\widetilde{Q}_{\alpha}\left(
\rho_{ABC},\tau_{AC},\omega_{C},\theta_{BC}\right)   &  =\Delta\left(
\rho_{ABC},\tau_{AC},\omega_{C},\theta_{BC}\right)  ,\\
\lim_{\alpha\rightarrow1}\widetilde{Q}_{\alpha}\left(  \rho_{ABC},\tau
_{AC},\omega_{C},\theta_{BC}\right)   &  =1,
\end{align}
by applying the relations $\rho_{ABC}=\rho_{ABC}^{0}\rho_{ABC}\rho_{ABC}^{0}$,
$\rho_{ABC}^{0}\tau_{AC}^{0}=\rho_{ABC}^{0}$, $\rho_{ABC}^{0}\theta_{BC}%
^{0}=\rho_{ABC}^{0}$, $\rho_{ABC}^{0}\omega_{C}^{0}=\rho_{ABC}^{0}$, $\left[
Z_{ABC}\left(  1\right)  \right]  ^{0}=\rho_{ABC}^{0}$, and their Hermitian
conjugates. Applying L'H\^{o}pital's rule, we find that%
\begin{align}
\lim_{\alpha\rightarrow1}\widetilde{\Delta}_{\alpha}\left(  \rho_{ABC}%
,\tau_{AC},\omega_{C},\theta_{BC}\right)   &  =\lim_{\alpha\rightarrow1}%
\frac{\frac{d}{d\alpha}\widetilde{Q}_{\alpha}\left(  \rho_{ABC},\tau
_{AC},\omega_{C},\theta_{BC}\right)  }{\widetilde{Q}_{\alpha}\left(
\rho_{ABC},\tau_{AC},\omega_{C},\theta_{BC}\right)  }\\
&  =\Delta\left(  \rho_{ABC},\tau_{AC},\omega_{C},\theta_{BC}\right)  .
\end{align}
Essentially the same proof establishes the limiting relation for the other
$\widetilde{\Delta}_{\alpha}$ quantities defined from \eqref{eq:Q_SW_a_2}-\eqref{eq:Q_SW_a_6}.

\section{Convergence to $\Delta_{\max}$}

\label{sec:convergence-Delta_max}This section gives a proof of
Proposition~\ref{prop:convergence-Delta_max}. Let $\rho_{ABC}\in
\mathcal{S}\left(  \mathcal{H}_{ABC}\right)  _{++}$, $\tau_{AC}\in
\mathcal{S}\left(  \mathcal{H}_{AC}\right)  _{++}$, $\theta_{BC}\in
\mathcal{S}\left(  \mathcal{H}_{BC}\right)  _{++}$, and $\omega_{C}%
\in\mathcal{S}\left(  \mathcal{H}_{C}\right)  _{++}$. We prove that%
\begin{equation}
\lim_{\alpha\rightarrow\infty}\widetilde{\Delta}_{\alpha}\left(  \rho
_{ABC},\tau_{AC},\omega_{C},\theta_{BC}\right)  =\Delta_{\max}\left(
\rho_{ABC},\tau_{AC},\omega_{C},\theta_{BC}\right)  .
\end{equation}
The method of proof is the same as that for \cite[Theorem~5]{MDSFT13}. By the
reverse triangle inequality for the $\alpha$ norm, we have that%
\begin{multline}
\left\vert \ \ \left\Vert \rho_{ABC}^{\frac{1}{2}}\tau_{AC}^{\frac{1-\alpha
}{2\alpha}}\omega_{C}^{\frac{\alpha-1}{2\alpha}}\theta_{BC}^{\frac{1-\alpha
}{\alpha}}\omega_{C}^{\frac{\alpha-1}{2\alpha}}\tau_{AC}^{\frac{1-\alpha
}{2\alpha}}\rho_{ABC}^{\frac{1}{2}}\right\Vert _{\alpha}-\left\Vert \rho
_{ABC}^{\frac{1}{2}}\tau_{AC}^{-\frac{1}{2}}\omega_{C}^{\frac{1}{2}}%
\theta_{BC}^{-1}\omega_{C}^{\frac{1}{2}}\tau_{AC}^{-\frac{1}{2}}\rho
_{ABC}^{\frac{1}{2}}\right\Vert _{\alpha}\ \ \right\vert \\
\leq\left\Vert \rho_{ABC}^{\frac{1}{2}}\tau_{AC}^{\frac{1-\alpha}{2\alpha}%
}\omega_{C}^{\frac{\alpha-1}{2\alpha}}\theta_{BC}^{\frac{1-\alpha}{\alpha}%
}\omega_{C}^{\frac{\alpha-1}{2\alpha}}\tau_{AC}^{\frac{1-\alpha}{2\alpha}}%
\rho_{ABC}^{\frac{1}{2}}-\rho_{ABC}^{\frac{1}{2}}\tau_{AC}^{-\frac{1}{2}%
}\omega_{C}^{\frac{1}{2}}\theta_{BC}^{-1}\omega_{C}^{\frac{1}{2}}\tau
_{AC}^{-\frac{1}{2}}\rho_{ABC}^{\frac{1}{2}}\right\Vert _{\alpha}.
\end{multline}
Then%
\begin{align}
&  \lim_{\alpha\rightarrow\infty}\widetilde{\Delta}_{\alpha}\left(  \rho
_{ABC},\tau_{AC},\omega_{C},\theta_{BC}\right) \nonumber\\
&  =\lim_{\alpha\rightarrow\infty}\frac{\alpha}{\alpha-1}\log\left\Vert
\rho_{ABC}^{\frac{1}{2}}\tau_{AC}^{\frac{1-\alpha}{2\alpha}}\omega_{C}%
^{\frac{\alpha-1}{2\alpha}}\theta_{BC}^{\frac{1-\alpha}{\alpha}}\omega
_{C}^{\frac{\alpha-1}{2\alpha}}\tau_{AC}^{\frac{1-\alpha}{2\alpha}}\rho
_{ABC}^{\frac{1}{2}}\right\Vert _{\alpha}\\
&  \leq\log\left(
\begin{array}
[c]{c}%
\lim_{\alpha\rightarrow\infty}\left\Vert \rho_{ABC}^{\frac{1}{2}}\tau
_{AC}^{-\frac{1}{2}}\omega_{C}^{\frac{1}{2}}\theta_{BC}^{-1}\omega_{C}%
^{\frac{1}{2}}\tau_{AC}^{-\frac{1}{2}}\rho_{ABC}^{\frac{1}{2}}\right\Vert
_{\alpha}+\\
\lim_{\alpha\rightarrow\infty}\left\Vert \rho_{ABC}^{\frac{1}{2}}\tau
_{AC}^{\frac{1-\alpha}{2\alpha}}\omega_{C}^{\frac{\alpha-1}{2\alpha}}%
\theta_{BC}^{\frac{1-\alpha}{\alpha}}\omega_{C}^{\frac{\alpha-1}{2\alpha}}%
\tau_{AC}^{\frac{1-\alpha}{2\alpha}}\rho_{ABC}^{\frac{1}{2}}-\rho_{ABC}%
^{\frac{1}{2}}\tau_{AC}^{-\frac{1}{2}}\omega_{C}^{\frac{1}{2}}\theta_{BC}%
^{-1}\omega_{C}^{\frac{1}{2}}\tau_{AC}^{-\frac{1}{2}}\rho_{ABC}^{\frac{1}{2}%
}\right\Vert _{\alpha}%
\end{array}
\right)
\end{align}%
\begin{align}
&  \leq\log\left(
\begin{array}
[c]{c}%
\lim_{\alpha\rightarrow\infty}\left\Vert \rho_{ABC}^{\frac{1}{2}}\tau
_{AC}^{-\frac{1}{2}}\omega_{C}^{\frac{1}{2}}\theta_{BC}^{-1}\omega_{C}%
^{\frac{1}{2}}\tau_{AC}^{-\frac{1}{2}}\rho_{ABC}^{\frac{1}{2}}\right\Vert
_{\alpha}+\dim\left(  \mathcal{H}_{ABC}\right)  \times\\
\lim_{\alpha\rightarrow\infty}\left\Vert \rho_{ABC}^{\frac{1}{2}}\tau
_{AC}^{\frac{1-\alpha}{2\alpha}}\omega_{C}^{\frac{\alpha-1}{2\alpha}}%
\theta_{BC}^{\frac{1-\alpha}{\alpha}}\omega_{C}^{\frac{\alpha-1}{2\alpha}}%
\tau_{AC}^{\frac{1-\alpha}{2\alpha}}\rho_{ABC}^{\frac{1}{2}}-\rho_{ABC}%
^{\frac{1}{2}}\tau_{AC}^{-\frac{1}{2}}\omega_{C}^{\frac{1}{2}}\theta_{BC}%
^{-1}\omega_{C}^{\frac{1}{2}}\tau_{AC}^{-\frac{1}{2}}\rho_{ABC}^{\frac{1}{2}%
}\right\Vert _{\infty}%
\end{array}
\right) \\
&  =\log\left\Vert \rho_{ABC}^{\frac{1}{2}}\tau_{AC}^{-\frac{1}{2}}\omega
_{C}^{\frac{1}{2}}\theta_{BC}^{-1}\omega_{C}^{\frac{1}{2}}\tau_{AC}^{-\frac
{1}{2}}\rho_{ABC}^{\frac{1}{2}}\right\Vert _{\infty}\\
&  =\Delta_{\max}\left(  \rho_{ABC},\tau_{AC},\omega_{C},\theta_{BC}\right)
\end{align}
and%
\begin{align}
&  \lim_{\alpha\rightarrow\infty}\widetilde{\Delta}_{\alpha}\left(  \rho
_{ABC},\tau_{AC},\omega_{C},\theta_{BC}\right) \nonumber\\
&  =\lim_{\alpha\rightarrow\infty}\frac{\alpha}{\alpha-1}\log\left\Vert
\rho_{ABC}^{\frac{1}{2}}\tau_{AC}^{\frac{1-\alpha}{2\alpha}}\omega_{C}%
^{\frac{\alpha-1}{2\alpha}}\theta_{BC}^{\frac{1-\alpha}{\alpha}}\omega
_{C}^{\frac{\alpha-1}{2\alpha}}\tau_{AC}^{\frac{1-\alpha}{2\alpha}}\rho
_{ABC}^{\frac{1}{2}}\right\Vert _{\alpha}\\
&  \geq\log\left(
\begin{array}
[c]{c}%
\lim_{\alpha\rightarrow\infty}\left\Vert \rho_{ABC}^{\frac{1}{2}}\tau
_{AC}^{-\frac{1}{2}}\omega_{C}^{\frac{1}{2}}\theta_{BC}^{-1}\omega_{C}%
^{\frac{1}{2}}\tau_{AC}^{-\frac{1}{2}}\rho_{ABC}^{\frac{1}{2}}\right\Vert
_{\alpha}-\\
\lim_{\alpha\rightarrow\infty}\left\Vert \rho_{ABC}^{\frac{1}{2}}\tau
_{AC}^{\frac{1-\alpha}{2\alpha}}\omega_{C}^{\frac{\alpha-1}{2\alpha}}%
\theta_{BC}^{\frac{1-\alpha}{\alpha}}\omega_{C}^{\frac{\alpha-1}{2\alpha}}%
\tau_{AC}^{\frac{1-\alpha}{2\alpha}}\rho_{ABC}^{\frac{1}{2}}-\rho_{ABC}%
^{\frac{1}{2}}\tau_{AC}^{-\frac{1}{2}}\omega_{C}^{\frac{1}{2}}\theta_{BC}%
^{-1}\omega_{C}^{\frac{1}{2}}\tau_{AC}^{-\frac{1}{2}}\rho_{ABC}^{\frac{1}{2}%
}\right\Vert _{\alpha}%
\end{array}
\right)
\end{align}%
\begin{align}
&  \geq\log\left(
\begin{array}
[c]{c}%
\lim_{\alpha\rightarrow\infty}\left\Vert \rho_{ABC}^{\frac{1}{2}}\tau
_{AC}^{-\frac{1}{2}}\omega_{C}^{\frac{1}{2}}\theta_{BC}^{-1}\omega_{C}%
^{\frac{1}{2}}\tau_{AC}^{-\frac{1}{2}}\rho_{ABC}^{\frac{1}{2}}\right\Vert
_{\alpha}-\dim\left(  \mathcal{H}_{ABC}\right)  \times\\
\lim_{\alpha\rightarrow\infty}\left\Vert \rho_{ABC}^{\frac{1}{2}}\tau
_{AC}^{\frac{1-\alpha}{2\alpha}}\omega_{C}^{\frac{\alpha-1}{2\alpha}}%
\theta_{BC}^{\frac{1-\alpha}{\alpha}}\omega_{C}^{\frac{\alpha-1}{2\alpha}}%
\tau_{AC}^{\frac{1-\alpha}{2\alpha}}\rho_{ABC}^{\frac{1}{2}}-\rho_{ABC}%
^{\frac{1}{2}}\tau_{AC}^{-\frac{1}{2}}\omega_{C}^{\frac{1}{2}}\theta_{BC}%
^{-1}\omega_{C}^{\frac{1}{2}}\tau_{AC}^{-\frac{1}{2}}\rho_{ABC}^{\frac{1}{2}%
}\right\Vert _{\infty}%
\end{array}
\right) \\
&  =\log\left\Vert \rho_{ABC}^{\frac{1}{2}}\tau_{AC}^{-\frac{1}{2}}\omega
_{C}^{\frac{1}{2}}\theta_{BC}^{-1}\omega_{C}^{\frac{1}{2}}\tau_{AC}^{-\frac
{1}{2}}\rho_{ABC}^{\frac{1}{2}}\right\Vert _{\infty}\\
&  =\Delta_{\max}\left(  \rho_{ABC},\tau_{AC},\omega_{C},\theta_{BC}\right)  .
\end{align}

\section{Approaches for proving Conjecture~\ref{conj:monotone-alpha} and proof
for a special case}

\label{sec:proof-approach}This section gives more details regarding the
approach outlined in Section~\ref{sec:proof-approach-outline}\ for proving
Conjecture~\ref{conj:monotone-alpha}. Let $\rho_{ABC}\in\mathcal{S}\left(
\mathcal{H}_{ABC}\right)  _{++}$, $\tau_{AC}\in\mathcal{S}\left(
\mathcal{H}_{AC}\right)  _{++}$, $\theta_{BC}\in\mathcal{S}\left(
\mathcal{H}_{BC}\right)  _{++}$, and $\omega_{C}\in\mathcal{S}\left(
\mathcal{H}_{C}\right)  _{++}$. We begin by introducing a variable%
\begin{equation}
\gamma=\alpha-1,
\end{equation}
and with%
\begin{equation}
Y\left(  \gamma\right)  \equiv\rho_{ABC}^{1+\gamma}\tau_{AC}^{\frac{-\gamma
}{2}}\omega_{C}^{\frac{\gamma}{2}}\theta_{BC}^{-\gamma}\omega_{C}%
^{\frac{\gamma}{2}}\tau_{AC}^{\frac{-\gamma}{2}}, \label{eq:Y_beta}%
\end{equation}
it follows that $\Delta_{\alpha}\left(  \rho_{ABC},\tau_{AC},\omega_{C}%
,\theta_{BC}\right)  $ is equal to%
\begin{equation}
\frac{1}{\alpha-1}\log\text{Tr}\left\{  \rho_{ABC}^{\alpha}\tau_{AC}%
^{\frac{1-\alpha}{2}}\omega_{C}^{\frac{\alpha-1}{2}}\theta_{BC}^{1-\alpha
}\omega_{C}^{\frac{\alpha-1}{2}}\tau_{AC}^{\frac{1-\alpha}{2}}\right\}
=\frac{1}{\gamma}\log\text{Tr}\left\{  Y\left(  \gamma\right)  \right\}  .
\end{equation}
Since $d\gamma/d\alpha=1$,%
\begin{equation}
\frac{d}{d\alpha}\left[  \frac{1}{\alpha-1}\log\text{Tr}\left\{  \rho
_{ABC}^{\alpha}\tau_{AC}^{\frac{1-\alpha}{2}}\omega_{C}^{\frac{\alpha-1}{2}%
}\theta_{BC}^{1-\alpha}\omega_{C}^{\frac{\alpha-1}{2}}\tau_{AC}^{\frac
{1-\alpha}{2}}\right\}  \right]  =\frac{d}{d\gamma}\left[  \frac{1}{\gamma
}\log\text{Tr}\left\{  Y\left(  \gamma\right)  \right\}  \right]  .
\end{equation}
We can then explicitly compute the derivative:%
\begin{align}
\frac{d}{d\gamma}\left[  \frac{1}{\gamma}\log\text{Tr}\left\{  Y\left(
\gamma\right)  \right\}  \right]   &  =-\frac{1}{\gamma^{2}}\log
\text{Tr}\left\{  Y\left(  \gamma\right)  \right\}  +\frac{\text{Tr}\left\{
\frac{d}{d\gamma}Y\left(  \gamma\right)  \right\}  }{\gamma\text{Tr}\left\{
Y\left(  \gamma\right)  \right\}  }\\
&  =\frac{\gamma\text{Tr}\left\{  \frac{d}{d\gamma}Y\left(  \gamma\right)
\right\}  -\text{Tr}\left\{  Y\left(  \gamma\right)  \right\}  \log
\text{Tr}\left\{  Y\left(  \gamma\right)  \right\}  }{\gamma^{2}%
\text{Tr}\left\{  Y\left(  \gamma\right)  \right\}  }.
\label{eq:derivative-exp}%
\end{align}
So%
\begin{multline}
\gamma\frac{d}{d\gamma}Y\left(  \gamma\right)  =\log\rho_{ABC}^{\gamma
}Y\left(  \gamma\right)  +\rho_{ABC}^{1+\gamma}\left[  \log\tau_{AC}%
^{-\gamma/2}\right]  \tau_{AC}^{\frac{-\gamma}{2}}\omega_{C}^{\frac{\gamma}%
{2}}\theta_{BC}^{-\gamma}\omega_{C}^{\frac{\gamma}{2}}\tau_{AC}^{\frac
{-\gamma}{2}}\\
+\rho_{ABC}^{1+\gamma}\tau_{AC}^{\frac{-\gamma}{2}}\left[  \log\omega
_{C}^{\gamma/2}\right]  \omega_{C}^{\frac{\gamma}{2}}\theta_{BC}^{-\gamma
}\omega_{C}^{\frac{\gamma}{2}}\tau_{AC}^{\frac{-\gamma}{2}}+\rho
_{ABC}^{1+\gamma}\tau_{AC}^{\frac{-\gamma}{2}}\omega_{C}^{\frac{\gamma}{2}%
}\left[  \log\theta_{BC}^{-\gamma}\right]  \theta_{BC}^{-\gamma}\omega
_{C}^{\frac{\gamma}{2}}\tau_{AC}^{\frac{-\gamma}{2}}\\
+\rho_{ABC}^{1+\gamma}\tau_{AC}^{\frac{-\gamma}{2}}\omega_{C}^{\frac{\gamma
}{2}}\theta_{BC}^{-\gamma}\omega_{C}^{\frac{\gamma}{2}}\left[  \log\omega
_{C}^{\gamma/2}\right]  \tau_{AC}^{\frac{-\gamma}{2}}+Y\left(  \gamma\right)
\log\tau_{AC}^{-\gamma/2}.
\end{multline}
If it is true that the numerator in (\ref{eq:derivative-exp}) is non-negative
for all $\rho_{ABC}$, then we can conclude the monotonicity in $\alpha$.

A potential path for proving the conjecture for the sandwiched version is to
follow a similar approach developed by Tomamichel \textit{et al}.~(see the
proof of \cite[Theorem~7]{MDSFT13}). Since we can write%
\begin{equation}
\widetilde{\Delta}_{\alpha}\left(  \rho_{ABC},\tau_{AC},\omega_{C},\theta
_{BC}\right)  =\max_{\gamma_{ABC}}\widetilde{D}_{\alpha}\left(  \rho
,\tau,\omega,\theta,\gamma\right)  ,
\end{equation}
where%
\begin{equation}
\widetilde{D}_{\alpha}\left(  \rho,\tau,\omega,\theta,\mu\right)  \equiv
\frac{\alpha}{\alpha-1}\log\text{Tr}\left\{  \rho_{ABC}^{1/2}\tau
_{AC}^{\left(  1-\alpha\right)  /2\alpha}\omega_{C}^{\left(  \alpha-1\right)
/2\alpha}\theta_{BC}^{\left(  1-\alpha\right)  /\alpha}\omega_{C}^{\left(
\alpha-1\right)  /2\alpha}\tau_{AC}^{\left(  1-\alpha\right)  /2\alpha}%
\rho_{ABC}^{1/2}\mu_{ABC}^{\left(  \alpha-1\right)  /\alpha}\right\}  ,
\label{eq:conjectured-monotone-1}%
\end{equation}
it suffices to prove that $\widetilde{D}_{\alpha}\left(  \rho,\tau
,\omega,\theta,\mu\right)  $ is monotone in $\alpha$. For this purpose, the
idea is similar to the above (i.e., try to show that the derivative of
$\widetilde{D}_{\alpha}\left(  \rho,\tau,\omega,\theta,\mu\right)  $ with
respect to $\alpha$ is non-negative). To this end, now let%
\begin{equation}
\gamma=\frac{\alpha-1}{\alpha},
\end{equation}
and with%
\begin{equation}
Z\left(  \gamma\right)  \equiv\rho_{ABC}^{1/2}\tau_{AC}^{\frac{-\gamma}{2}%
}\omega_{C}^{\frac{\gamma}{2}}\theta_{BC}^{-\gamma}\omega_{C}^{\frac{\gamma
}{2}}\tau_{AC}^{\frac{-\gamma}{2}}\rho_{ABC}^{1/2}\mu_{ABC}^{\gamma},
\end{equation}
it follows that (\ref{eq:conjectured-monotone-1}) is equal to%
\begin{equation}
\widetilde{D}_{\alpha}\left(  \rho,\tau,\omega,\theta,\mu\right)  =\frac
{1}{\gamma}\log\text{Tr}\left\{  Z\left(  \gamma\right)  \right\}  .
\end{equation}
Then since $d\gamma/d\alpha=1/\alpha^{2}$,%
\begin{equation}
\frac{d}{d\alpha}\left[  \widetilde{D}_{\alpha}\left(  \rho,\tau,\omega
,\theta,\mu\right)  \right]  =\frac{1}{\alpha^{2}}\frac{d}{d\gamma}\left[
\frac{1}{\gamma}\log\text{Tr}\left\{  Z\left(  \gamma\right)  \right\}
\right]  .
\end{equation}
Computing the derivative then results in%
\begin{align}
\frac{d}{d\gamma}\left[  \frac{1}{\gamma}\log\text{Tr}\left\{  Z\left(
\gamma\right)  \right\}  \right]   &  =-\frac{1}{\gamma^{2}}\log
\text{Tr}\left\{  Z\left(  \gamma\right)  \right\}  +\frac{\text{Tr}\left\{
\frac{d}{d\gamma}Z\left(  \gamma\right)  \right\}  }{\gamma\text{Tr}\left\{
Z\left(  \gamma\right)  \right\}  }\\
&  =\frac{\gamma\text{Tr}\left\{  \frac{d}{d\gamma}Z\left(  \gamma\right)
\right\}  -\text{Tr}\left\{  Z\left(  \gamma\right)  \right\}  \log
\text{Tr}\left\{  Z\left(  \gamma\right)  \right\}  }{\gamma^{2}%
\text{Tr}\left\{  Z\left(  \gamma\right)  \right\}  }.
\label{eq:derivative-exp-1}%
\end{align}
The calculation of the derivative $\gamma$Tr$\left\{  \frac{d}{d\gamma
}Z\left(  \gamma\right)  \right\}  $ is very similar to what we have shown
above. So, in order to prove the conjecture, it suffices to prove that the
numerator of the last line above is non-negative.

If the above approach is successful, one could take essentially the same
approach to prove all of the other conjectured monotonicities detailed in
Conjecture~\ref{conj:monotone-alpha}.

\subsection{Proof of Conjecture~\ref{conj:monotone-alpha} for $\alpha$ in a
neighborhood of one}

\label{sec:proof-neighborhood}We can prove that the numerator of
(\ref{eq:derivative-exp}) is non-negative for $\gamma$ in a neighborhood of
zero. To this end, consider a Taylor expansion of $Y\left(  \gamma\right)  $
in (\ref{eq:Y_beta})\ around $\gamma$ equal to zero (so around $\alpha$ equal
to one). Indeed, consider that%
\begin{align}
X^{1+\gamma}  &  =X+\gamma X\log X+\frac{\gamma^{2}}{2}X\log^{2}X+O\left(
\gamma^{3}\right)  ,\\
X^{\gamma}  &  =I+\gamma\log X+\frac{\gamma^{2}}{2}\log^{2}X+O\left(
\gamma^{3}\right)  .
\end{align}
For our case, we make the following substitutions into Tr$\left\{  Y\left(
\gamma\right)  \right\}  $:%
\begin{align}
\rho_{ABC}^{1+\gamma}  &  =\rho_{ABC}+\gamma\rho_{ABC}\log\rho_{ABC}%
+\frac{\gamma^{2}}{2}\rho_{ABC}\log^{2}\rho_{ABC}+O\left(  \gamma^{3}\right)
,\\
\theta_{BC}^{\frac{-\gamma}{2}}  &  =I-\frac{\gamma}{2}\log\theta_{BC}%
+\frac{\gamma^{2}}{8}\log^{2}\theta_{BC}+O\left(  \gamma^{3}\right)  ,\\
\omega_{C}^{\frac{\gamma}{2}}  &  =I+\frac{\gamma}{2}\log\omega_{C}%
+\frac{\gamma^{2}}{8}\log^{2}\omega_{C}+O\left(  \gamma^{3}\right)  ,\\
\tau_{AC}^{-\gamma}  &  =I-\gamma\log\tau_{AC}+\frac{\gamma^{2}}{2}\log
^{2}\tau_{AC}+O\left(  \gamma^{3}\right)  .
\end{align}
After a rather tedious calculation, we find that%
\begin{equation}
\text{Tr}\left\{  Y\left(  \gamma\right)  \right\}  =\text{Tr}\left\{
\rho_{ABC}\right\}  +\gamma\Delta\left(  \rho,\tau,\omega,\theta\right)
+\frac{\gamma^{2}}{2}\left[  V\left(  \rho,\tau,\omega,\theta\right)  +\left[
\Delta\left(  \rho,\tau,\omega,\theta\right)  \right]  ^{2}\right]  +O\left(
\gamma^{3}\right)  ,
\end{equation}
where $V\left(  \rho,\tau,\omega,\theta\right)  $ is a quantity for which it
seems natural to call the \textit{tripartite information variance}:
\begin{equation}
V\left(  \rho,\tau,\omega,\theta\right)  \equiv\text{Tr}\left\{  \rho
_{ABC}\left[  \log\rho_{ABC}-\log\tau_{AC}-\log\theta_{BC}+\log\omega
_{C}-\Delta\left(  \rho,\tau,\omega,\theta\right)  \right]  ^{2}\right\}  .
\end{equation}
A special case of this is a quantity which we can call the \textit{conditional
mutual information variance} of $\rho_{ABC}$:%
\begin{equation}
V\left(  A;B|C\right)  _{\rho}\equiv\text{Tr}\left\{  \rho_{ABC}\left[
\log\rho_{ABC}-\log\rho_{AC}-\log\rho_{BC}+\log\rho_{C}-I\left(  A;B|C\right)
_{\rho}\right]  ^{2}\right\}  .
\end{equation}
The mutual information variance defined in \cite{TT13}\ is a special case of
the above quantity when $C$ is trivial. For any Hermitian operator $H$, we
have that%
\begin{equation}
\left\langle H^{2}\right\rangle _{\rho}-\left\langle H\right\rangle _{\rho
}^{2}\geq0.
\end{equation}
So taking $H\equiv\log\rho_{ABC}-\log\tau_{AC}-\log\theta_{BC}+\log\omega_{C}%
$, we conclude that $V\left(  \rho,\tau,\omega,\theta\right)  \geq0$, an
observation central to our development here. We will make the abbreviations
$\Delta\equiv\Delta\left(  \rho,\tau,\omega,\theta\right)  $ and $V\equiv
V\left(  \rho,\tau,\omega,\theta\right)  $\ from here forward, so that%
\begin{equation}
\text{Tr}\left\{  Y\left(  \gamma\right)  \right\}  =1+\gamma\Delta
+\frac{\gamma^{2}}{2}\left[  V+\Delta^{2}\right]  +O\left(  \gamma^{3}\right)
. \label{eq:beta-small-proof-1}%
\end{equation}
So this implies that%
\begin{align}
\gamma\text{Tr}\left\{  \frac{d}{d\gamma}Y\left(  \gamma\right)  \right\}   &
=\gamma\Delta+\gamma^{2}\left[  V+\Delta^{2}\right]  +O\left(  \gamma
^{3}\right)  ,\\
\text{Tr}\left\{  Y\left(  \gamma\right)  \right\}  \log\text{Tr}\left\{
Y\left(  \gamma\right)  \right\}   &  =\left[  1+\gamma\Delta+\frac{\gamma
^{2}}{2}\left[  V+\Delta^{2}\right]  +O\left(  \gamma^{3}\right)  \right]
\log\left[  1+\gamma\Delta+\frac{\gamma^{2}}{2}\left[  V+\Delta^{2}\right]
+O\left(  \gamma^{3}\right)  \right]  .
\end{align}
Then for small $\gamma$, we have the following Taylor expansion for the
logarithm:%
\begin{align}
\log\left[  1+\gamma\Delta+\frac{\gamma^{2}}{2}\left[  V+\Delta^{2}\right]
+O\left(  \gamma^{3}\right)  \right]   &  =\gamma\Delta+\frac{\gamma^{2}}%
{2}\left[  V+\Delta^{2}\right]  -\frac{\gamma^{2}\Delta^{2}}{2}+O\left(
\gamma^{3}\right) \\
&  =\gamma\Delta+\frac{\gamma^{2}}{2}V+O\left(  \gamma^{3}\right)  ,
\end{align}
which gives%
\begin{align}
\text{Tr}\left\{  Y\left(  \gamma\right)  \right\}  \log\text{Tr}\left\{
Y\left(  \gamma\right)  \right\}   &  =\left[  1+\gamma\Delta+\frac{\gamma
^{2}}{2}\left[  V+\Delta^{2}\right]  +O\left(  \gamma^{3}\right)  \right]
\left[  \gamma\Delta+\frac{\gamma^{2}}{2}V+O\left(  \gamma^{3}\right)  \right]
\nonumber\\
&  =\gamma\Delta+\frac{\gamma^{2}}{2}V+\gamma^{2}\Delta^{2}+O\left(
\gamma^{3}\right)  .
\end{align}
Finally, we can say that%
\begin{align}
\gamma\text{Tr}\left\{  \frac{d}{d\gamma}Y\left(  \gamma\right)  \right\}
-\text{Tr}\left\{  Y\left(  \gamma\right)  \right\}  \log\text{Tr}\left\{
Y\left(  \gamma\right)  \right\}   &  =\gamma\Delta+\gamma^{2}\left[
V+\Delta^{2}\right]  -\left[  \gamma\Delta+\frac{\gamma^{2}}{2}V+\gamma
^{2}\Delta^{2}\right]  +O\left(  \gamma^{3}\right) \nonumber\\
&  =\frac{\gamma^{2}}{2}V+O\left(  \gamma^{3}\right)  .
\label{eq:beta-small-proof-last}%
\end{align}
If $V>0$, we can conclude that as long as $\gamma$ is very near to zero, all
terms $O\left(  \gamma^{3}\right)  $ are negligible in comparison to
$\frac{\gamma^{2}}{2}V$, and the monotonicity holds in such a regime. A
development similar to the above one establishes the other variations of
(\ref{eq:conj-1}) for $\gamma$ in a neighborhood of zero. (Note that this
argument does not work if $V=0$.)

A similar kind of development shows that the conjecture in (\ref{eq:conj-3})
and its variations hold for $\gamma$ in a neighborhood of zero. We only sketch
the main idea since it is similar to the previous development. We first
observe that we can rewrite Tr$\left\{  Z\left(  \gamma\right)  \right\}  $ in
the following way:%
\begin{equation}
\text{Tr}\left\{  Z\left(  \gamma\right)  \right\}  =\left\langle
\varphi\right\vert \tau_{AC}^{\frac{-\gamma}{2}}\omega_{C}^{\frac{\gamma}{2}%
}\theta_{BC}^{-\gamma}\omega_{C}^{\frac{\gamma}{2}}\tau_{AC}^{\frac{-\gamma
}{2}}\rho_{ABC}^{1/2}\otimes\left(  \mu_{A^{\prime}B^{\prime}C^{\prime}}%
^{T}\right)  ^{\gamma}\left\vert \varphi\right\rangle ,
\end{equation}
where $A^{\prime}B^{\prime}C^{\prime}$ are some systems isomorphic to $ABC$
and%
\begin{equation}
\left\vert \varphi\right\rangle _{ABC,A^{\prime}B^{\prime}C^{\prime}}%
\equiv\rho_{ABC}^{1/2}\otimes I_{A^{\prime}B^{\prime}C^{\prime}}\left\vert
\Gamma\right\rangle _{ABC,A^{\prime}B^{\prime}C},
\end{equation}
with $\left\vert \Gamma\right\rangle $ the maximally entangled vector. Then a
Taylor expansion about $\gamma=0$ (another tedious calculation)\ gives that%
\begin{equation}
\text{Tr}\left\{  Z\left(  \gamma\right)  \right\}  =\text{Tr}\left\{
\rho_{ABC}\right\}  +\gamma\left\langle \varphi\right\vert H_{ABC,A^{\prime
}B^{\prime}C^{\prime}}\left\vert \varphi\right\rangle +\frac{\gamma^{2}}%
{2}\left\langle \varphi\right\vert H_{ABC,A^{\prime}B^{\prime}C^{\prime}}%
^{2}\left\vert \varphi\right\rangle +O\left(  \gamma^{3}\right)  ,
\end{equation}
where%
\begin{equation}
H_{ABC,A^{\prime}B^{\prime}C^{\prime}}\equiv\log\omega_{C}-\log\tau_{AC}%
-\log\theta_{BC}+\log\mu_{A^{\prime}B^{\prime}C^{\prime}}^{T}.
\end{equation}
Then we know that%
\begin{equation}
\left\langle \varphi\right\vert H_{ABC,A^{\prime}B^{\prime}C^{\prime}}%
^{2}\left\vert \varphi\right\rangle -\left[  \left\langle \varphi\right\vert
H_{ABC,A^{\prime}B^{\prime}C^{\prime}}\left\vert \varphi\right\rangle \right]
^{2}\geq0. \label{eq:important-inequality-for-proof}%
\end{equation}
From here, we can show that the numerator of (\ref{eq:derivative-exp-1}) is
non-negative for small $\gamma$ by following the same development as in
(\ref{eq:beta-small-proof-1})-(\ref{eq:beta-small-proof-last}) (substitute
$\left\langle \varphi\right\vert H_{ABC,A^{\prime}B^{\prime}C^{\prime}%
}\left\vert \varphi\right\rangle $ for $\Delta$ and the LHS in
(\ref{eq:important-inequality-for-proof}) for $V$). The development for the
other variations of (\ref{eq:conj-3}) is similar.

\section{Dimension bounds and other inequalities}

\label{sec:dim-bounds} For the bounds in this appendix, we make the following
definitions:%
\begin{align}
I_{\alpha}^{\prime\prime}\left(  A;B|C\right)  _{\rho}  &  \equiv\inf
_{\sigma_{ABC}}\Delta_{\alpha}\left(  \rho_{ABC},\sigma_{AC},\sigma_{C}%
,\sigma_{BC}\right)  ,\\
\widetilde{I}_{\alpha}^{\prime\prime}\left(  A;B|C\right)  _{\rho}  &
\equiv\inf_{\sigma_{ABC}}\widetilde{\Delta}_{\alpha}\left(  \rho_{ABC}%
,\sigma_{AC},\sigma_{C},\sigma_{BC}\right)  ,
\end{align}
where the optimizations are over $\sigma_{ABC}$ such that supp$\left(
\rho_{ABC}\right)  \subseteq\ $supp$\left(  \sigma_{ABC}\right)  $.

\begin{proposition}
Let $\rho_{ABC}\in\mathcal{S}\left(  \mathcal{H}_{ABC}\right)  $. The
following dimension bound holds for $\alpha\in\lbrack0,1)\cup(1,2]$:%
\begin{equation}
I_{\alpha}^{\prime\prime}\left(  A;B|C\right)  _{\rho}\leq2\min\left\{  \log
d_{A},\log d_{B}\right\}  , \label{eq:dimension-bound-a}%
\end{equation}
and the following one holds for $\alpha\in(1/2,1)\cup(1,\infty)$:%
\begin{equation}
\widetilde{I}_{\alpha}^{\prime\prime}\left(  A;B|C\right)  _{\rho}\leq
2\min\left\{  \log d_{A},\log d_{B}\right\}  . \label{eq:dimension-bound-b}%
\end{equation}

\end{proposition}

\begin{proof}
We first prove that the following dimension bounds hold%
\begin{align}
I_{\alpha}^{\prime\prime}\left(  A;B|C\right)  _{\rho}  &  \leq\log
d_{A}-H_{\alpha}\left(  A|BC\right)  _{\rho},\label{eq:dimension-bound-1}\\
I_{\alpha}^{\prime\prime}\left(  A;B|C\right)  _{\rho}  &  \leq\log
d_{B}-H_{\alpha}\left(  B|AC\right)  _{\rho},\label{eq:dimension-bound-2}\\
\widetilde{I}_{\alpha}^{\prime\prime}\left(  A;B|C\right)  _{\rho}  &
\leq\log d_{A}-\widetilde{H}_{\alpha}\left(  A|BC\right)  _{\rho},\\
\widetilde{I}_{\alpha}^{\prime\prime}\left(  A;B|C\right)  _{\rho}  &
\leq\log d_{B}-\widetilde{H}_{\alpha}\left(  B|AC\right)  _{\rho}.
\end{align}
The inequality in (\ref{eq:dimension-bound-2}) follows from%
\begin{align}
I_{\alpha}^{\prime\prime}\left(  A;B|C\right)  _{\rho}  &  =\inf_{\sigma
_{ABC}}\frac{1}{\alpha-1}\log\text{Tr}\left\{  \rho_{ABC}^{\alpha}\sigma
_{AC}^{\left(  1-\alpha\right)  /2}\sigma_{C}^{\left(  \alpha-1\right)
/2}\sigma_{BC}^{1-\alpha}\sigma_{C}^{\left(  \alpha-1\right)  /2}\sigma
_{AC}^{\left(  1-\alpha\right)  /2}\right\} \\
&  \leq\inf_{\sigma_{AC}}\frac{1}{\alpha-1}\log\text{Tr}\left\{  \rho
_{ABC}^{\alpha}\sigma_{AC}^{\left(  1-\alpha\right)  /2}\sigma_{C}^{\left(
\alpha-1\right)  /2}\left(  \pi_{B}\otimes\sigma_{C}\right)  ^{1-\alpha}%
\sigma_{C}^{\left(  \alpha-1\right)  /2}\sigma_{AC}^{\left(  1-\alpha\right)
/2}\right\} \\
&  =\inf_{\sigma_{AC}}\frac{1}{\alpha-1}\log\text{Tr}\left\{  \rho
_{ABC}^{\alpha}\pi_{B}^{1-\alpha}\sigma_{AC}^{1-\alpha}\right\} \\
&  =\log d_{B}-\left(  -\min_{\sigma_{AC}}\frac{1}{\alpha-1}\log
\text{Tr}\left\{  \rho_{ABC}^{\alpha}\sigma_{AC}^{1-\alpha}\right\}  \right)
\\
&  =\log d_{B}-H_{\alpha}\left(  B|AC\right)  _{\rho},
\end{align}
where $\pi_{B}=I_{B}/d_{B}$ and $H_{\alpha}\left(  B|AC\right)  _{\rho}%
\equiv-\inf_{\sigma_{AC}}D_{\alpha}(\rho_{ABC}\Vert I_{B}\otimes\sigma_{AC})$
as in~\eqref{eq:cond_alpha}. The bound in (\ref{eq:dimension-bound-1}) follows
similarly by choosing $\sigma_{ABC}=\pi_{A}\otimes\sigma_{BC}$. The proofs for
the sandwiched R\'{e}nyi CMIs follow similarly, except we end up with the
sandwiched R\'{e}nyi conditional entropy in the upper bound.

To prove (\ref{eq:dimension-bound-a}), we use the duality relation proved in
\cite[Lemma~6]{TCR09}. From (\ref{eq:dimension-bound-1}), we know that%
\begin{align}
I_{\alpha}^{\prime\prime}\left(  A;B|C\right)  _{\rho}  &  \leq\log d_{A}%
+\inf_{\sigma_{BC}}\frac{1}{\alpha-1}\log\text{Tr}\left\{  \rho_{ABC}^{\alpha
}\sigma_{BC}^{1-\alpha}\right\} \\
&  \leq\log d_{A}+\frac{1}{\alpha-1}\log\text{Tr}\left\{  \rho_{ABC}^{\alpha
}\rho_{BC}^{1-\alpha}\right\} \\
&  \equiv\log d_{A}-H_{\alpha}\left(  A|BC\right)  _{\rho|\rho}\\
&  =\log d_{A}+H_{\beta}\left(  A|D\right)  _{\rho|\rho}\\
&  \leq\log d_{A}+H_{\beta}\left(  A\right)  _{\rho}\\
&  \leq2\log d_{A},
\end{align}
where $H_{\alpha}\left(  A|BC\right)  _{\rho|\rho}\equiv-D_{\alpha}(\rho
_{ABC}\Vert I_{A}\otimes\rho_{BC})$. The second equality follows from the
duality from \cite[Lemma~6]{TCR09}, i.e.,%
\begin{equation}
H_{\alpha}\left(  A|BC\right)  _{\rho|\rho}=-H_{\beta}\left(  A|D\right)
_{\rho|\rho},
\end{equation}
where $\rho_{ABCD}$ is a purification of $\rho_{ABC}$ and $\beta$ is chosen so
that $\alpha+\beta=2$. The third inequality follows from data processing and
the last from a dimension bound on the R\'{e}nyi entropy.

The inequality (\ref{eq:dimension-bound-b}) follows from the duality of the
sandwiched conditional R\'{e}nyi entropy \cite[Theorem~10]{MDSFT13}:%
\begin{equation}
\widetilde{H}_{\alpha}\left(  A|BC\right)  _{\rho}=-\widetilde{H}_{\beta
}\left(  A|D\right)  _{\rho},
\end{equation}
where $\widetilde{H}_{\alpha}\left(  A|BC\right)  _{\rho}\equiv-\inf
_{\sigma_{BC}}\widetilde{D}_{\alpha}(\rho_{ABC}\Vert I_{A}\otimes\sigma_{BC}%
)$, $\rho_{ABCD}$ is a purification of $\rho_{ABC}$ and $\beta$ is chosen so
that $\frac{1}{\alpha}+\frac{1}{\beta}=2$. So this means that%
\begin{align}
\widetilde{I}_{\alpha}^{\prime\prime}\left(  A;B|C\right)  _{\rho}  &
\leq\log d_{A}-\widetilde{H}_{\alpha}\left(  A|BC\right)  _{\rho}\\
&  =\log d_{A}+\widetilde{H}_{\beta}\left(  A|D\right)  _{\rho}\\
&  \leq\log d_{A}+\widetilde{H}_{\beta}\left(  A\right)  _{\rho}\\
&  \leq2\log d_{A},
\end{align}
where the second inequality follows from data processing and the last is a
universal bound on the R\'{e}nyi entropy.
\end{proof}

\begin{proposition}
Let $\rho_{ABC}\in\mathcal{S}\left(  \mathcal{H}_{ABC}\right)  $. The
following bounds hold for $\alpha\in\lbrack0,1)\cup(1,\infty)$:%
\begin{align}
I_{\alpha}^{\prime\prime}\left(  A;B|C\right)  _{\rho}  &  \leq I_{\alpha
}\left(  A;BC\right)  _{\rho},\\
I_{\alpha}^{\prime\prime}\left(  A;B|C\right)  _{\rho}  &  \leq I_{\alpha
}\left(  B;AC\right)  _{\rho},
\end{align}
and the following hold for $\alpha\in(0,1)\cup(1,\infty)$:%
\begin{align}
\widetilde{I}_{\alpha}^{\prime\prime}\left(  A;B|C\right)  _{\rho}  &
\leq\widetilde{I}_{\alpha}\left(  A;BC\right)  _{\rho},\\
\widetilde{I}_{\alpha}^{\prime\prime}\left(  A;B|C\right)  _{\rho}  &
\leq\widetilde{I}_{\alpha}\left(  B;AC\right)  _{\rho}.
\end{align}

\end{proposition}

\begin{proof}
A proof for the first inequality follows from%
\begin{align}
I_{\alpha}^{\prime\prime}\left(  A;B|C\right)  _{\rho}  &  =\inf_{\sigma
_{ABC}}\frac{1}{\alpha-1}\log\text{Tr}\left\{  \rho_{ABC}^{\alpha}\sigma
_{AC}^{\left(  1-\alpha\right)  /2}\sigma_{C}^{\left(  \alpha-1\right)
/2}\sigma_{BC}^{1-\alpha}\sigma_{C}^{\left(  \alpha-1\right)  /2}\sigma
_{AC}^{\left(  1-\alpha\right)  /2}\right\} \\
&  \leq\inf_{\sigma_{BC}}\frac{1}{\alpha-1}\log\text{Tr}\left\{  \rho
_{ABC}^{\alpha}\left(  \rho_{A}\otimes\sigma_{C}\right)  ^{\left(
1-\alpha\right)  /2}\sigma_{C}^{\left(  \alpha-1\right)  /2}\sigma
_{BC}^{1-\alpha}\sigma_{C}^{\left(  \alpha-1\right)  /2}\left(  \rho
_{A}\otimes\sigma_{C}\right)  ^{\left(  1-\alpha\right)  /2}\right\} \\
&  =\inf_{\sigma_{BC}}\frac{1}{\alpha-1}\log\text{Tr}\left\{  \rho
_{ABC}^{\alpha}\left(  \rho_{A}^{1-\alpha}\otimes\sigma_{BC}^{1-\alpha
}\right)  \right\} \\
&  \equiv I_{\alpha}\left(  A;BC\right)  _{\rho},
\end{align}
as defined in~\eqref{eq:alpha_mutual}. A proof for the second inequality
follows similarly by choosing $\sigma_{ABC}=\rho_{B}\otimes\sigma_{AC}$.
Proofs for the last two inequalities are similar, except the sandwiched
R\'{e}nyi mutual information is defined for a bipartite state $\rho_{AB}$ as%
\begin{equation}
\widetilde{I}_{\alpha}\left(  A;B\right)  _{\rho}\equiv\inf_{\sigma_{B}}%
\frac{1}{\alpha-1}\log\text{Tr}\left\{  \left[  \left(  \rho_{A}\otimes
\sigma_{B}\right)  ^{\frac{1-\alpha}{2\alpha}}\rho_{AB}\left(  \rho_{A}%
\otimes\sigma_{B}\right)  ^{\frac{1-\alpha}{2\alpha}}\right]  ^{\alpha
}\right\}  .
\end{equation}

\end{proof}

\bibliographystyle{plain}
\bibliography{Ref}

\end{document}